%% file: wct.tex
\pdfoutput=1
\documentclass[11pt]{article}
\usepackage{graphicx}
\usepackage{color}
\usepackage{subfigure}
\usepackage{amsmath,amssymb,amsthm}
\usepackage[margin=1in]{geometry}
\usepackage[small,bf]{caption}
\usepackage{paralist}
\usepackage{comment}
\usepackage[colorlinks,linkcolor=black,bookmarksopen,bookmarksnumbered,
citecolor=black,urlcolor=black]{hyperref}

\allowdisplaybreaks

\graphicspath{{figs/wct/examples/},{figs/wct/meshes/},{figs/wct/misc/}}

\newtheorem{theorem}{Theorem}
\newtheorem{proposition}[theorem]{Proposition}
\newtheorem{lemma}[theorem]{Lemma}
\newtheorem{corollary}[theorem]{Corollary}

\theoremstyle{definition}
\newtheorem{remark}[theorem]{Remark}
\newtheorem{example}[theorem]{Example}
\newtheorem*{definition}{Definition}

\newcommand{\RR}{\ensuremath{\mathbb{R}}}

\newcommand{\closure}[1]{\ensuremath{\mathrm{Cl}(#1)}}
\newcommand{\affine}{\ensuremath{\mathrm{aff}}}
\newcommand{\boundary}[1]{\ensuremath{\mathrm{Bd}(#1)}}
\newcommand{\interior}{\ensuremath{\mathrm{Int}}}
\newcommand{\link}{\ensuremath{\mathrm{Lk}}}
\newcommand{\eqB}{\ensuremath{B}}

\newcommand{\cone}[2]{\ensuremath{#1 * #2}}
\renewcommand{\star}{\ensuremath{\mathrm{St}}}

\newlength{\dotlen}

\definecolor{gry50}{rgb}{0.5, 0.5, 0.5}
\newcommand{\gray}[1]{\textcolor{gry50}{#1}}

\begin{document}
\title{Geometric and Combinatorial Properties
  of Well-Centered Triangulations in Three and Higher Dimensions}
\author{Evan VanderZee \and Anil N. Hirani \and Damrong Guoy \and
  Vadim Zharnitsky \and Edgar Ramos}

\date{}

\maketitle

\abstract{An $n$-simplex is said to be $n$-well-centered
if its circumcenter lies in its interior.  We introduce
several other geometric conditions and an
algebraic condition that can be used to determine
whether a simplex is $n$-well-centered.  These
conditions, together with some other observations, are
used to describe restrictions on the local combinatorial
structure of simplicial meshes in which every simplex
is well-centered.  In particular, it is shown that in
a $3$-well-centered ($2$-well-centered) tetrahedral mesh
there are at least $7$ ($9$) edges incident to each
interior vertex, and these bounds are sharp.
Moreover, it is shown that,
in stark contrast to the $2$-dimensional
analog, where there are exactly two vertex links that
prevent a well-centered triangle mesh in~$\RR^{2}$,
there are infinitely many vertex links that prohibit
a well-centered tetrahedral mesh in~$\RR^{3}$.}

\section{Introduction}

An $n$-dimensional simplex is {\emph{$n$-well-centered}} if its
circumcenter lies in its interior. More generally, it is
{\emph{$k$-well-centered}} if each of its $k$-dimensional faces is
$k$-well-centered.  It is {\emph{completely well-centered}} if it is
$k$-well-centered for each $k$, $1 \le k \le n$ \cite{VaHiGu2008}.
Several authors have noted the possible application of well-centered
meshes to particular problems~\cite{VaHiGu2008}.  Among these are
Nicolaides~\cite{Nicolaides1992} and
Sazonov~et~al.~\cite{SaHaMoWe2006,SaHaMoWe2007}, who have discussed
the covolume method and its application in electromagnetics
simulations.  Also, Kimmel and Sethian~\cite{KiSe1998} described an
algorithm for numerically solving the Eikonal equation on triangulated
domains.  Their algorithm, which can be used to compute geodesic paths
on triangulated surfaces, is described first for acute triangulations
(i.e., $2$-well-centered triangulations) and requires additional work
for triangulations that are nonacute.  \"Ung\"or and Sheffer
\cite{UnSh2002} used acute planar triangulations when they introduced
the tent-pitching algorithm for space-time meshing.  Well-centered
meshes also find application within Discrete Exterior Calculus (DEC),
a framework for designing numerical methods for partial differential
equations \cite{Hirani2003,DeHiLeMa2005}.  In DEC, a sufficient (but
not necessary) condition for discretizing the Hodge star operator so
that it is represented by a diagonal matrix is to use a well-centered
mesh. The diagonal matrix leads to efficient numerical solution.

Constructing well-centered meshes is a nontrivial task, so
researchers have put effort into finding ways to work
around an algorithmic requirement for well-centered meshes.
Recently, however, there has been progress towards making
well-centered meshes through mesh optimization.
In~\cite{VaHiGuRa2007} we described
a heuristic for obtaining well-centered meshes of planar domains
by starting with a given triangle mesh of the domain and
relocating the interior vertices of the mesh to minimize a cost
function defined on the coordinates of those vertices.  We later
generalized the cost function, making it possible to optimize
meshes of any dimension~\cite{VaHiGuRa2009}.
In~\cite{VaHiGu2008} we give a variety of examples of
tetrahedra that are and are not well-centered, and we show
that it is possible to mesh many simple shapes in~$\RR^3$ with
completely well-centered tetrahedra.  In many cases, the completely
well-centered tetrahedral meshes of~\cite{VaHiGu2008} were obtained
by creating initial relatively high-quality meshes by hand and
applying the optimization method to the initial meshes.

Not all simplicial meshes can be made well-centered by optimizing
the cost functions of~\cite{VaHiGuRa2007} and~\cite{VaHiGuRa2009}.
In some cases the combinatorial properties of the mesh
prevent the mesh from becoming well-centered.  In such
cases the mesh will not be well-centered for any choice of
vertex coordinates.  This paper develops theory and intuition
related to the geometry of well-centered simplices and applies
the geometric conditions to investigate combinatorial properties
of well-centered meshes.

Before discussing the specific results of this paper, a few comments
about terminology are in order.  The term {\emph{well-centered}} may
be used without a qualifying dimension to refer to the general
concept, or may refer to one of the more precise terms if the
context makes clear which more precise term is appropriate.  When
speaking of a triangle, for example, well-centered simultaneously
means $2$-well-centered and completely well-centered, since every
simplex is trivially $1$-well-centered.  All of these definitions
can be applied to simplicial meshes or simplicial complexes embedded
in Euclidean space.  Thus if a simplicial mesh is said to be
$k$-well-centered, this means that every $k$-dimensional simplex
appearing in the mesh properly contains its circumcenter.

\section{Results}

After giving definitions and notation in Sec.~\ref{sec:notation},
we introduce in Sec.~\ref{sec:characterize} several geometric
and algebraic conditions for an $n$-simplex to be well-centered.
Sections~\ref{sec:combinatorial3wccond}
and~\ref{sec:combinatorial2wctetcond} apply the theory
developed in Sec.~\ref{sec:characterize} to establish
conditions on the combinatorial structure of the
neighborhood of a vertex in a well-centered tetrahedral
mesh.  Finally, Sec.~\ref{sec:cubeapps} records some
observations specific to constructing well-centered meshes
of the cube, and Sec.~\ref{sec:conclusion}
offers some concluding thoughts.  We enumerate
the contributions of this paper in more
detail in the following paragraphs.

\begin{inparaenum}[(a)]
In Sec.~\ref{sec:characterize}
we prove three new conditions on when an $n$-simplex is
$n$-well-centered, each phrased in terms of the location
of a vertex $v_{i}$ given the facet of the simplex
opposite $v_{i}$.  These conditions are \item
a necessary condition expressed in terms of geometry
(Prop.~\ref{prop:char_nec}---the Cylinder Condition),
\item a sufficient condition expressed geometrically
(Prop.~\ref{prop:char_suf}---the Prism Condition),
and \item a both necessary and sufficient condition expressed
in terms of cubic polynomial inequalities
(Prop.~\ref{prop:necandsufpoly}).  The two geometric
conditions are generalizations to higher dimensions
of conditions in $\RR^{2}$.

Section~\ref{sec:combinatorial3wccond} investigates
combinatorial properties that follow from the results
in Sec.~\ref{sec:characterize}.
\item We prove a new combinatorial condition that must
be satisfied by the link of an interior vertex in an
$n$-well-centered mesh in $\RR^{n}$ (Thm.~\ref{thm:oneringcond}).
\item As an easy corollary we show that in a $3$-well-centered
tetrahedral mesh in $\RR^{3}$, every interior vertex
has at least seven incident edges (Cor.~\ref{cor:min3wc}).
\item We show that, in stark contrast to the analogous case
in $\RR^{2}$, where there are only two vertex links that
cannot appear in a $2$-well-centered mesh, there are
infinitely many vertex links that cannot appear in a
$3$-well-centered mesh in $\RR^{3}$ (Cor.~\ref{cor:infno3wc}).
\item We also construct an infinite family of vertex
links that can appear in a completely well-centered
tetrahedral mesh in~$\RR^{3}$. \item The section
closes by showing that if a vertex link can appear in
a $3$-well-centered mesh and the vertex link satisfies
some minor additional conditions, then degree three vertices
can be successively inserted into the vertex link to create
an infinite family of vertex links that can appear
in a $3$-well-centered mesh (Prop.~\ref{prop:deg3vtx3wc}).

Section~\ref{sec:combinatorial2wctetcond} develops combinatorial
conditions that $2$-well-centered tetrahedral meshes
in $\RR^{3}$ must satisfy. \item We prove in Thm.~\ref{thm:nminus3}
that no triangulation of $S^{2}$ on $m$ vertices
with a vertex of degree at least $m - 3$ can appear
in a $2$-well-centered tetrahedral mesh in $\RR^{3}$.
\item It follows that in a $2$-well-centered tetrahedral mesh
in $\RR^{3}$, every interior vertex has at least
nine incident edges (Cor.~\ref{cor:min2wc}).
\item We show that vertices of degree three can be inserted
into or deleted from triangulations that permit $2$-well-centered
neighborhoods to create other triangulations that
permit $2$-well-centered neighborhoods (Prop.~\ref{prop:deg3}).
\item Vertices of degree four can also be added to
such triangulations (Prop.~\ref{prop:deg4}).
\end{inparaenum}

At several points in the paper, we make claims about additional
results beyond what is actually proved in this paper.  Further
details about some of these claims appear in the first
author's dissertation~\cite{VanderZee2010}.

\section{Definitions and Notation}
\label{sec:notation}

We begin by introducing some definitions and notation that
will be used throughout the paper.  A simplex is referred
to with a Greek letter, usually $\sigma$ or $\tau$.  A superscript
for a simplex indicates the dimension, so, for example, $\sigma^{n}$
is an $n$-simplex.  The notation
$\sigma^{n} = [v_{0}v_{1}\ldots v_{n}]$ is used to
indicates that $\sigma^{n}$ is the
convex hull of the $n + 1$ vertices $v_{0}, v_{1}, \ldots, v_{n}$.
It is assumed that the vertices of a simplex are in general
position, i.e., that the vertices are affinely independent, so
$\sigma^{n}$ is fully $n$-dimensional.
The \emph{circumcenter} of a simplex $\sigma$ is denoted $c(\sigma)$.
For an $n$-simplex $\sigma^{n}$ embedded in $\RR^{n}$, $c(\sigma^{n})$
is the unique point which has the same distance
from every vertex of $\sigma^{n}$.
When $\sigma^{n}$ is embedded in $\RR^{m}$ for $m > n$,
$c(\sigma^{n})$ is the unique point that among all points equidistant
from the vertices of $\sigma^{n}$ minimizes
the distance to the vertices of $\sigma^{n}$.  The \emph{circumradius}
of a simplex $\sigma$, i.e., the distance from $c(\sigma)$ to the vertices
of $\sigma$, is denoted $R(\sigma)$.

We also use the {\emph{cone}}
operation of algebraic topology \cite{Munkres1984},
writing $\cone{u}{\sigma^n}$ to indicate the simplex formed
by taking the convex hull of a vertex $u$ together with
the $n$-dimensional simplex $\sigma^{n}$ to form a simplex
of dimension $n + 1$.  This notation may also be used
for a set $K$ of simplices;  $\cone{u}{K}$ is
the set of simplices $\{\cone{u}{\sigma} : \sigma \in K\}$.
The {\emph{affine hull}} of a set $S \subset \RR^{m}$,
which we denote $\affine(S)$, is the smallest affine
space that contains $S$.  For a simplex
$\sigma^{n} = [v_{0}\ldots v_{n}]$ we can define it as
\[
\affine([v_{0}\ldots v_{n}]) = \left\{\sum_{i=0}^{n} \lambda_{i}v_{i} :
    \sum_{i=0}^{n} \lambda_{i} = 1,\  -\infty < \lambda_{i} < \infty
    \quad\mathrm{for}\  i = 0,\ldots,n \right\}.
\]
The affine hull of a simplex $\sigma$
may also be called the {\emph{plane of $\sigma$}}.

When referring to a simplex $\sigma$, the {\emph{boundary}} of
the simplex, denoted $\boundary{\sigma}$,
is the union of the set of {\emph{proper
faces}} of $\sigma$, i.e., the set of all faces of $\sigma$ other than
$\sigma$ itself.  The {\emph{interior}} of the simplex, denoted
by $\interior(\sigma)$, is defined as
$\sigma - \boundary{\sigma}$.  More generally, for a set $S$,
we use $\interior(S)$ to refer to the interior of $S$ taken with
respect to the usual topology of $\affine(S)$.
For the {\emph{closure}} of a set $S$ we use the
notation $\closure{S}$.

For a vertex $u$ of a simplicial complex we define $\star~u$,
the {\emph{star}} of the vertex, to be
the union of the interiors of all the simplices for which $u$ is
a vertex.  The closure of the star, or the
{\emph{closed star}}, $\closure{\star~u}$, is
the union of all simplices incident to $u$.  The {\emph{link}}
of a vertex $u$ is defined by $\link~u = \closure{\star~u} - \star~u$.
Many of the terms briefly defined here are defined and
discussed at more length in \cite{Munkres1984}.

We wish to avoid any ambiguity about the dimension of the circumsphere
or circumball of a simplex $\sigma$.  Throughout this paper
the objects {\emph{circumsphere}} and {\emph{circumball}}
always live in $\affine(\sigma)$.  Thus the circumsphere of a
triangle is always a copy of $S^1$, even when the triangle is
embedded in $\RR^3$ as the facet of a tetrahedron.
The {\emph{equatorial ball}} of a simplex $\sigma$, sometimes
denoted $\eqB(\sigma)$, is a ball of radius $R(\sigma)$
centered at $c(\sigma)$, but distinguished from the circumball
of $\sigma$ by the fact that the equatorial ball is considered
in a higher-dimensional space.  For example, the equatorial
ball of a triangle $\tau$ considered in $\RR^3$ is
the unique $3$-dimensional ball that has the circumcircle
of $\tau$ as an equator (see Fig.~\ref{fig:eqball}).
The ambient higher-dimensional space that contains the
equatorial ball should be made clear wherever the
term is used.  In this paper, wherever the term equatorial ball
is used and a higher-dimensional space
is not explicitly specified, the term appears in the context of
a simplex $\sigma^{n} = \cone{u}{\tau^{n-1}}$, and
$\eqB(\tau^{n-1})$ is an $n$-dimensional subset of $\affine(\sigma^{n})$.

We frequently discuss the facets of a simplex
$\sigma^{n} = [v_{0}\ldots v_{n}]$.  As a matter of convention,
the facets usually are denoted
$\tau^{n-1}_{0}, \ldots, \tau^{n-1}_{n}$
with the understanding that the facet $\tau^{n-1}_{i}$ is the
facet opposite vertex $v_{i}$.  Thus
$\sigma^{n} = \cone{v_{i}}{\tau^{n-1}_{i}}$ for each
$i = 0,1,\ldots,n$.

\section{Characterizing the Well-Centered Simplex}
\label{sec:characterize}

We now investigate geometric properties of an $n$-well-cen\-tered
$n$-sim\-plex.  The context for this discussion is
a simplex $\sigma^{n} = \cone{u}{\tau^{n-1}}$ with facet
$\tau^{n-1}$ given.  The vertex $u$ is free to move, and we
wish to determine whether $\sigma^{n}$ is $n$-well-cen\-tered
based on the position of $u$ relative to the fixed vertices
of $\tau^{n-1}$.

We first recall an alternate geometric characterization
of the $n$-well-centered $n$-simplex and state its consequences
in this context.
The alternate characterization is stated in terms of equatorial
balls.  Using the result, which adopts the notational convention
that $\sigma^{n} = \cone{v_{i}}{\tau^{n-1}_{i}}$ for each
$i = 0,1,\ldots,n$, one can determine whether an $n$-simplex
is $n$-well-centered without explicitly computing $c(\sigma^{n})$.

\bigskip

\begin{theorem}[Equatorial Balls Condition]
\label{thm:eqballs}
The simplex $\sigma^{n}$ is $n$-well-centered
if and only if vertex $v_{i}$ lies strictly outside
$B(\tau_{i}^{n-1})$ for each $i=0,1,\ldots,n$.
\end{theorem}
\begin{proof}
See \cite{VaHiGuRa2009}.
\end{proof}

\bigskip

\begin{figure}
\centering
\begin{minipage}[c]{93pt}
\vspace{35pt}
\includegraphics[width=93pt, trim=169pt 38pt 146pt 144pt, clip]
  {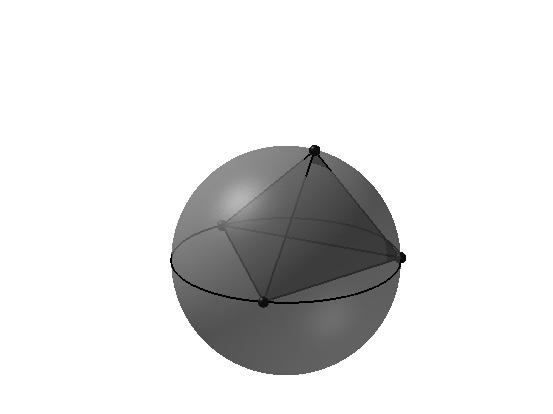}%
\end{minipage}%
\hspace{10pt}%
\begin{minipage}[c]{81pt}
\vspace{24pt}
\includegraphics[width=81pt, trim=222pt 81pt 141pt 149pt, clip]
  {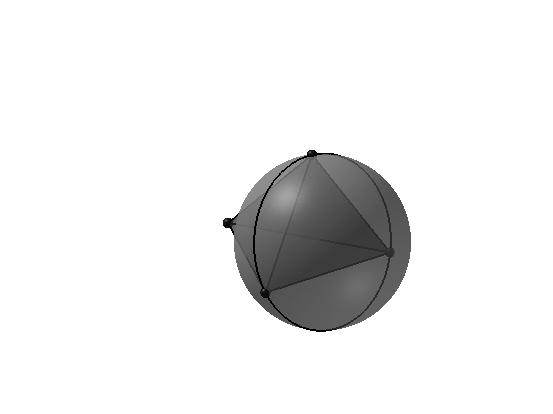}%
\end{minipage}%
\hspace{10pt}%
\begin{minipage}[c]{101pt}
\vspace{14.5pt}
\includegraphics[width=101pt, trim=106pt 26pt 111pt 98pt, clip]
  {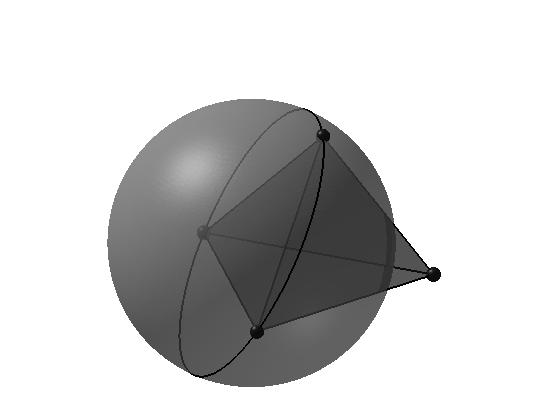}%
\end{minipage}%
\hspace{10pt}%
\begin{minipage}[c]{85pt}
\includegraphics[width=85pt, trim=175pt 65pt 113pt 85pt, clip]
  {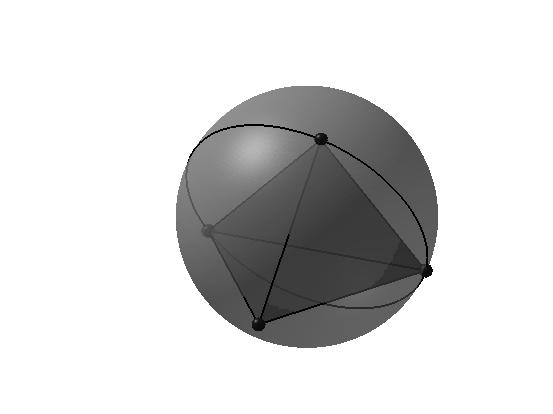}%
\end{minipage}%
\caption{Four views of the same $3$-well-centered tetrahedron $\sigma^{3}$
  in the same orientation.  From left to right the views show
  $\sigma^{3}$ with the equatorial balls
  of its bottom, right, left, and rear facets.
  For each facet $\tau^{2}_{i}$ of $\sigma^{3}$, the circumsphere (i.e.,
  circumcircle) of $\tau^{2}_{i}$, which is an equator of the
  equatorial ball $B(\tau^{2}_{i})$, is shown.  Because the tetrahedron is
  $3$-well-centered, the vertex $v_{i}$ opposite facet
  $\tau^{2}_{i}$ lies outside of $B(\tau^{2}_{i})$; an $n$-simplex is
  $n$-well-centered if and only if for each vertex $v_{i}$,
  $v_{i}$ lies outside of the equatorial ball of the facet
  $\tau^{n-1}_{i}$ opposite $v_{i}$.  For the bottom facet
  and rear facet views in the figure, the reader may need to
  look closely to see that the edges incident to $v_{i}$ do
  pierce $B(\tau^{2}_{i})$, and $v_{i}$ lies outside $B(\tau^{2}_{i})$.}
\label{fig:eqball}
\end{figure}

Figure~\ref{fig:eqball} illustrates Theorem~\ref{thm:eqballs} as
it applies to a tetrahedron.  For each vertex $v_{i}$ of the
tetrahedron, Fig.~\ref{fig:eqball} shows the equatorial ball
$B(\tau_{i})$ of the facet $\tau_{i}$ opposite $v_{i}$, emphasizing
in a darker color the circumcircle of $\tau_{i}$ (which is an
equator of $B(\tau_{i})$).  The figure shows that in each case
$v_{i}$ is outside the equatorial ball of $\tau_{i}$, so we can
conclude that the tetrahedron is $3$-well-centered.  Moreover,
this same condition is satisfied by every $3$-well-centered
tetrahedron.

In the context of an $n$-simplex $\sigma^{n}$ with a free
vertex~$u$ opposite a fixed facet~$\tau^{n-1}$,
Theorem~\ref{thm:eqballs} becomes a necessary condition
that vertex $u$ must satisfy if $\sigma^{n}$ is to be
$n$-well-centered.

\bigskip

\begin{corollary}[One-Facet Equatorial Ball Condition]
\label{cor:onefcteqball}
Let $\sigma^{n} = \cone{u}{\tau^{n-1}}$.
If the simplex $\sigma^{n}$ is $n$-well-centered, then
$u$ lies strictly outside of $\eqB(\tau^{n-1})$.
\end{corollary}

\bigskip

To introduce the remaining results of this section and get a
somewhat different perspective on Corollary~\ref{cor:onefcteqball},
we consider the sketch in Fig.~\ref{fig:sphericaltriangle}.
In the sketch, $\tau$ is a given triangle in $\RR^{3}$ with
fixed vertices.  Triangle $\tau$ represents a facet of a
tetrahedron $\sigma$ whose fourth vertex $u$ has not yet
been determined.  (Thus tetrahedron $\sigma$ does not
appear in Fig.~\ref{fig:sphericaltriangle}.)  We
suppose, however, that the circumcenter $c(\sigma)$ of the
tetrahedron is known, so $u$ is constrained to lie on the
circumsphere of $\sigma$.  The two sides of
Fig.~\ref{fig:sphericaltriangle} show two different
cases; on the left $\tau$ is not $2$-well-cen\-tered, and
on the right $\tau$ is $2$-well-centered.

\begin{figure}
\centering
\resizebox{!}{160pt}{
  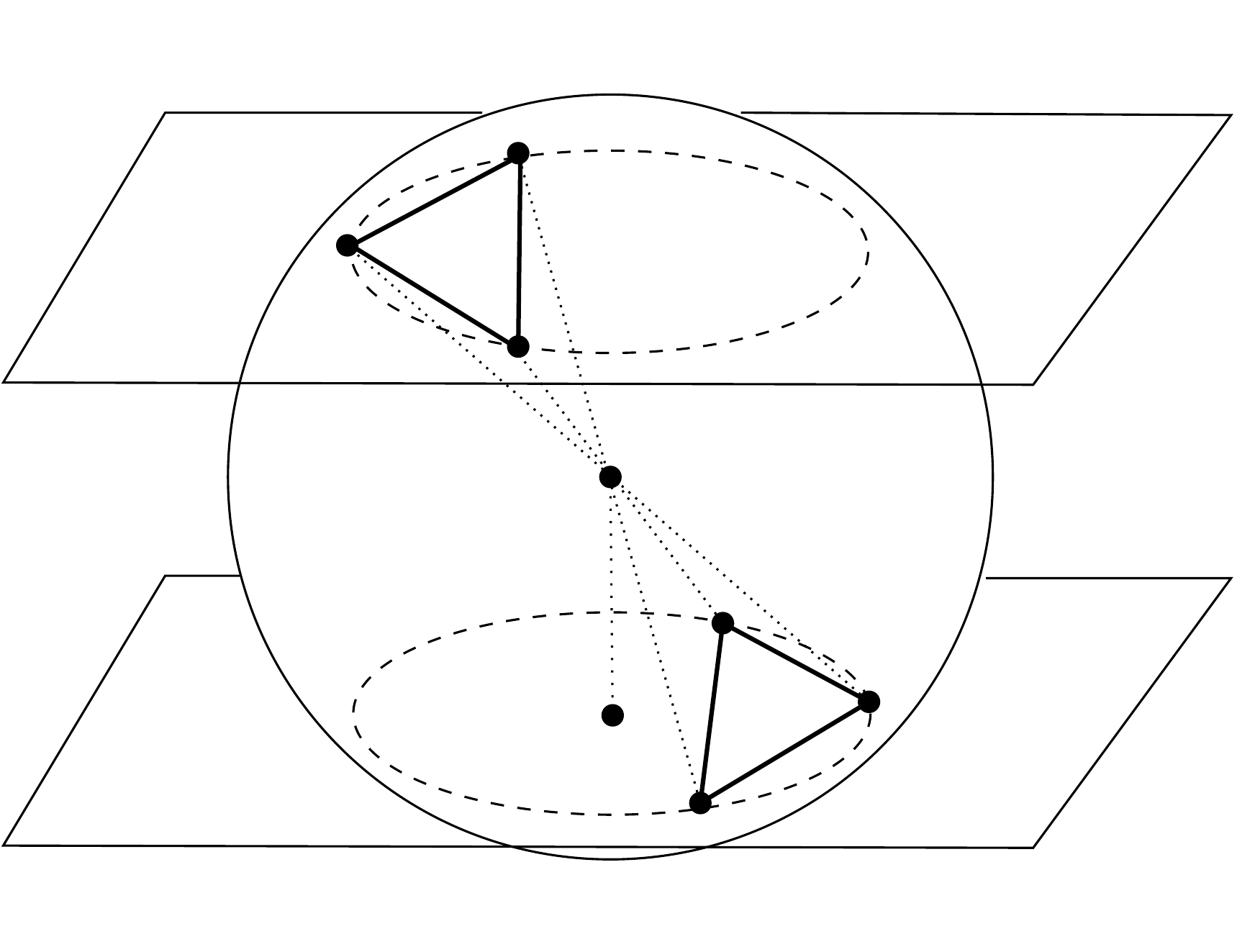}%
\hspace{30pt}%
\resizebox{!}{160pt}{
  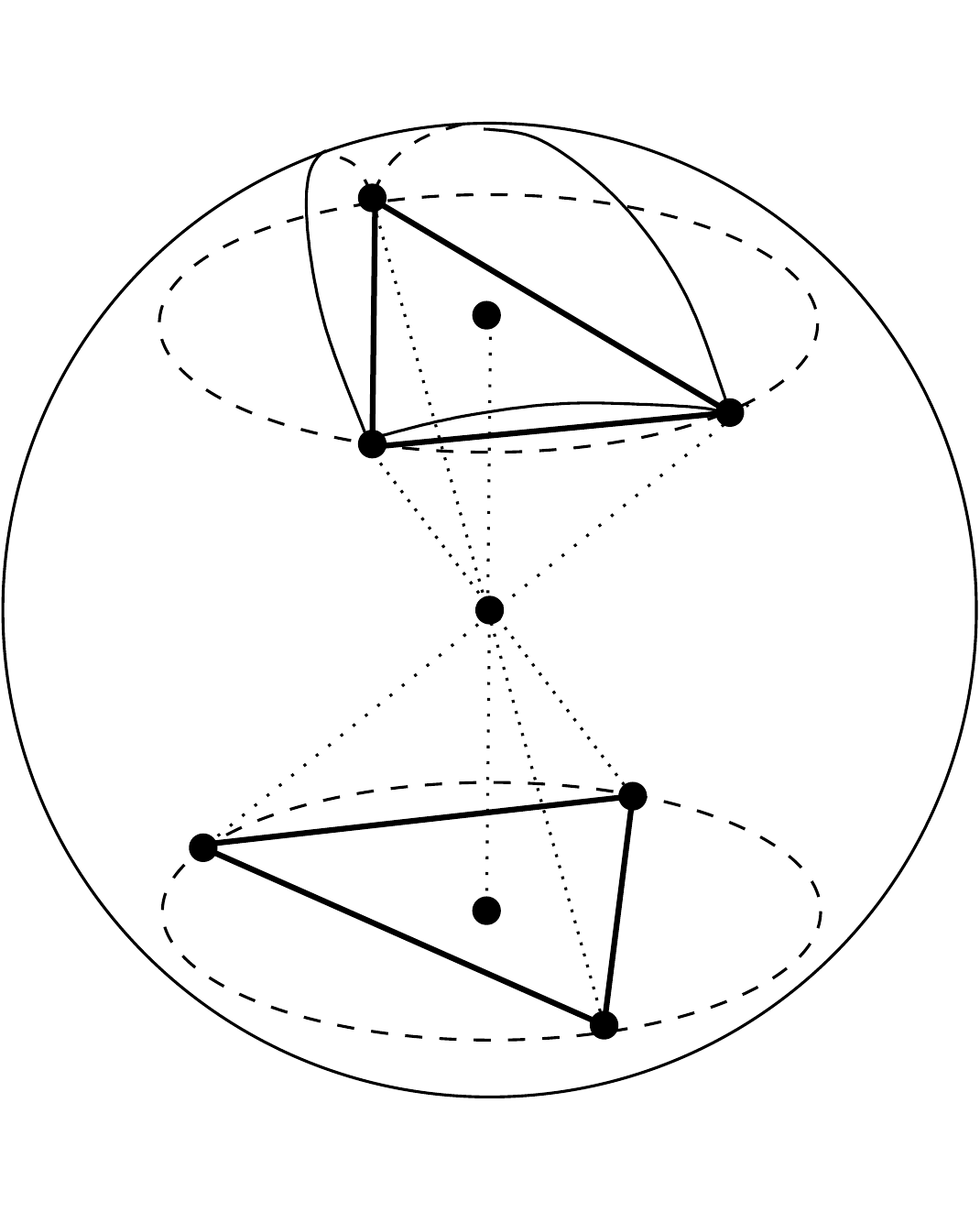}
\caption{A facet $\tau$ of a tetrahedron $\sigma$ (not shown in the
    figure) and the reflection $\tau^{\prime}$ of $\tau$ through
    $c(\sigma)$.  In the case on the left $\tau$ is
    not $2$-well-centered; on the right $\tau$ is
    $2$-well-centered.  The tetrahedron $\sigma$
    is $3$-well-centered if and only if its vertex
    $u$ opposite $\tau$ lies in the spherical triangle
    determined by the intersection of the circumsphere of
    $\sigma$ and a cone on $\tau^{\prime}$ with apex
    $c(\sigma)$.  (The spherical triangle is shown in
    the case on the right, but not on the left.)}
\label{fig:sphericaltriangle}
\end{figure}

Now $\tau$ lies in the plane $\affine(\tau)$, and the reflections
of $\tau$ and $\affine(\tau)$ through $c(\sigma)$ are
$\tau^{\prime}$ and $\affine(\tau^{\prime})$ respectively.  The plane
$\affine(\tau)$ intersects the circumsphere to determine
a lower spherical cup $C$, and $\affine(\tau^{\prime})$
determines an upper spherical cup $C^{\prime}$.
The necessary condition of Corollary~\ref{cor:onefcteqball}
says that when $\sigma$ is $3$-well-centered
$u$ does not lie in the lower spherical cup $C$.  From the
geometry one can see that, in fact, if $\sigma$ is to be
$3$-well-centered, then $u$ must lie strictly inside the spherical
triangle determined by the intersection of $C^{\prime}$, the upper
cup of the circumsphere of $\sigma$, with the geometric cone on
$\tau^{\prime}$ with apex $c(\sigma)$.  (This spherical
triangle is drawn in the case on the right, but not in
the case on the left.)

In particular, there is a necessary
condition that $u$ must lie strictly in the upper cup
$C^{\prime}$ of the sphere.  Thus, speaking with regard to
the orthogonal projection into $\affine(\tau)$, $u$ projects
(vertically in the figure) to the interior of the
circumdisk of $\tau$.  Moreover, if $\tau$ is
$2$-well-centered, as on the right, then the projection of
the spherical triangle into $\affine(\tau)$ contains the
projection of $\tau^{\prime}$ into $\affine(\tau)$.
Thus if the projection of $u$ into $\affine(\tau)$
lies inside the projection of $\tau^{\prime}$
into $\affine(\tau)$, this is sufficient
to establish that $\sigma$ is $3$-well-centered.
These conditions and their generalizations into higher
dimensions are the first two conditions discussed in this
section.  The geometric intuition developed here is formalized
and proved algebraically in Propositions~\ref{prop:char_nec}
and~\ref{prop:char_suf}.

Finally, we consider varying the position of $c(\sigma)$.
Notice that as $c(\sigma)$ moves in
Fig.~\ref{fig:sphericaltriangle} from the circumcenter
of $\tau$ upward along a line orthogonal to $\affine(\tau)$,
the spherical triangle of $u$-positions that produce
a $3$-well-centered tetrahedron with circumcenter $c(\sigma)$
sweeps out a solid $3$-dimensional region.  Tetrahedron $\sigma$
will be $3$-well-centered if and only if $u$ lies in this region.
The section closes by describing this region for arbitrary
dimensions in terms of polynomial inequalities.
(See Fig.~\ref{fig:full_region}.)

Now we state a proposition that gives
a necessary condition for an \mbox{$n$-simplex}~$\sigma^{n}$
to be $n$-well-centered.  See Fig.~\ref{fig:char_nec_prop}.

\bigskip

\begin{proposition}[Cylinder Condition]
\label{prop:char_nec}
Let $\sigma^n$ be an $n$-well-centered $n$-simplex in $\RR^{n}$
with $u$ a vertex of $\sigma^{n}$ and
$\tau^{n-1}$ the facet of $\sigma^{n}$
opposite $u$.  That is, let $\sigma^{n} = \cone{u}{\tau^{n-1}}.$
Let $P$ be the orthogonal projection
$P:\RR^{n} \to \affine(\tau^{n-1}).$
Then $\lVert P(u) - c(\tau^{n-1}) \rVert < R(\tau^{n-1})$, i.e.,
$u$ projects to the interior of the circumball of $\tau^{n-1}$.
\end{proposition}
\begin{proof}
Consider the coordinate system on $\RR^{n}$ such that
$c(\sigma^{n})$ is the origin,
and $\mathrm{aff}(\tau^{n-1}) = \{x \in \RR^{n} : x_{n} = k\}$
for some constant $k \le 0$.  In this coordinate system, $P$
is the projection
\[
P: (x_{1},\ldots,x_{n-1},x_{n}) \mapsto (x_{1},\ldots,x_{n-1},k).
\]
Let $u = (x_{1}, \ldots, x_{n})$ in this
coordinate system.  We have assumed that $\sigma^{n}$ is $n$-well-centered,
so $c(\sigma^{n})$ (the origin) is strictly interior to $\sigma^{n}$.
It follows that $k < 0$ and $x_{n} > 0$.

\begin{figure}
\centering
\includegraphics[width=250pt, trim=118pt 150pt 87pt 130pt, clip]
  {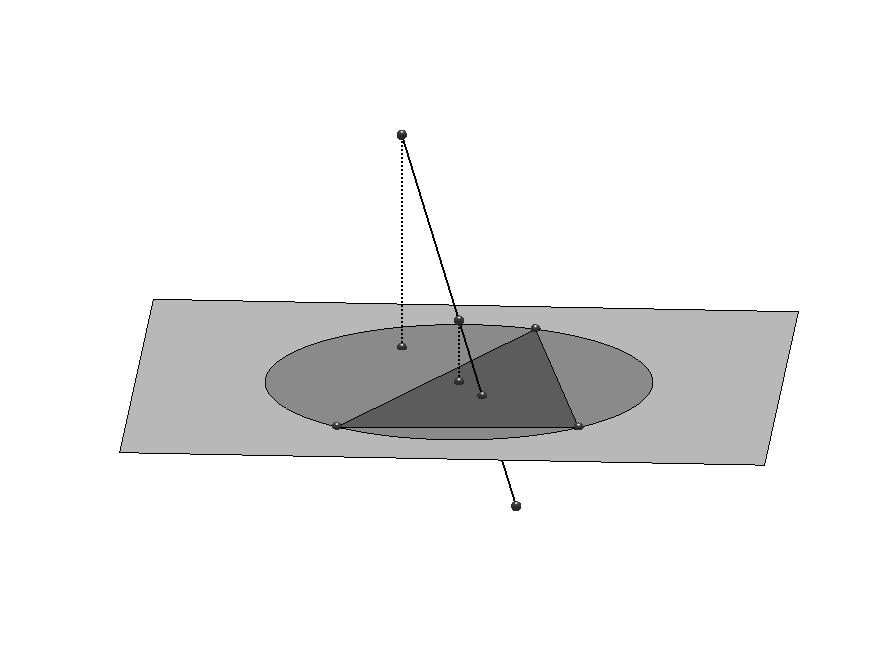}
\begin{picture}(0, 0)(250, 0)
\put(3, 28){\small{$\affine(\tau^{n-1})$}}
\put(91, 137){\small{$u$}}
\put(76, 59){\small{$P(u)$}}
\put(138, 93){\small{$c(\sigma^{n}) = (0, 0, 0)$}}
\put(137, 94){\line(-2, -3){15}}
\put(86, 44){\small{$c(\tau^{n-1})$}}
\put(133, 41){\small{$w$}}
\put(146, 2){\small{$-u$}}
\put(115, 100){\small{$\ell$}}
\put(170, 67){\small{$\tau^{n-1}$}}
\put(169, 68){\line(-2, -1){27}}
\end{picture}
\caption{Because the tetrahedron $\sigma^{n} = \cone{u}{\tau^{n-1}}$ is
  $3$-well-centered, $P(u)$ lies inside the circumball of $\tau^{n-1}$.}
\label{fig:char_nec_prop}
\end{figure}
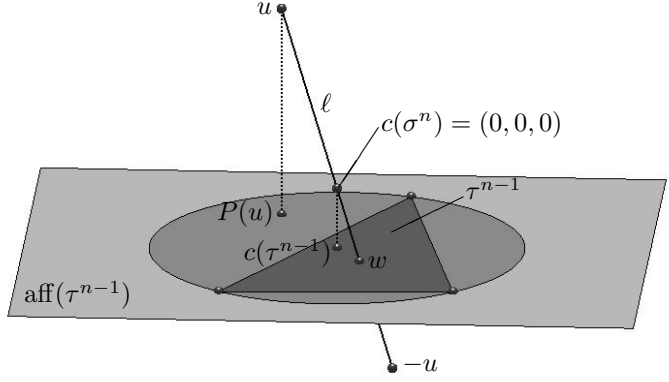

Consider the line segment $\ell$ from $u$ to $-u$.  Observe that $\ell$
is a diameter of the circumsphere of $\sigma^{n}$.  Moreover,
$\ell \cap \interior(\tau^{n-1}) \ne \emptyset$.  This follows from
the fact that $\sigma^{n}$ is $n$-well-centered;  we have
$\sigma^{n} = \cone{u}{\tau^{n-1}}$ and
$c(\sigma^{n}) \in \interior(\sigma^{n})$,
so there must be some point $w \in \interior(\tau^{n-1})$ such
that $c(\sigma^{n})$ lies on $uw \subsetneq \ell$.  We notice, then,
that the point $-u$ lies below $\affine(\tau^{n-1})$ and conclude
that $x_{n} > -k$.

By the Pythagorean theorem, $R(\tau^{n-1})^{2} + k^2 = R(\sigma^{n})^2$.
We also have
\[
\sum_{i=1}^{n} x_{i}^{2} = R(\sigma^{n})^2,
\]
since $u$ lies on the circumsphere of $\sigma^{n}$.  It follows
that
\[
\quad\lVert P(u) - c(\tau^{n-1}) \rVert^2 = \sum_{i=1}^{n-1} x_{i}^{2}
    = R(\sigma^{n})^2 - x_{n}^2 < R(\sigma^{n})^2 - k^2 = R(\tau^{n-1})^2.
\]
\end{proof}

\bigskip

The statement is not limited to $\sigma^{n} \subset \RR^{n},$ of course.
For a simplex $\sigma^{n} \subset \RR^{m}$ with $m > n,$ there
exists a coordinate system such that
\[
\mathrm{aff}(\sigma^{n}) = \{x \in \RR^{m} :
  x_{i}=0\ \mathrm{for}\ i=n+1,\ldots,m\},
\]
and the same proof applies.

\begin{figure}
\centering
\begin{minipage}[c]{63pt}
\includegraphics[width=63pt, trim=254pt 237pt 238pt 227pt, clip]
  {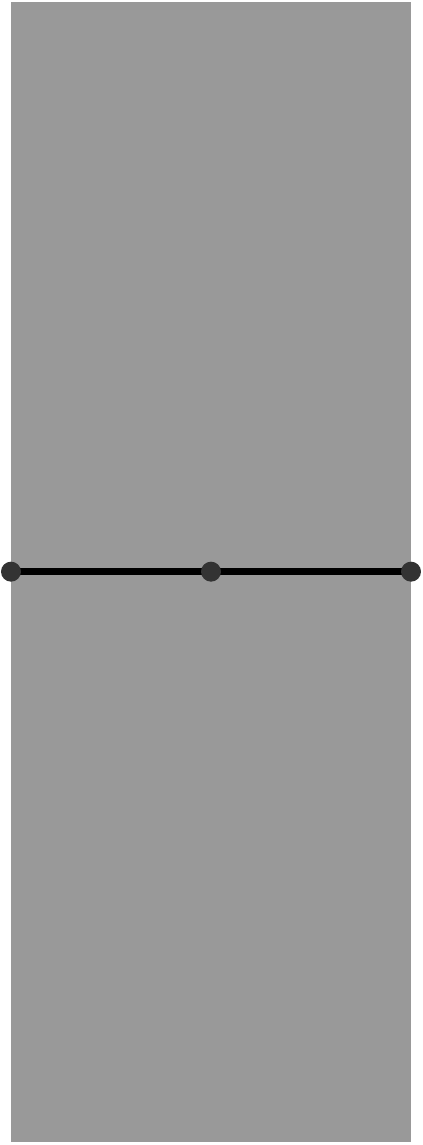}%
\end{minipage}%
\hspace{50pt}%
\begin{minipage}[c]{63pt}
\includegraphics[width=63pt, trim=672pt 287pt 543pt 138pt, clip]
  {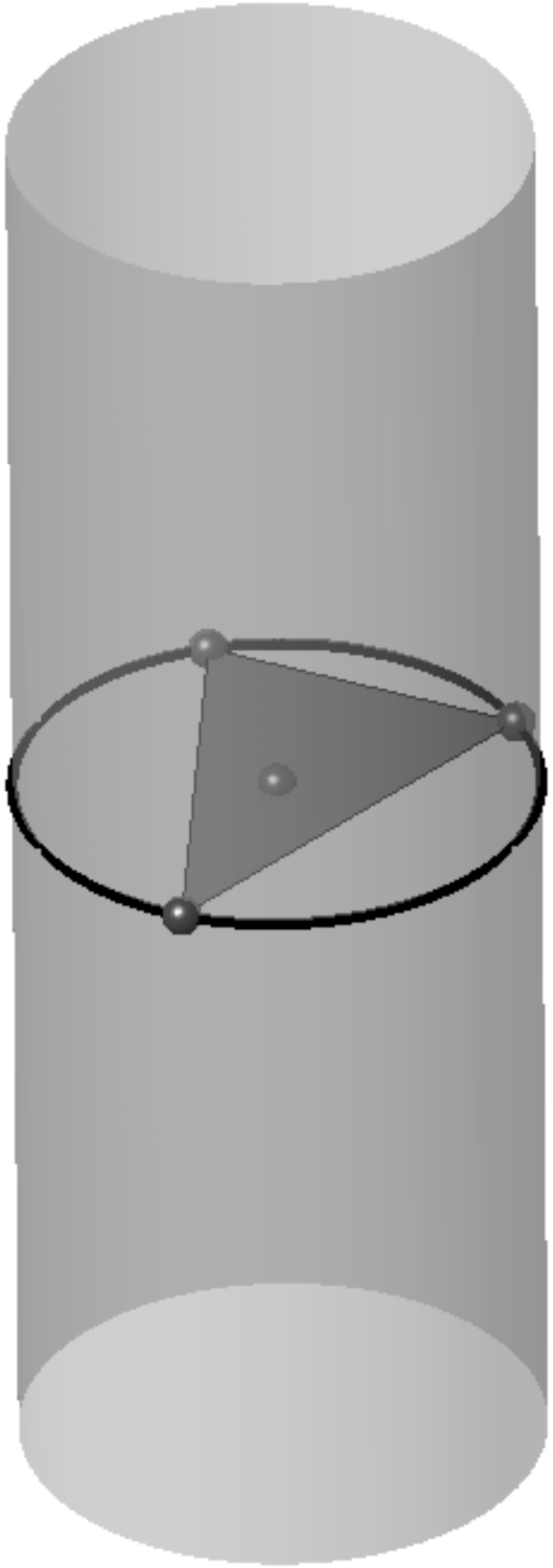}%
\end{minipage}
\caption{If a simplex $\cone{u}{\tau^{n-1}}$ is $n$-well-centered,
  then $u$ is interior to a solid right spherical cylinder over the
  circumsphere of $\tau^{n-1}$.}
\label{fig:nec_cond_region}
\end{figure}

\begin{remark}
Given a particular simplex $\tau^{n-1} \subset \RR^{n}$,
Proposition~\ref{prop:char_nec} provides a geometric
necessary condition on the location of vertex $u$ to
create an $n$-well-centered
simplex $\sigma^{n} = \cone{u}{\tau^{n-1}}.$
Vertex $u$ must lie within a solid right spherical cylinder
over the circumsphere of $\tau^{n-1}$ if $\sigma^{n}$ is to be
$n$-well-centered.  Figure~\ref{fig:nec_cond_region} illustrates
the condition in 2D and 3D, making it clear how this condition
generalizes from the familiar 2-D case into higher dimensions.
In each case the vertices of the base
simplex $\tau^{n-1}$, as well as the circumcenter
$c(\tau^{n-1})$, are marked by small dark-colored balls.
If $\cone{u}{\tau^{n-1}}$ is $n$-well-centered,
then the vertex $u$ must lie inside the gray cylinder over
the circumsphere of $\tau^{n-1}$.  In the notation
of Fig.~\ref{fig:sphericaltriangle}, where the
circumcenter of $\sigma$ is known, the Cylinder Condition
says that $u$ must lie either in the upper cup $C^{\prime}$ or
the lower cup $C$.
\end{remark}

\begin{remark}
The One-Facet Equatorial Ball Condition
(Corollary~\ref{cor:onefcteqball})
is also a necessary condition on the location of vertex $u$.
In $\RR^2$, the combination of Corollary~\ref{cor:onefcteqball}
and the Cylinder Condition is sufficient to guarantee that a
triangle (a $2$-simplex) is acute (is $2$-well-centered).  In
$\RR^n$ for $n \ge 3$, however, an $n$-simplex $\cone{u}{\tau^{n-1}}$
for which $u$ satisfies both of these necessary conditions
might not be $n$-well-centered.
\end{remark}

\begin{example}
For example, consider the tetrahedron $\sigma = \sigma^{3}$ with
vertices $(-0.152, 0.864, -0.48)$,
$(-0.64, -0.6, -0.48)$, $(0.6, -0.64, -0.48)$,
and $(-0.192, -0.64, 0.744)$, whose circumcenter
lies at the origin.  For three of the four facets $\tau^{2}_{i}$
of $\sigma^{3}$, vertex $v_{i}$ satisfies both necessary conditions
with respect to $\tau^{2}_{i}$, and for the fourth facet
$v_{i}$ satisfies the Cylinder Condition, but not the
One-Facet Equatorial Ball Condition.  The tetrahedron $\sigma^{3}$
is not $3$-well-centered.

\begin{figure}
\centering
\begin{minipage}[c]{170pt}
\includegraphics[width=170pt, trim=270pt 130pt 240pt 100pt, clip]
  {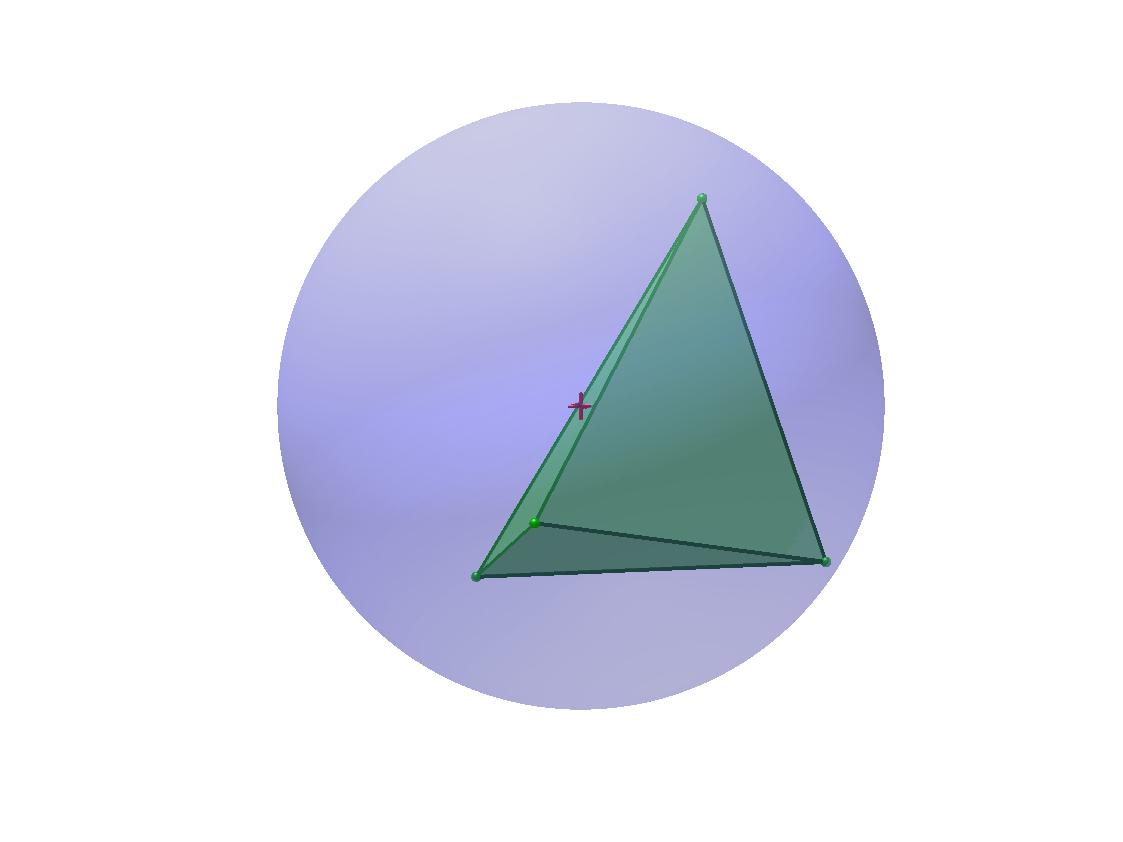}
\end{minipage}%
\hspace{20pt}%
\begin{minipage}[c]{200pt}
\begin{minipage}[c]{90pt}
\includegraphics[width=90pt, trim=145pt 74pt 126pt 59pt, clip]
  {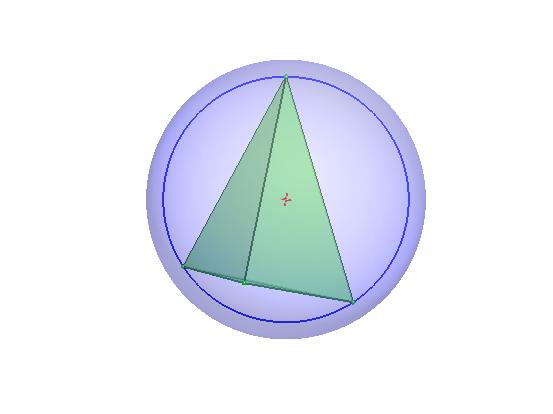}
\end{minipage}%
\hspace{10pt}%
\begin{minipage}[c]{90pt}
\includegraphics[width=90pt, trim=170pt 98pt 151pt 82pt, clip]
  {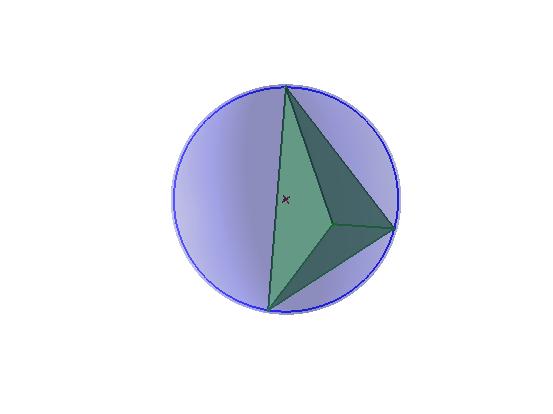}
\end{minipage}\\[7pt]
\begin{minipage}[c]{90pt}
\includegraphics[width=90pt, trim=172pt 100pt 154pt 85pt, clip]
  {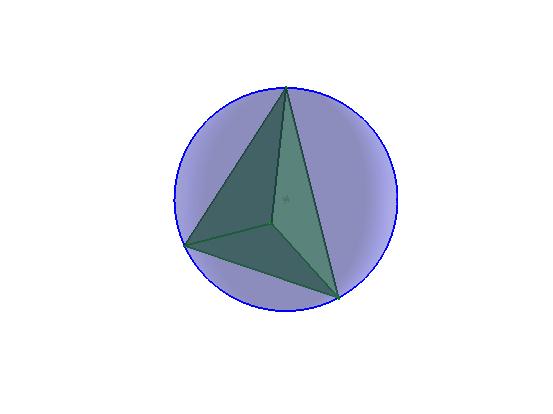}
\end{minipage}%
\hspace{10pt}%
\begin{minipage}[c]{90pt}
\includegraphics[width=90pt, trim=183pt 111pt 164pt 96pt, clip]
  {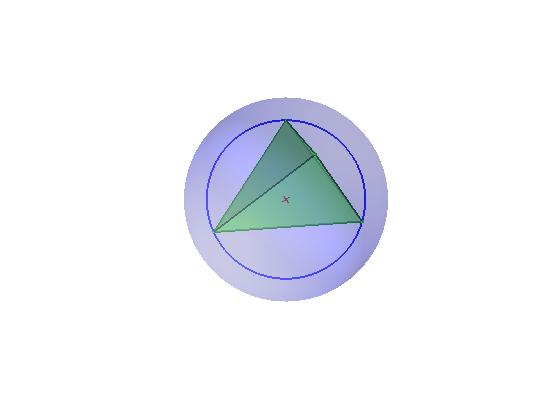}
\end{minipage}
\end{minipage}
\caption{A tetrahedron that is not $3$-well-centered, even though
  every vertex satisfies the necessary condition of
  Proposition~\ref{prop:char_nec}.  Three of the vertices also lie
  outside the equatorial balls of their respective opposite facets.}
\label{fig:tetan2y3n}
\end{figure}

Figure~\ref{fig:tetan2y3n} shows
several different views of $\sigma^{3}$.  The large view at left
shows that $\sigma^{3}$ is not $3$-well-centered; the circumcenter
of $\sigma^{3}$, marked by a small axes indicator, lies outside
the tetrahedron.  The four small views on the
right side of Fig.~\ref{fig:tetan2y3n}
are views directly down onto the facets $\tau^{2}_{i}$ of $\sigma^{3}$.
In each case the circumcircle of $\tau^{2}_{i}$
is rendered in a darker color, and one can see that the vertex above
the facet projects to the interior of the circumball of the facet, i.e.,
that vertex $v_{i}$ satisfies the Cylinder Condition with
respect to $\tau^{2}_{i}$.
In three of the four cases --- all except the case at lower left ---
the vertex also satisfies the One-Facet Equatorial Ball Condition.
The particular example in Fig.~\ref{fig:tetan2y3n} is also
mentioned in~\cite{VaHiGu2008}, which gives some additional
statistics on the tetrahedron.
\end{example}

\begin{example}
The tetrahedron with vertices at
$(-0.01, -0.01, -0.01)$, $(1, 0, 0)$, $(0, 1, 0)$, and $(0, 0, 1)$
is another tetrahedron that is not $3$-well-centered.
It also has three vertices that satisfy the
One-Facet Equatorial Ball Condition and four vertices that
satisfy the Cylinder Condition.  This example is dihedral
acute, in contrast to the previous example.
\end{example}

The above examples illustrate that the One-Facet Equatorial
Ball Condition and the Cylinder Condition are not enough to
establish that the $n$-simplex $\cone{u}{\tau^{n-1}}$
is $n$-well-centered.  However, the following proposition does
provide sufficient conditions that
$\cone{u}{\tau^{n-1}}$ is $n$-well-centered.  See also
Fig.~\ref{fig:char_suf_prop}.

\bigskip

\begin{proposition}[Prism Condition]
\label{prop:char_suf}
Let $\tau^{n-1}$ be an $(n-1)$-well-centered simplex in $\RR^{n}$
and $\sigma^{n} = \cone{u}{\tau^{n-1}}$.
If $u$ lies outside the equatorial
ball $\eqB(\tau^{n-1})$ and the reflection of $P(u)$ through
$c(\tau^{n-1})$ is interior to $\tau^{n-1}$, then
$\sigma^{n}$ is $n$-well-centered.
\end{proposition}
\begin{proof}
We assume the stated hypothesis and
take the same coordinate system that was used in the proof of
Proposition~\ref{prop:char_nec}.  Observe that if
$u$ were on the equatorial ball of $\tau^{n-1}$, then
$c(\sigma^{n})$ would lie in $\tau^{n-1}$, coinciding with
$c(\tau^{n-1})$.  Because
$u$ lies outside the equatorial ball of $\tau^{n-1}$, $c(\sigma^{n})$
lies interior to the same halfspace as $u$ with respect to
$\affine(\tau^{n-1})$.
It follows that $k < 0$ and $x_{n} > k$.

\begin{figure}
\centering
\includegraphics[width=250pt, trim=118pt 150pt 87pt 130pt, clip]
  {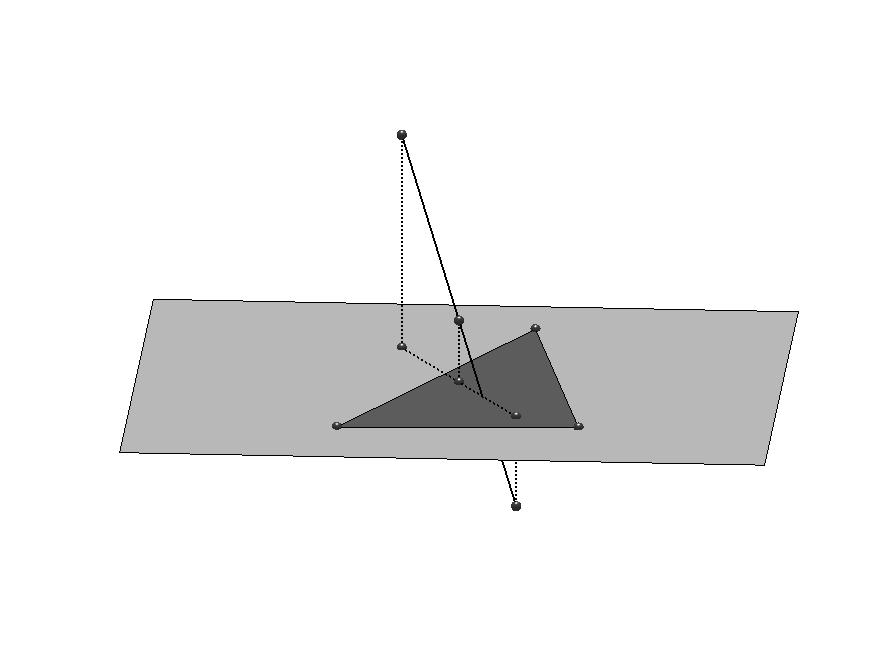}
\begin{picture}(0, 0)(250, 0)
\put(3, 28){\small{$\affine(\tau^{n-1})$}}
\put(91, 137){\small{$u$}}
\put(76, 59){\small{$P(u)$}}
\put(138, 93){\small{$c(\sigma^{n}) = (0, 0, 0)$}}
\put(137, 94){\line(-2, -3){15}}
\put(86, 44){\small{$c(\tau^{n-1})$}}
\put(170, 42){\small{$P(-u)$}}
\put(169, 45){\line(-3, -1){26.5}}
\put(146, 2){\small{$-u$}}
\put(115, 100){\small{$\ell$}}
\put(170, 67){\small{$\tau^{n-1}$}}
\put(169, 68){\line(-2, -1){27}}
\end{picture}
\caption{Because $P(-u)$ and $c(\tau^{n-1}) = P(c(\sigma^{n}))$ are both
  interior to $\tau^{n-1}$ and $c(\sigma^{n})$ is above
  $\affine(\tau^{n-1})$, we know that the tetrahedron
  $\sigma^{n} = \cone{u}{\tau^{n-1}}$ is $3$-well-centered.}
\label{fig:char_suf_prop}
\end{figure}
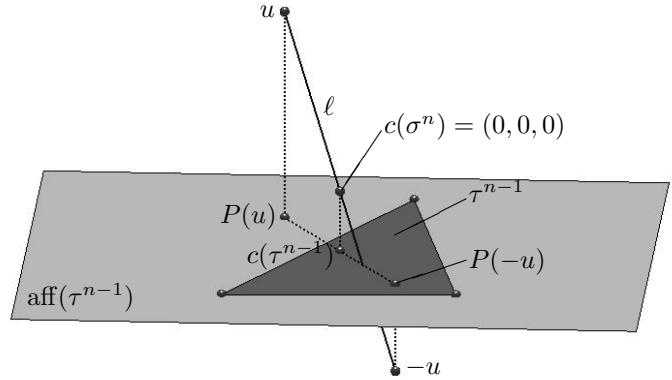

Observe that, as shown in Fig.~\ref{fig:char_suf_prop}, the reflection of
$P(u)$ through $c(\tau^{n-1})$ is $P(-u)$.
By the hypothesis, $P(-u)$ is interior to
$\tau^{n-1}$.  Thus $P(-u)$ is interior to the circumball
of $\sigma^{n}$ and
\[
\lVert P(u) \rVert^2 = \lVert P(-u) \rVert^2
  = k^2 + \sum_{i=1}^{n-1} x_{i}^{2} <
  R(\sigma^{n})^2 = \sum_{i=1}^{n} x_{i}^{2}
\]
It follows that $\lvert x_{n} \rvert > \lvert k \rvert = -k$.
Since we know that $x_{n} > k$, we conclude that $x_{n} > -k > 0$.

Let $\ell$ be the line segment from $u$ to $-u.$  We will show that
$\ell$ intersects the interior of $\tau^{n-1}$.
Then, because $\sigma^{n} = \cone{u}{\tau^{n-1}}$ and
$k < 0 < x_{n}$ (so that $c(\sigma^{n}) \in \ell$ is above
$\tau^{n-1}$ and below $u$), we will be able to conclude
that $c(\sigma^{n})$ is interior to $\sigma^{n}$.
We know that $P(c(\sigma^{n})) = c(\tau^{n-1})$
is interior to $\tau^{n-1}$ because $\tau^{n-1}$ is
$(n-1)$-well-centered.  Since $P(-u)$ and $P(c(\sigma^{n}))$ are both
interior to $\tau^{n-1},$ the line segment from $c(\sigma^{n})$ to
$-u$, which is contained in $\ell$, is interior to the (convex)
infinite prism $\tau^{n-1} \times \RR$.  Moreover, $0 > k > -x_{n}$
(i.e., $c(\sigma^{n})$ is above $\tau^{n-1}$ and $-u$ is
below $\tau^{n-1}$), so this part of segment $\ell$ intersects
the interior of $\tau^{n-1}$.
\end{proof}

\bigskip

As was the case for Proposition~\ref{prop:char_nec},
Proposition~\ref{prop:char_suf} is not limited to
$\sigma^{n} \subset \RR^{n}$;  in higher-dimen\-sional
spaces $\RR^{m}$ there is a coordinate
system such that
\[
\mathrm{aff}(\sigma^{n}) = \{x \in \RR^{m} :
  x_{i}=0\ \mathrm{for}\ i=n+1,\ldots,m\},
\]
and the same proof applies.

After reading Proposition~\ref{prop:char_suf}, one might ask whether
the requirement that the facet $\tau^{n-1}$ be $(n-1)$-well-centered can
be removed from the proposition.  It may already be clear
from the discussion of Fig.~\ref{fig:sphericaltriangle} that
the answer to this question is no.  The tetrahedron in
Fig.~\ref{fig:tetan2n3n} is an explicit example that
confirms the requirement cannot be removed.

\begin{example}
\begin{figure}
\centering
\includegraphics[width=150pt, trim=250pt 110pt 210pt 80pt, clip]
  {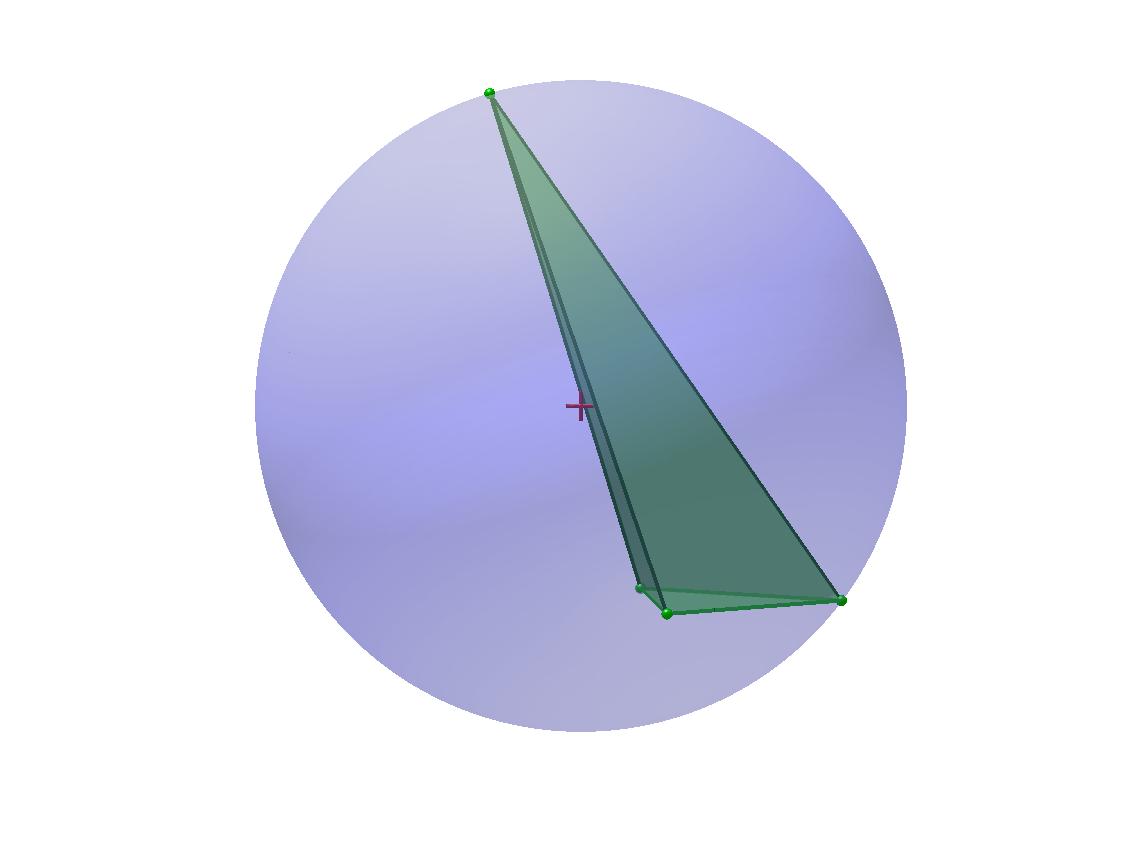}
\caption{A tetrahedron for which the top vertex and bottom facet satisfy
  all of the sufficient conditions for being $3$-well-centered except
  that the bottom facet is not $2$-well-centered.}
\label{fig:tetan2n3n}
\end{figure}

The tetrahedron in Fig.~\ref{fig:tetan2n3n} is the convex hull
of vertices
$v_{0} = (0.224, -0.768, -0.6)$,
$v_{1} = (0.8, 0, -0.6)$, $v_{2} = (0.224, 0.768, -0.6)$, and
$v_{3} = (-0.28, 0, 0.96)$.
The bottom facet in Fig.~\ref{fig:tetan2n3n}, which is the triangle
$\tau^{2}_{3} = [v_{0}v_{1}v_{2}]$, lies in the plane
$x_{3} = -0.6$ and is an obtuse triangle.
The obtuse angle is at vertex $v_{2}$, the rightmost vertex in
Fig.~\ref{fig:tetan2n3n}.  Taking this bottom facet to be $\tau^{2}$
as in Proposition~\ref{prop:char_suf}, and the top vertex to
be $u = v_{3}$, we satisfy the conditions that
$u$ lie outside the equatorial ball of $\tau^2$ and that the
reflection of $P(u)$ through $c(\tau^2)$ be interior to $\tau^2$.
Indeed, $c(\tau^{2}) = (0, 0, -0.6)$ and $R(\tau^{2}) = 0.8$
with $\lVert u - c(\tau^{2})\rVert = \sqrt{2.512} \approx 1.58$,
so $u$ is outside $\eqB(\tau^{2})$, and $P(-u) = (0.28, 0, -0.6)$
lies inside $\tau^{2}$.
Thus we satisfy all of the Prism Condition
except the requirement that $\tau^2$ be $2$-well-centered.
It is clear from Fig.~\ref{fig:tetan2n3n} that this is
not sufficient;  the circumcenter of tetrahedron $\cone{u}{\tau^{2}}$,
marked by a small axes indicator, lies outside the tetrahedron,
so the tetrahedron is not $3$-well-centered.
\end{example}

\begin{remark}
\begin{figure}
\centering
\begin{minipage}[c]{63pt}
\includegraphics[width=63pt, trim=254pt 237pt 238pt 227pt, clip]
  {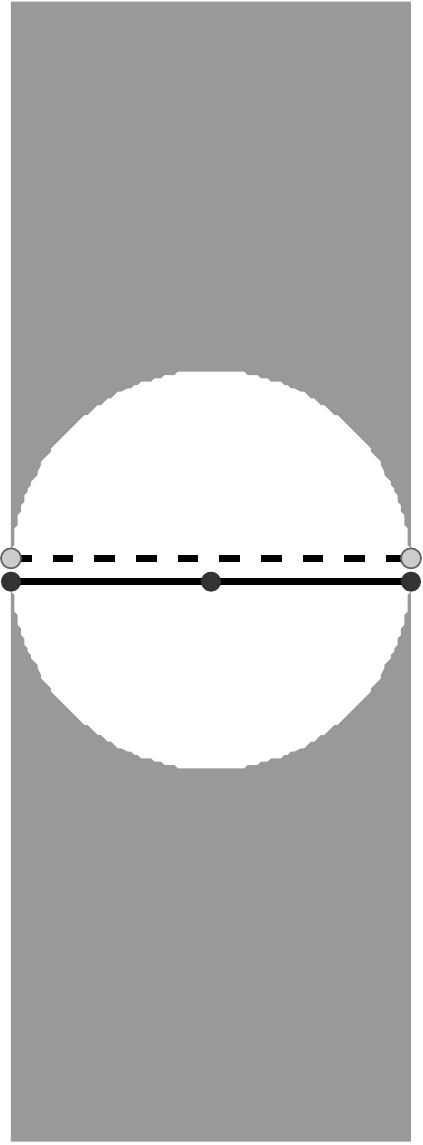}%
\end{minipage}%
\hspace{50pt}%
\begin{minipage}[c]{63pt}
\includegraphics[width=63pt, trim=672pt 287pt 543pt 138pt, clip]
  {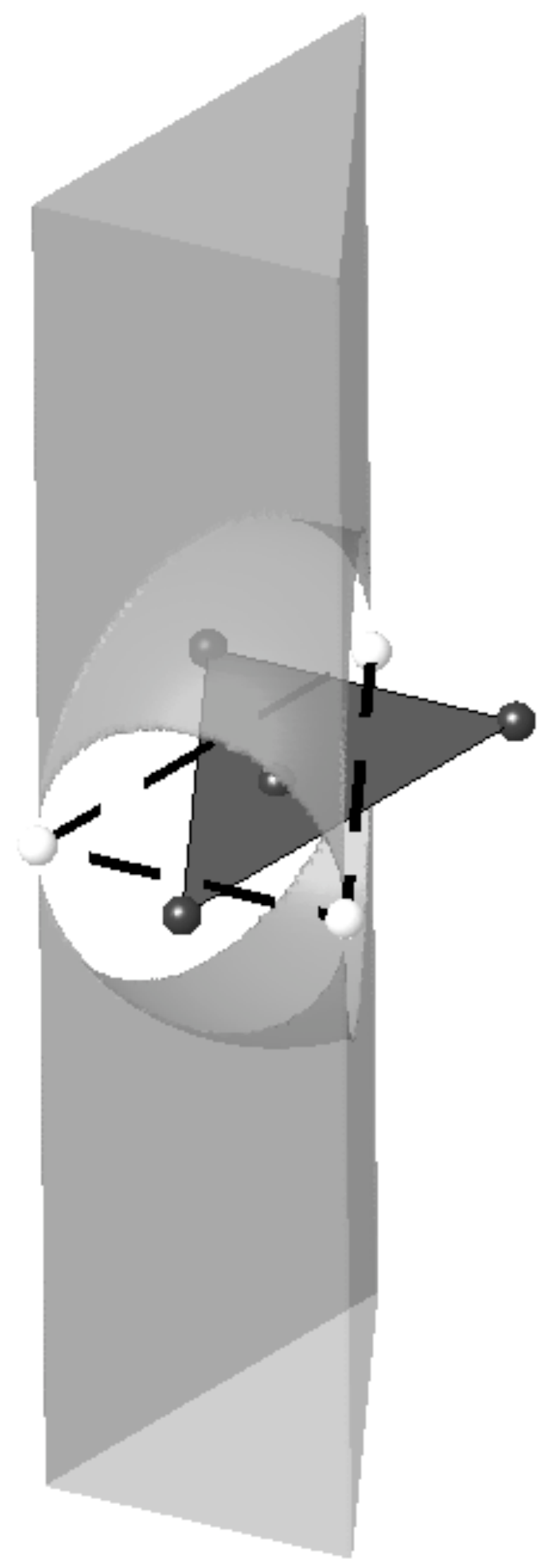}%
\end{minipage}
\caption{If the base simplex $\tau^{n-1}$ is $(n-1)$-well-centered
  and vertex $u$ is both outside $B(\tau^{n-1})$ and inside the
  infinite prism over the reflection of $\tau^{n-1}$ through
  $c(\tau^{n-1})$, then the simplex $\cone{u}{\tau^{n-1}}$
  is $n$-well-centered.}
\label{fig:suf_cond_region}
\end{figure}

Like the condition of Proposition~\ref{prop:char_nec}, the
condition of Proposition~\ref{prop:char_suf} has a nice
geometric interpretation.  Given an $(n-1)$-well-centered
facet $\tau^{n-1}$,
if the vertex $u$ opposite $\tau^{n-1}$ lies outside
$B(\tau^{n-1})$ and within an infinite
prism (a right cylinder) over
the reflection of $\tau^{n-1}$ through its circumcenter, then
$\sigma^{n} = \cone{u}{\tau^{n-1}}$ is $n$-well-centered.
Figure~\ref{fig:suf_cond_region} portrays the region defined by the
Prism Condition for specific examples in $2$ and $3$ dimensions.
In each case the base simplex $\tau^{n-1}$ is shown in dark colors and solid
lines, and its reflection is outlined with lighter colors and dashed
lines.  In the figure, each $\tau^{n-1}$ is $(n-1)$-well-centered,
so for a vertex $u$ lying inside the prism over the reflection
of $\tau^{n-1}$ through $c(\tau^{n-1})$ and outside
the equatorial ball of $\tau^{n-1}$,
i.e., for a vertex $u$ lying in the gray region
shown in Fig.~\ref{fig:suf_cond_region},
the simplex $\cone{u}{\tau^{n-1}}$ will be $n$-well-centered.  Note that
on the left in Fig.~\ref{fig:suf_cond_region} the base simplex and its
reflection should actually lie on top of each other, but are set
slightly apart in the drawing so the reader can distinguish them
from each other.
\end{remark}

We have now established two different conditions for an $n$-simplex
to be $n$-well-centered.  One condition is a necessary condition,
and the other condition is a sufficient condition.  Both conditions
are stated in terms of the location of a vertex $u$ relative
to the facet $\tau^{n-1}$ opposite $u$.  The regions defined
by the necessary condition and the sufficient condition may be
quite different from each other.  For example, in the 3-D
portions of Figs.~\ref{fig:nec_cond_region} and~\ref{fig:suf_cond_region}
the same base simplex $\tau^{n-1}$ yields rather different regions
for the two conditions.  It is natural
to seek a precise description of the region where the vertex
$u$ will produce an $n$-well-centered $n$-simplex
$\cone{u}{\tau^{n-1}}$.
The following discussion develops just such a set of conditions
on the location of $u$.  The conditions take the form of a
system of cubic polynomial inequalities in the coordinates of $u$.
The simplex $\cone{u}{\tau^{n-1}}$ will be $n$-well-centered
if and only if the coordinates of $u$ satisfy the polynomial
inequalities.

The inequalities are derived from a linear system of equations discussed
in \cite{BeHi2009}.  This linear system, which provides one way to
compute the circumcenter of a simplex $\sigma^{n}$ embedded in
$\RR^{m}$ for $m \ge n$, is briefly reviewed here.
We may write the circumcenter $c$ of a simplex $\sigma^{n} =
[v_{0}v_{1}\ldots v_{n}]$ as a linear combination of the vertices
$v_{i} \in \RR^{m}$,
\[
c = \alpha_{0}v_{0} + \alpha_{1}v_{1} + \cdots + \alpha_{n}v_{n},
\]
with the coefficients $\alpha_{i}$ satisfying
$\sum_{i=0}^{n} \alpha_{i} = 1$.
The coefficients $\alpha_{i}$ are known as the
\emph{barycentric coordinates} of the circumcenter.
The condition that $\sigma^{n}$ be $n$-well-centered
is the same as the condition that $0 < \alpha_{i}$ for every
$\alpha_{i}$, i.e., the condition that the circumcenter be a convex
combination of the vertices of $\sigma^{n}$ with strictly positive
coefficients.

Suppose we are given the coordinates of the vertices $v_{i}$
of $\sigma^{n}$.  We know that
\[
\langle c - v_{i}, c - v_{i} \rangle = \lVert c - v_{i} \rVert^2
    = R^2
\]
for each vertex $v_{i}.$  Introducing the
variable $\lambda = R^2 - \lVert c \rVert^2$, we obtain the
$n+1$ equations $2\langle c, v_{i}\rangle + \lambda =
\lVert v_{i} \rVert^2$.  Since the vertices $v_{i}$ are known,
each equation is a linear equation in the $n+2$ unknowns
$\alpha_{0}, \alpha_{1}, \ldots, \alpha_{n}, \lambda.$
The final equation of the system is $\sum_{i=0}^{n} \alpha_{i} = 1$,
which forces the $\alpha_{i}$ to be barycentric coordinates.
As long as this linear system of $n+2$ equations in $n+2$
unknowns is nonsingular, we can solve for the barycentric
coordinates.  If the simplex is nondegenerate, i.e., if the
$n+1$ vertices are affinely independent, then the simplex has a
unique, finite circumcenter, which has unique barycentric coordinates.
It follows that the linear system has a unique
solution; hence the matrix is nonsingular.

Let $A$ be the matrix of this linear system and $b$ the right-hand
side,
\[
A =
\begin{pmatrix}
2\langle v_{0}, v_{0}\rangle & 2\langle v_{0}, v_{1} \rangle
& \cdots & 2\langle v_{0}, v_{n}\rangle & 1\\
2\langle v_{1}, v_{0}\rangle & 2\langle v_{1}, v_{1} \rangle
& \cdots & 2\langle v_{1}, v_{n}\rangle & 1\\
\vdots & \vdots & \ddots & \vdots & \vdots \\
2\langle v_{n}, v_{0}\rangle & 2\langle v_{n}, v_{1} \rangle
& \cdots & 2\langle v_{n}, v_{n}\rangle & 1\\
1 & 1 & \cdots & 1 & 0
\end{pmatrix},
\qquad
b =
\begin{pmatrix}
\langle v_{0}, v_{0}\rangle\\
\langle v_{1}, v_{1}\rangle\\
\vdots\\
\langle v_{n}, v_{n}\rangle\\
1
\end{pmatrix}.
\]
For $i=0,1,\ldots,n$ we let $A_{i}$ be the matrix $A$ with
column $i + 1$ replaced by $b$.  Cramer's rule tells us that
$\alpha_{i} = \det(A_{i})/\det(A).$  If we consider vertices
$v_{0},\ldots,v_{n-1}$ to be the vertices of some given $\tau^{n-1}$
and $v_{n}$ to be a free vertex $u$, then the barycentric
coordinates $\alpha_{i}$ are rational functions of the coordinates of $u$.
Thus the conditions $\alpha_{i} > 0$ become algebraic inequalities
in the coordinates of $u$.

To simplify matrix $A$ a little, we translate each vertex of the
simplex by $-v_{0}$.  The translation
may change the value of $\lambda$ in the solution vector --- in fact,
$\lambda = 0$ always holds for the translated system ---
but the barycentric coordinates of the circumcenter are not changed
by translating the vertices of the simplex.
If $m > n$ we make one further simplification.
In the translated coordinate system, we rotate the simplex
about the origin $v_{0}$ to obtain a simplex for which
vector $v_{i} - v_{0} \in \{x:x_{n+1} = \cdots = x_{m} = 0\}$
for each $i = 1,\ldots,n$.  Rotation
about the origin is an orthogonal transformation, so it does
not change any of the entries of the linear system and does
not affect the barycentric coordinates.

If we restrict our
attention to one of the open halfspaces bounded by
$\affine(\tau^{n-1})$, we have either $\det(A) > 0$ or
$\det(A) < 0$ throughout the halfspace, because $\det(A)$ is a
continuous function of the entries in $A$ and $A$ is
singular only when $u \in \affine(\tau^{n-1})$.
We will see that, in fact, $\det(A) \le 0$ holds everywhere,
so $\det(A) < 0$ throughout the halfspace.

The first row and the first
column of $A$ in the simplified linear system are all zeroes
except for the last entry, which is $1$ in both cases.  Computing the
determinant of $A$ by first expanding across the first row and then
expanding down the first column (one with an odd number of entries
and the other with an even number of entries) we find that
$\det(A) = -\det(B)$ where $B$ is the submatrix of $A$ spanning rows
$2$ to $n + 1$ and columns $2$ to $n + 1$.  The $n \times n$
submatrix $B$ has the form $2V^{\mathrm{T}}V$, where $V$ is the
$m \times n$ matrix
\[
V = \begin{pmatrix}
v_{1} - v_{0} & v_{2} - v_{0} & \cdots & v_{n} - v_{0}
\end{pmatrix}.
\]
Because of the earlier rotation of the simplex,
the last $m - n$ coordinates of each
vector $v_{i} - v_{0}$ are zeroes, and if we take $\widetilde{V}$
to be the first $n$ rows of $V$, then $\widetilde{V}$ is
an $n \times n$ matrix that satisfies
$V^{\mathrm{T}}V = \widetilde{V}^{\mathrm{T}}\widetilde{V}$.
It follows that $B = 2\widetilde{V}^{\mathrm{T}}\widetilde{V}$.
Thus $\det(B) = 2^{n}\det(\widetilde{V})^2 \ge 0$.  Observing that
$\det(\widetilde{V})$ is the signed volume of the parallelepiped
spanned by the vectors that form the columns of $\widetilde{V}$, we
note that $\det(B) > 0$ holds when the columns of
$\widetilde{V}$ are linearly independent, i.e. when the
vertices of the original simplex are affinely independent.

Thus with the assumption that $\tau^{n-1}$ is a fully
$(n-1)$-dimensional simplex, we know that $\det(A) < 0$
when the vertex $u$ lies in either of the open halfspaces
bounded by $\affine(\tau^{n-1})$.  For $u$ outside
$\affine(\tau^{n-1})$, then, we conclude that
$\alpha_{i} = \det(A_{i})/\det(A) > 0$
if and only if $\det(A_{i}) < 0$.  Hence the
simplex $\cone{u}{\tau^{n-1}}$ will be $n$-well-centered
if and only if the coordinates of $u$ satisfy the polynomial
inequality $\det(A_{i}) < 0$.

It remains to show that the equation $\det(A_{i}) = 0$ is a
polynomial in the coordinates of $u$ of degree at most $3$.
To do this we examine the entries of $A_{i}$ that depend on $u$.
All of these entries appear in row $n + 1$ or in column $n + 1$.
At most two of these entries are quadratic in the coordinates
of $u$---the entry at position $(n+1, n+1)$ and the entry at
$(n+1, i+1)$.  (Only one entry is quadratic in the coordinates of
$u$ when $i = n$.)  Every other entry that depends on $u$ is
linear in the coordinates of $u$.
Using $S_{n}$ to denote the group of permutations on $n$ letters,
the determinant of an $n\times n$ matrix $M$ can be written as
\[
\det(M) = \sum_{\pi \in S_{n}} \mathrm{sgn}(\pi) \prod_{j=1}^{n} M_{j\pi(j)},
\]
where $M_{jk}$ stands for the entry in row $j$, column $k$ of matrix $M$,
and $\mathrm{sgn}(\pi)$ is the signum function applied to the permutation.
Considering the structure of matrix $A_{i}$,
we observe that each product in this definition of $\det(A_{i})$
involves at most two terms that depend on $u$, and at most one
of these---the entry selected from row $n + 1$---is quadratic
in the coordinates of $u$.  Thus the determinant is a summation
of terms that are polynomial in the coordinates of $u$ and have
degree at most $3$.

\begin{figure}
\centering
\includegraphics[width=54pt, trim=617pt 178pt 567pt 142pt, clip]
  {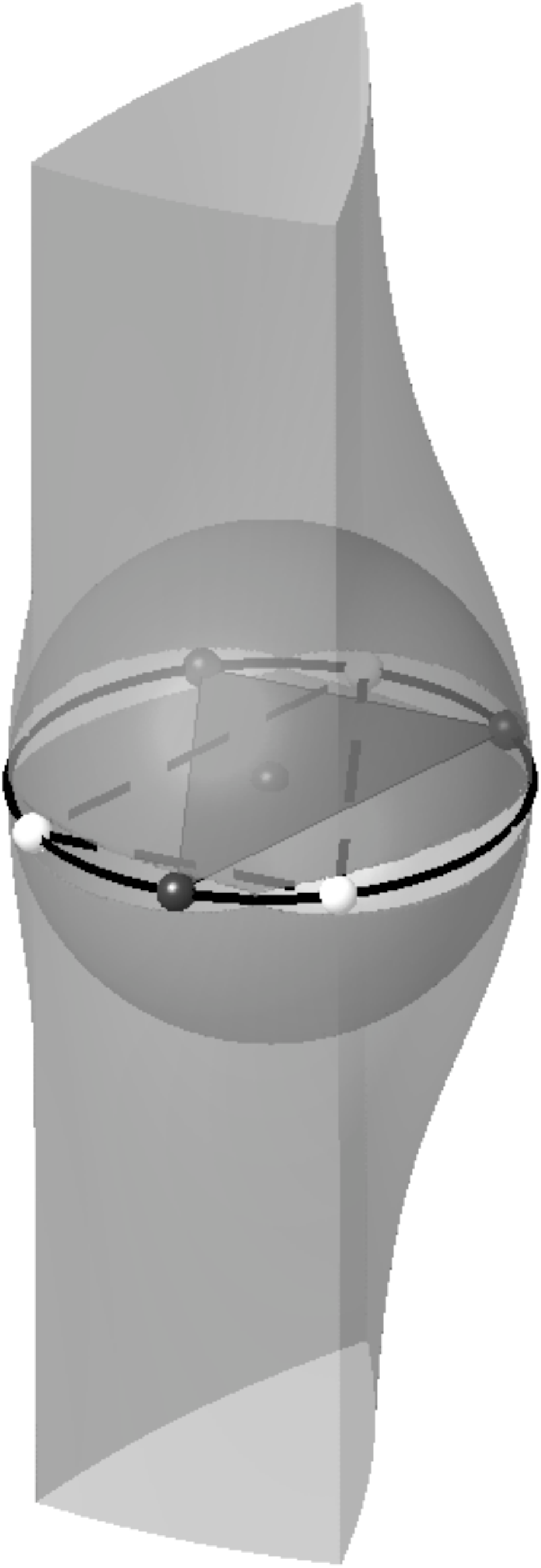}%
\caption{Given a facet $\tau^{n-1}$, the region where the vertex $u$
  may lie to produce an $n$-well-centered simplex $\cone{u}{\tau^{n-1}}$
  is defined by a system of polynomial inequalities.  When $\tau^{n-1}$
  is $(n-1)$-well-centered, so that there are regions related to both
  the necessary and sufficient conditions, the actual region where
  $u$ may lie is somewhere in between the regions defined by the necessary
  Cylinder Condition (Fig.~\ref{fig:nec_cond_region})
  and the sufficient Prism Condition (Fig.~\ref{fig:suf_cond_region}).}
\label{fig:full_region}
\end{figure}

We can also explain this from the perspective of computing
the determinant by expanding it along a row or column.
We will consider a specific example with $i=2$ arising from a
tetrahedron  (dimension $n=3$), but the discussion applies to the
general case.  We have
\[
A_{2} =
\begin{pmatrix}
2\langle v_{0}, v_{0}\rangle & 2\langle v_{0}, v_{1} \rangle
& \langle v_{0}, v_{0}\rangle & 2\langle v_{0}, u\rangle & 1\\
2\langle v_{1}, v_{0}\rangle & 2\langle v_{1}, v_{1} \rangle
& \langle v_{1}, v_{1}\rangle & 2\langle v_{1}, u\rangle & 1\\
2\langle v_{2}, v_{0}\rangle & 2\langle v_{2}, v_{1} \rangle
& \langle v_{2}, v_{2}\rangle & 2\langle v_{2}, u\rangle & 1\\
2\langle u, v_{0}\rangle & 2\langle u, v_{1} \rangle
& \langle u, u\rangle & 2\langle u, u\rangle & 1\\
1 & 1 & 1 & 1 & 0
\end{pmatrix}
\]
for this particular example.
If we compute $\det(A_{i})$ by expanding down column
$n+1$ (column 4, in this case), we find that term
$n+1$ of the summation is a quadratic function
of the coordinates of $u$ multiplied by the determinant of a
submatrix that is constant with respect to $u$.  In our example,
this is the fourth term in the summation,
\[
2\langle u, u\rangle\cdot\det\!\!
\begin{pmatrix}
2\langle v_{0}, v_{0}\rangle & 2\langle v_{0}, v_{1} \rangle
& \langle v_{0}, v_{0}\rangle & 1\\
2\langle v_{1}, v_{0}\rangle & 2\langle v_{1}, v_{1} \rangle
& \langle v_{1}, v_{1}\rangle & 1\\
2\langle v_{2}, v_{0}\rangle & 2\langle v_{2}, v_{1} \rangle
& \langle v_{2}, v_{2}\rangle & 1\\
1 & 1 & 1 & 0
\end{pmatrix}.
\]
The remaining $n+1$ of the terms in the summation are
linear (or constant) functions of $u$ multiplied by a determinant
of some other submatrix of $A_{i}$ that is not constant with respect to
$u$.  For our example, the first term of the summation is
\[
-2\langle v_{0}, u\rangle\cdot\det\!\!
\begin{pmatrix}
2\langle v_{1}, v_{0}\rangle & 2\langle v_{1}, v_{1} \rangle
& \langle v_{1}, v_{1} \rangle & 1\\
2\langle v_{2}, v_{0}\rangle & 2\langle v_{2}, v_{1} \rangle
& \langle v_{2}, v_{2} \rangle & 1\\
2\langle u, v_{0}\rangle & 2\langle u, v_{1} \rangle
& \langle u, u\rangle & 1\\
1 & 1 & 1 & 0
\end{pmatrix}.
\]
Expanding the appropriate row (usually row $n$) of each of these
submatrices in similar fashion, we obtain
a summation of terms that are either linear or
quadratic in the coordinates of $u$ (at most one term is quadratic),
each multiplied by the determinant of a smaller submatrix
that is constant with respect to $u$.

We state the conclusions of the foregoing discussion as
a formal proposition.

\bigskip
\begin{proposition}
\label{prop:necandsufpoly}
Let $\sigma^{n} = \cone{u}{\tau^{n-1}}$ for a fixed facet
$\tau^{n-1}$.  The $n$-simplex $\sigma^{n}$ is $n$-well-centered
if and only if the coordinates of vertex $u$ satisfy the
inequalities $\det(A_{i}) < 0$, which are cubic polynomial
inequalities in the coordinates of $u$.
\end{proposition}
\bigskip

Figure~\ref{fig:full_region} gives a graphical representation of the
precise region where the vertex $u$ may be placed to produce
a $3$-well-centered
tetrahedron $\cone{u}{\tau^{n-1}}$.  The facet $\tau^{n-1}$ used
in Fig.~\ref{fig:full_region} is the same facet used to illustrate
the necessary condition for a tetrahedron in
Fig.~\ref{fig:nec_cond_region} and the sufficient condition
for a tetrahedron in Fig.~\ref{fig:suf_cond_region}, so
readers can see for this specific case how the full region compares
to the regions defined by the necessary condition and the sufficient
condition.  The facet $\tau^{n-1}$ along with its circumcircle and
the reflection of $\tau^{n-1}$ through $c(\tau^{n-1})$ are shown in
Fig.~\ref{fig:full_region} to aid this comparison.
It should also be noted that Fig.~\ref{fig:full_region} was generated
using MATLAB's isosurface function and evaluations of the polynomial
inequalities on a finite grid, so the graphical representation has
some slight imperfections.
For instance, the entire circumcircle of
$\tau^{n-1}$ lies in the boundary of the region
even though in Fig.~\ref{fig:full_region} it appears that there is a
small gap above and below $\affine(\tau^{n-1})$.

\section{Local Combinatorial Properties of
  $3$-Well-Centered Tetrahedral Meshes}
\label{sec:combinatorial3wccond}

The geometric properties of the $n$-well-centered
$n$-simplex discussed in Sec.~\ref{sec:characterize}
have implications for the combinatorial properties
of well-centered meshes.  As a simple motivating
example we consider the $2$-dimensional case of a
triangle mesh in the plane.  If $v$ is a vertex
interior to this mesh and there are fewer than
five edges incident to $v$, then some angle incident
to $v$ has measure $\pi/2$ radians or larger.  Thus
the mesh has a nonacute triangle.  This geometric
observation can be restated as a combinatorial property
of $2$-well-centered (i.e., acute) triangle meshes.
Namely, there are at least five edges incident
to every interior vertex of an acute triangle mesh
in $\RR^{2}$.  This well-known fact is a key ingredient in
the generation of $2$-well-centered triangle meshes through
optimization of vertex coordinates; the mesh must satisfy this
combinatorial condition at every interior vertex if optimizing the
vertex coordinates is to have any hope of finding an acute mesh.

Similarly, tetrahedral meshes in $\RR^{3}$
that are $2$-well-centered or $3$-well-centered
must satisfy certain local mesh connectivity conditions.
These combinatorial conditions, which are key to creating
well-centered tetrahedral meshes, are analyzed in the
next two sections.
This section develops some of the combinatorial properties of
$3$-well-centered tetrahedral meshes, and the next section
examines combinatorial properties of $2$-well-centered
tetrahedral meshes.

The combinatorial properties of tetrahedral meshes in $\RR^{3}$
are more complex than the analogous properties for triangle meshes
in $\RR^{2}$.  In a triangle mesh in $\RR^{2}$, the link
of an interior vertex is a set of edges that form a cycle
around the vertex, i.e., a triangulation of a topological
circle~($S^{1}$).  The number of edges incident to the interior
vertex, which is the number of vertices on the cycle, completely
characterizes the neighborhood of the vertex.
In tetrahedral meshes in $\RR^3$, on the other hand,
the link of an interior vertex
is a triangulation of a topological sphere $S^{2}$.
Thus the number of edges incident to the vertex does not
completely characterize the neighborhood of the vertex.
We do, however, prove necessary conditions on the number of
edges that must be incident to an interior vertex in a
tetrahedral mesh in $\RR^{3}$ in order for the mesh
to be $3$-well-centered, $2$-well-centered, or
completely well-centered.  We also show that there is
no sufficient condition in terms of the number of edges
incident to an interior vertex.

Much of the discussion in Secs.~\ref{sec:combinatorial3wccond}
and~\ref{sec:combinatorial2wctetcond}, then, is phrased in
terms of the link of an interior vertex.  For a tetrahedral
mesh in $\RR^{3}$, this is a triangulation of $S^{2}$, which
corresponds to a planar triangulation in a graph theoretic
sense.  We try to avoid the term planar triangulation to
prevent possible confusion with triangle meshes in $\RR^{2}$.

We begin with two results that apply in arbitrary dimension.
The first lemma generalizes the following statement about
planar triangle meshes, using the Cylinder Condition
(Proposition~\ref{prop:char_nec}) to relate geometry to combinatorics.
If a planar triangle $\sigma^{2} = [abc]$ is subdivided into
three triangles by adding a vertex $u$ interior to $\sigma^{2}$
and adding edges $[ua]$, $[ub]$, and $[uc]$ to obtain
$\cone{u}{\boundary{\sigma^{2}}}$, then at most one of the
three triangles $[abu]$, $[bcu]$, and $[cau]$ in
$\cone{u}{\boundary{\sigma^{2}}}$ is an acute triangle.  The
main idea of the proof of Lemma~\ref{lemma:simplex_partition}
is illustrated by Fig.~\ref{fig:smplxprttn}.

\bigskip

\begin{lemma}
\label{lemma:simplex_partition}
For $n \ge 2$, let $\sigma^{n} = [v_{0}v_{1}\ldots v_{n}]$ have facets
$\tau^{n-1}_{0},\tau^{n-1}_{1},\ldots,\tau^{n-1}_{n}.$  If
$u$ is a point lying in $\sigma^{n}$,
then at most one of the $n$-simplices of
$\cone{u}{\boundary{\sigma^{n}}}$, i.e., at most one of the simplices
$\cone{u}{\tau^{n-1}_{0}}, \cone{u}{\tau^{n-1}_{1}},
\ldots, \cone{u}{\tau^{n-1}_{n}}$,
is an $n$-well-centered $n$-simplex.
\end{lemma}
\begin{proof}

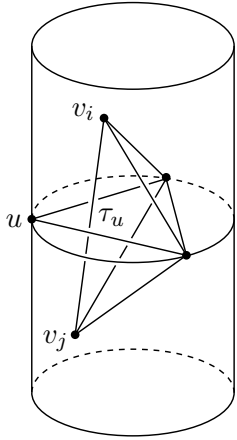
\begin{figure}
\centering
\includegraphics[width=80pt, trim=44pt 593pt 436pt 9pt, clip]
  {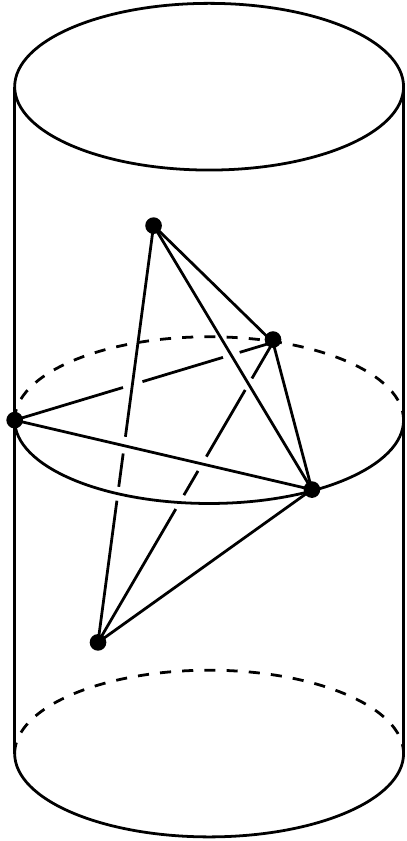}%
\begin{picture}(0,0)
\put(-87,80){$u$}
\put(-62,122){$v_{i}$}
\put(-73,36){$v_{j}$}
\put(-52,82){$\tau_{u}$}
\end{picture}
\caption{If $v_{i}$ and $v_{j}$ are both
  contained in the interior of the solid right spherical
  cylinder over the circumcircle
  of $\tau_{u}$,
  then $u$ lies outside the simplex formed from $v_{i}$, $v_{j}$,
  and the other vertices of $\tau_{u}$.}
\label{fig:smplxprttn}
\end{figure}

It suffices to prove the statement when $u$ is in the
interior of $\sigma^{n}$.  Indeed, if $u$ is on the
boundary of $\sigma^{n}$ and two or more of the simplices
$\cone{u}{\tau^{n-1}_{i}}$ are $n$-well-centered, then
we can slightly perturb $u$ into the interior and
obtain a point $u \in \interior(\sigma^{n})$
with at least two $n$-well-centered $n$-simplices.
Thus we assume that $u \in \interior(\sigma^{n})$.

Let $\tau^{n-1}_{i}$ and $\tau^{n-1}_{j}$ be two
distinct facets of $\sigma^{n}$.
Then $\cone{u}{\tau^{n-1}_{i}}$ and $\cone{u}{\tau^{n-1}_{j}}$
are $n$-sim\-pli\-ces, and $\tau^{n-1}_{i} \cap \tau^{n-1}_{j}$
is an $(n-2)$-dimensional face of $\sigma^{n}$.
The face $\tau^{n-1}_{i} \cap \tau^{n-1}_{j}$ is incident to
all but two of the
vertices of $\sigma^{n}$, the two vertices $v_{i}$ and $v_{j}$.
(Recall that $v_{i}$ is opposite $\tau^{n-1}_{i}$ and
$v_{j}$ is opposite $\tau^{n-1}_{j}$.)
Notice that $\cone{u}{\tau^{n-1}_{i}}$
and $\cone{u}{\tau^{n-1}_{j}}$ have a common facet, the
$(n-1)$-simplex
$\tau^{n-1}_{u} := \cone{u}{(\tau^{n-1}_{i} \cap \tau^{n-1}_{j})}$.
We let $T \subset \affine(\sigma^{n})$ be the solid right
spherical cylinder over the
circumball of $\tau^{n-1}_{u}$.

Assume towards contradiction that $\cone{u}{\tau^{n-1}_{i}}$ and
$\cone{u}{\tau^{n-1}_{j}}$ are both $n$-well-centered.
By the Cylinder Condition (Proposition~\ref{prop:char_nec}), both
$v_{i}$ and $v_{j}$ lie in $\interior(T)$.
Now $T$ is a convex set, and all the vertices of $\sigma^{n}$
lie in $T$, so $\sigma^{n} \subset T$.
On the other hand, $u$ lies on the circumsphere of
$\tau^{n-1}_{u}$, so $u \in \boundary{T}$.
Thus $u \notin \interior(\sigma^{n}) \subset \interior(T)$,
contradicting the assumption we made in the first
paragraph of the proof.  We conclude that at most
one of $\cone{u}{\tau^{n-1}_{i}}$,
$\cone{u}{\tau^{n-1}_{j}}$ is $n$-well-centered.
\end{proof}

\bigskip

The next theorem shows that Lemma~\ref{lemma:simplex_partition}
has implications for the local combinatorial properties of
$n$-well-centered meshes.  The theorem is stated using
the language of simplicial complexes.  We say that a vertex
$u$ is an {\emph{interior vertex}} in an $n$-dimensional
simplicial complex embedded in $\RR^{n}$ if
(the underlying space of) $\link~u$
is homeomorphic to $S^{n-1}$, the sphere of dimension $n - 1$.
Thus the closed star of $u$ is homeomorphic to an $n$-dimensional
ball in $\RR^{n}$, and the point $v$ lies in the interior
of the ball in the standard topology on $\RR^{n}$.

When we speak of an abstract simplicial complex $K$ we make
an important distinction between an embedding of $K$ and
a geometric realization of $K$.  An {\emph{embedding}} of $K$
is an assignment of coordinates in $\RR^{n}$ to the vertices
of $K$ such that $K$ is a simplicial complex in $\RR^{n}$
with vertices at the specified locations.  By a
{\emph{geometric realization}} of $K$ we mean merely some
assignment of coordinates in $\RR^{n}$ to the vertices of $K$.
Thus in a geometric realization of $K$ in $\RR^{n}$, it is
possible for $K$ to have self-intersections.
Figure~\ref{fig:invexample}, which is related to the
proof of Theorem~\ref{thm:oneringcond}, illustrates
the distinction between these two terms.

\bigskip

\begin{theorem}[One-Ring Necessary Condition]
\label{thm:oneringcond}
Let $u$ be an interior vertex of an $n$-dimensional simplicial
complex $M$ (e.g., a mesh) embedded in $\RR^{n}$, and set
$L = \link~u$.  If there exists an abstract finite $n$-dimensional
simplicial complex $K$ such that
\begin{enumerate}[(i)]
\item $K$ is an $n$-manifold complex (with boundary)
\item $\boundary{K}$ is isomorphic to $L$, and
\item for every $n$-simplex $\sigma^{n} \in K$, there are at least
  two $(n-1)$-simplices in $\boundary{\sigma^{n}} \cap L$,
\end{enumerate}
then $\cone{u}{L}$ is not $n$-well-centered.
\end{theorem}
\begin{proof}
We first observe that every vertex of~$K$ must also
be a vertex of~$L$.  By assumption (i), every simplex of
$K$ is a face of some $n$-simplex of $K$, so
if $K$ had a vertex~$v$ not in~$L$,
then there would be some $n$-dimensional
simplex~$\sigma^{n}~\in~K$ incident to~$v$, and
$\sigma^{n}$ would have only one $(n-1)$-dimensional face
not incident to~$v$.  Since $v \notin L$, 
it follows that $\boundary{\sigma^{n}} \cap L$
would contain at most one $(n - 1)$-simplex, and
(iii) would not be satisfied.

\begin{figure}
\centering
\begin{minipage}[c]{120pt}
\includegraphics[width=120pt, trim=0pt 0pt 0pt 0pt, clip]
  {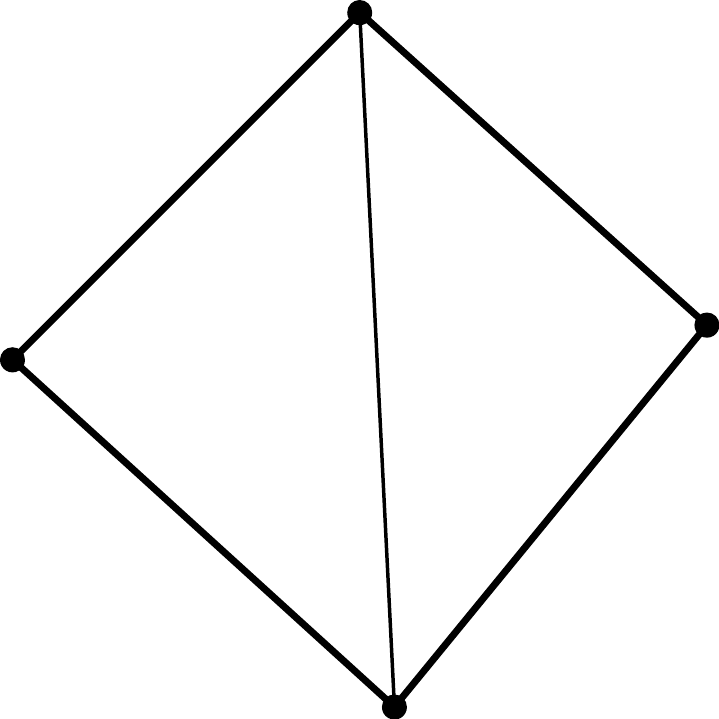}%
\begin{picture}(0, 0)
\put(-112, 58){$v_{0}$}
\put(-47, 2){$v_{1}$}
\put(-54, 115){$v_{2}$}
\put(-17, 64){$v_{3}$}
\end{picture}
\end{minipage}%
\hspace{50pt}%
\begin{minipage}[c]{70pt}
\hspace{3pt}%
\includegraphics[width=67pt, trim=0pt 0pt 0pt 0pt, clip]
  {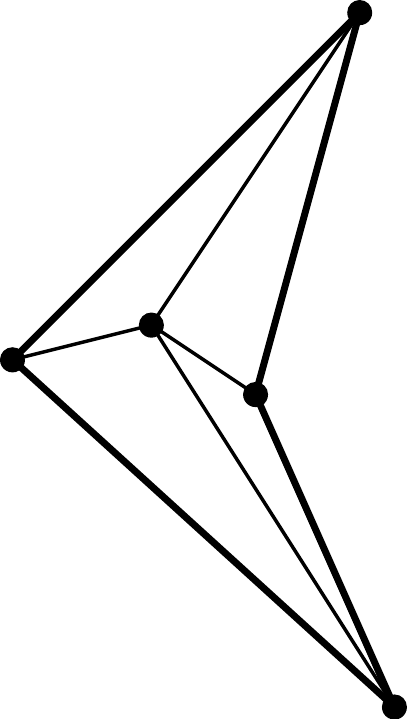}%
\begin{picture}(0, 0)
\put(-70, 49){$v_{0}$}
\put(2, 3){$v_{1}$}
\put(-3, 114){$v_{2}$}
\put(-21, 51){$v_{3}$}
\put(-48, 56){$u$}
\end{picture}
\end{minipage}%
\hspace{50pt}%
\begin{minipage}[c]{70pt}
\includegraphics[width=67pt, trim=0pt 0pt 0pt 0pt, clip]
  {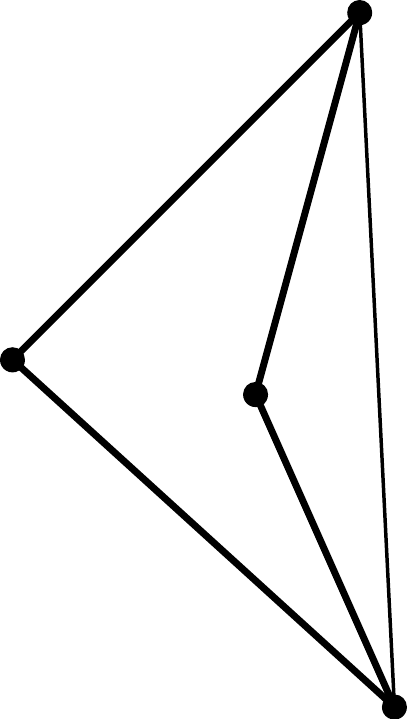}%
\begin{picture}(0, 0)
\put(-70, 49){$v_{0}$}
\put(2, 3){$v_{1}$}
\put(-3, 114){$v_{2}$}
\put(-21, 51){$v_{3}$}
\end{picture}
\end{minipage}\\[8pt]
\hspace{15pt}%
\begin{minipage}[c]{120pt}
\centering
Abstract Complex $K$, $\boundary{K} = L$
\end{minipage}%
\hspace{50pt}%
\begin{minipage}[c]{70pt}
\centering
Embedding of $M = \cone{u}{L}$
\end{minipage}%
\hspace{35pt}%
\begin{minipage}[c]{100pt}
\centering
Induced Geometric Realization of $K$
\end{minipage}%
\caption{In $\RR^{2}$ the simplicial complex $K$ consisting of
  two triangles $[v_{0} v_{1} v_{2}]$ and $[v_{1} v_{3} v_{2}]$
  and their faces (left) satisfies the hypothesis of
  Theorem~\ref{thm:oneringcond} for the mesh $M = \cone{u}{L}$
  embedded in $\RR^{2}$ (center), so $M$ is not
  $2$-well-centered.  The embedding of $M$ in~$\RR^{2}$ induces
  a geometric realization of $K$ into $\RR^{2}$ (right).  The
  geometric realization of $K$ is not an embedding in this case,
  since $[v_{1} v_{3} v_{2}]$ is inverted here.  The particular
  embedding of $M$ does not affect the existence of the abstract
  complex $K$ in Theorem~\ref{thm:oneringcond}; there is
  no embedding of $M$ that is $2$-well-centered.}
\label{fig:invexample}
\end{figure}
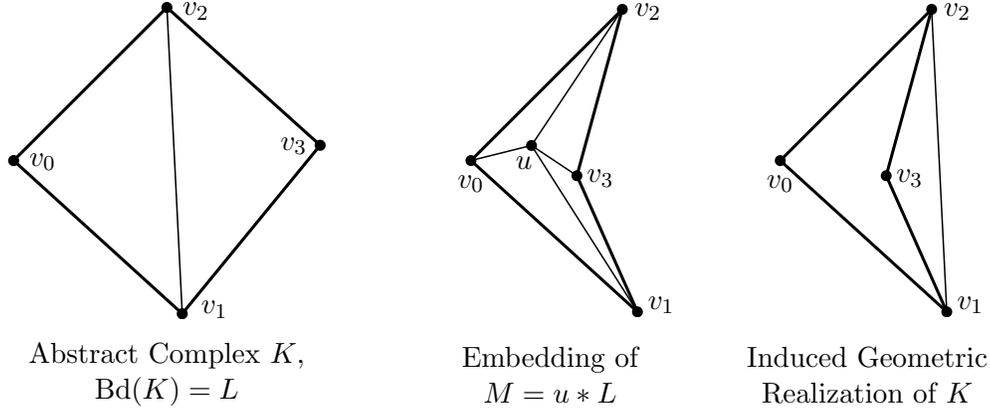

The embedding of $M$ in $\RR^{n}$ includes an embedding of
$\cone{u}{L}$ in $\RR^{n}$.  Since every vertex of $K$ is a
vertex of $L$, this embedding of $\cone{u}{L}$ in $\RR^{n}$
induces a geometric realization of $K$ into $\RR^{n}$.
(As shown in Fig.~\ref{fig:invexample}, the geometric
realization might not be an embedding.)

We have an embedding of the simplicial complex $\cone{u}{L}$ in
$\RR^{n}$.  Since it is an embedding, each $n$-dimensional
simplex is a fully $n$-dimensional geometric object, and
we have consistent orientation.  Moreover, $L$ is
star-shaped with respect to $u$.  We claim that by
(i) and (ii) this implies that there is some simplex
in the induced geometric realization of $K$ that contains
the point $u$ (possibly on its boundary).  We return
to this claim in a moment, but first we show how
this completes the proof.

Fix a simplex $\sigma^{n} \in K$ that contains $u$.
Now consider the $n$-simplices of
$\cone{u}{\boundary{\sigma^{n}}}$.  By assumption (iii)
of the hypothesis, at least two of these simplices have
a facet in $L$.  Each simplex of $\cone{u}{\boundary{\sigma^{n}}}$
with a facet in $L$ is a member of $\cone{u}{L}$, and by
Lemma~\ref{lemma:simplex_partition} at most
one of these simplices is $n$-well-centered.
We conclude that at least one of the simplices
of $\cone{u}{L} \subseteq M$ is not $n$-well-centered.

Now we prove the claim that there is a simplex
of the geometric realization of $K$ that contains
the point $u$.  Choose a line $\ell$ through $u$
in general position.  General position here means
that $\ell$ does not intersect any face of $K$
of dimension less than $n - 1$.  Such an $\ell$
can be chosen unless $u$ itself lies on a simplex $\rho^{k}$
of $K$ of dimension $k < n - 1$, and in that
case we are done, since there is some $\sigma^{n} \succ \rho^{k}$
that contains $u$.

Since $\cone{u}{L}$ is a simplicial complex
and $L$ is star-shaped with respect to $u$,
$\ell$ intersects exactly two simplices of $L$,
each of dimension $n - 1$, and the intersection
points are in opposite directions from $u$ along $\ell$.
For reference, we designate a $+$ and a $-$ direction
and name facet $\tau^{n-1}_{+}$ (resp.~$\tau^{n-1}_{-}$) as
the facet of $L$ intersected by $\ell$ in the $+$ ($-$) direction
from $u$.  Starting from $\tau^{n-1}_{+}$ we describe a walk
along $\ell$ through $n$-simplices and $(n-1)$-simplices
of the geometric realization of $K$ that ends at $\tau^{n-1}_{-}$.
By continuity of this walk and $\tau^{n-1}_{+}$, $\tau^{n-1}_{-}$
in opposite directions from $u$, there must be some $n$-simplex
in the geometric realization of $K$ that contains $u$.

The walk is as follows.  Since $K$ is a manifold with
a boundary and $\tau^{n-1}_{+}$ is on the boundary, there
is a unique $\sigma^{n}_{1}$ incident to
$\tau^{n-1}_{0} := \tau^{n-1}_{+}$.
Then for a given $\sigma^{n}_{i}$ the walk is on $\ell$
at $\tau^{n-1}_{i-1}$, and $\ell$ intersects some unique
second facet of $\sigma^{n}_{i}$, which we name $\tau^{n-1}_{i}$.
As long as $\tau^{n-1}_{i} \ne \tau^{n-1}_{-}$, we are
not on the boundary of $K$, so (since $K$ is manifold)
there are exactly two $n$-dimensional
simplices incident to $\tau^{n-1}_{i}$.  One of these is
$\sigma^{n}_{i}$, and the other we name $\sigma^{n}_{i+1}$.
Since $K$ is a manifold complex, the sequence $\tau^{n-1}_{i}$
has no repetitions and must eventually end at $\tau^{n-1}_{-}$.
(The $\sigma_i^{n}$ in the sequence may flip 
back and forth in orientation, which corresponds to
the walk going back and forth along $\ell$.)
\end{proof}

\bigskip

It is worth noting that the existence of the abstract simplicial
complex $K$ has no dependence on the particular embedding
of $M$ in $\RR^{n}$.  Theorem~\ref{thm:oneringcond} is
really a combinatorial statement, and we can use it to
show that a particular {\emph{abstract}} simplicial
complex $L = \boundary{K}$ cannot appear as the link
of an interior vertex in an $n$-well-centered mesh
embedded in $\RR^{n}$.

The case $n = 3$ is of particular interest.  Using
the One-Ring Necessary Condition of Theorem~\ref{thm:oneringcond}
it is fairly easy to establish a tight lower bound on the
number of edges incident to a vertex in a $3$-well-centered
tetrahedral mesh embedded in $\RR^{3}$.

\bigskip

\begin{corollary}
\label{cor:min3wc}
Let $M$ be a $3$-well-centered tetrahedral mesh embedded
in $\RR^{3}$.  For every vertex $u$
interior to $M$, at least $7$ edges of $M$
are incident to $u$.
\end{corollary}
\begin{proof}
Britton and Dunitz have assembled a catalog of all polyhedra
with at most $8$ vertices, which includes all the triangulations
of $S^{2}$ with at most $8$ vertices \cite{BrDu1973}.
By Theorem~\ref{thm:oneringcond} it suffices
to show that each such triangulation $L$ of $S^{2}$ with at
most $6$ vertices has a corresponding tetrahedral
complex $K$ such that each tetrahedron of $K$ has at least
two facets in common with $L$.

\begin{figure}
\centering
\begin{minipage}[c]{80pt}
\includegraphics[width=80pt, trim=0pt 0pt 0pt 0pt, clip]
  {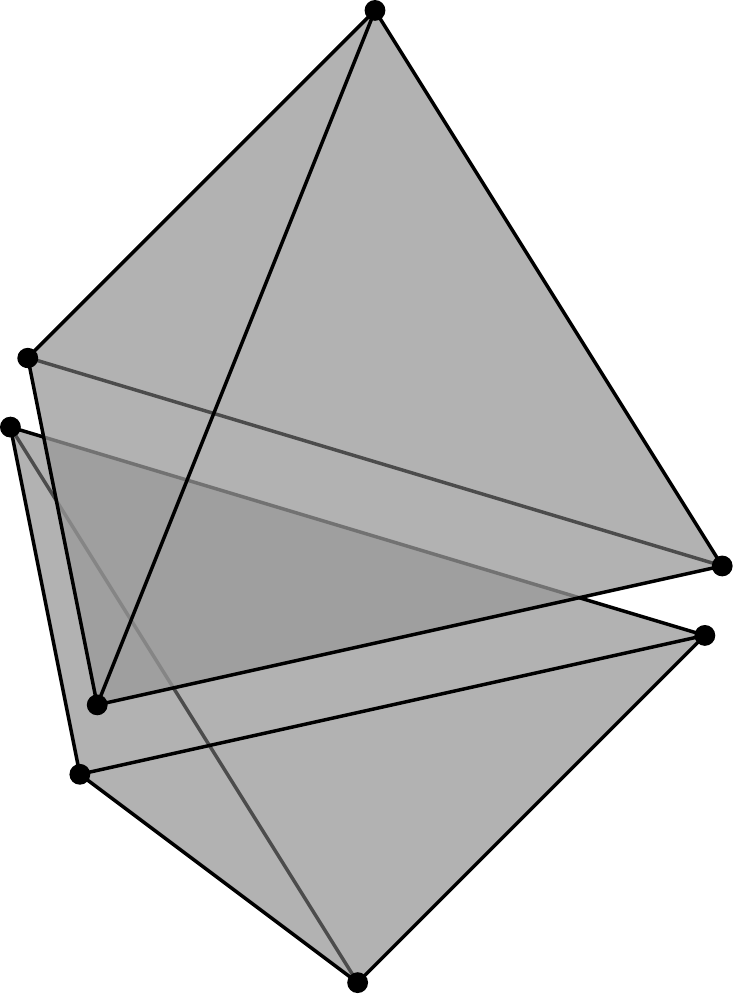}
\end{minipage}%
\hspace{30pt}%
\begin{minipage}[c]{80pt}
\includegraphics[width=80pt, trim=0pt 0pt 0pt 0pt, clip]
  {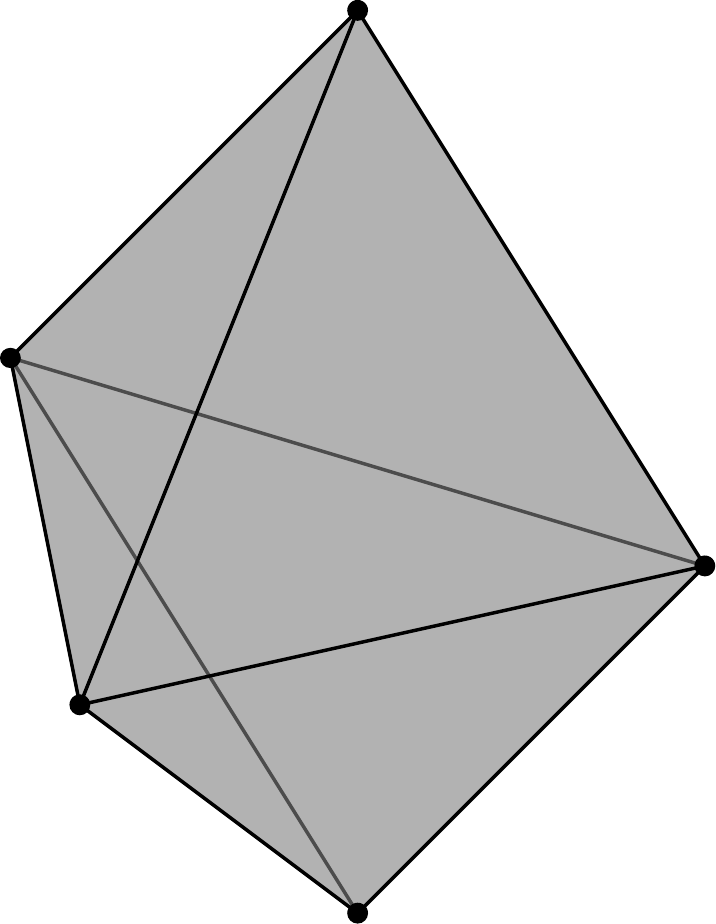}
\end{minipage}%
\hspace{30pt}%
\begin{minipage}[c]{110pt}
\includegraphics[width=110pt, trim=0pt 0pt 0pt 0pt, clip]
  {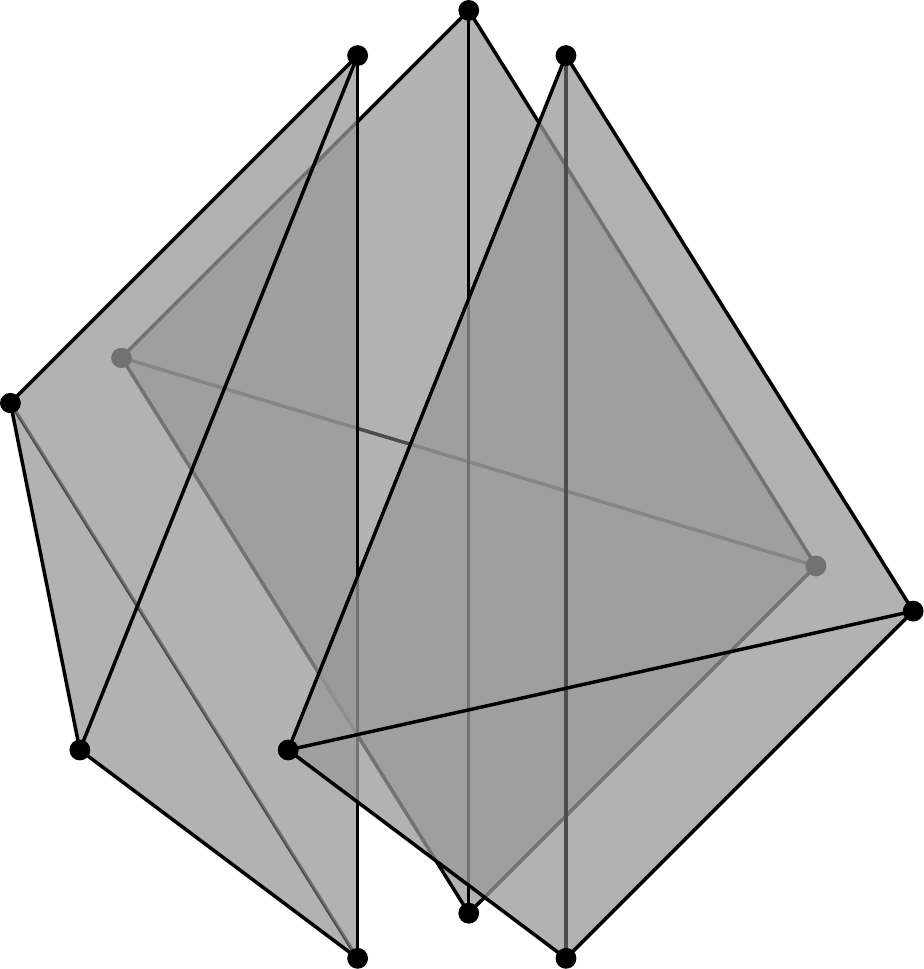}
\end{minipage}%
\caption{There is only one triangulation of $S^2$ with $5$
  vertices, and it has two corresponding tetrahedral complexes
  such that each tetrahedron has at least two facets in common with
  the triangulation.}
\label{fig:vlnk44433}
\end{figure}

There is only one triangulation of $S^{2}$ with $4$
vertices---the boundary of a tetrahedron.  The
corresponding tetrahedral complex is that single
tetrahedron.

There is also only one triangulation of $S^{2}$ with $5$
vertices.  This triangulation is shown in Fig.~\ref{fig:vlnk44433}
along with two corresponding tetrahedral complexes.  Either
complex certifies that the triangulation cannot be the
link of any vertex in a $3$-well-centered mesh.

\begin{figure}
\centering
\begin{minipage}[c]{80pt}
\includegraphics[width=80pt, trim=0pt 0pt 0pt 0pt, clip]
  {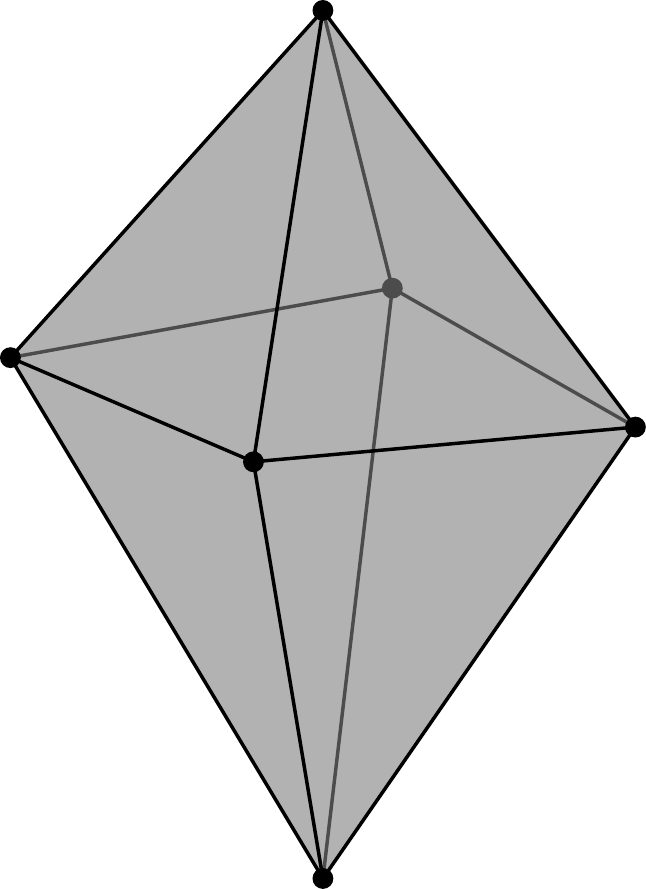}
\end{minipage}%
\hspace{30pt}%
\begin{minipage}[c]{120pt}
\includegraphics[width=120pt, trim=0pt 0pt 0pt 0pt, clip]
  {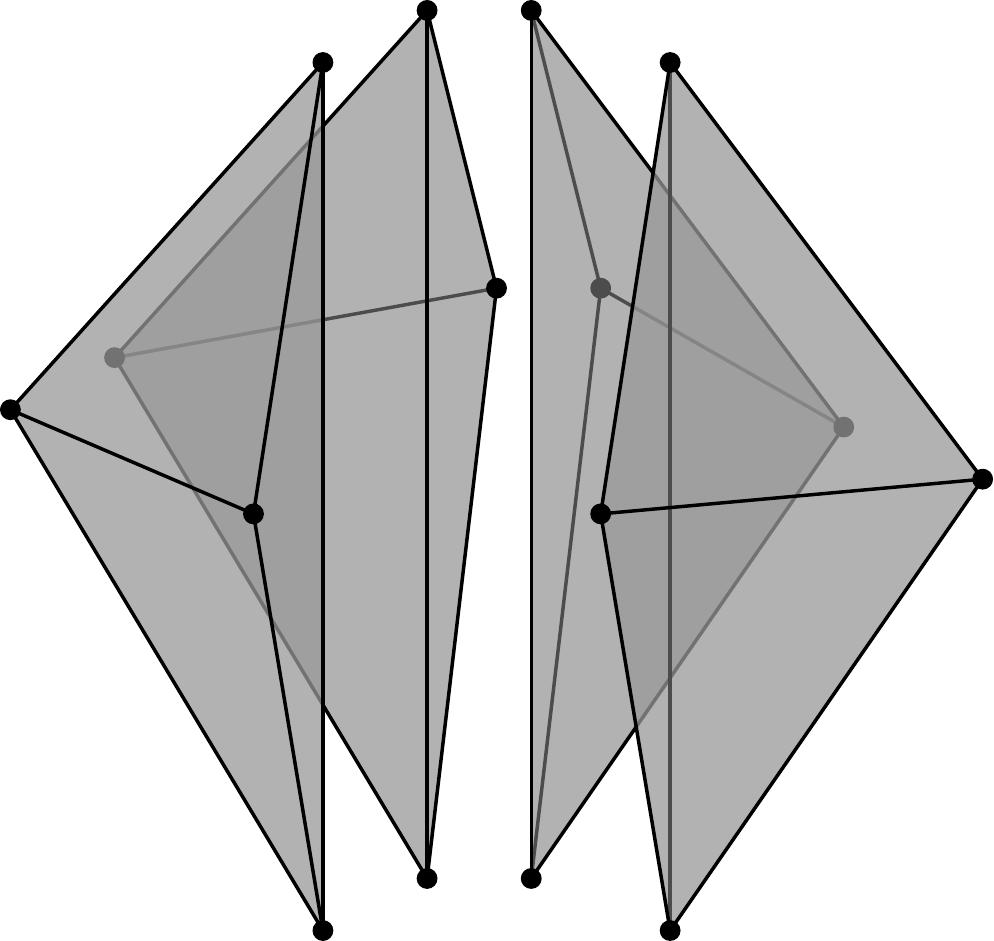}
\end{minipage}%
\caption{For one of the triangulations of $S^2$ with $6$
  vertices, each vertex has exactly four neighbors.  There
  is a tetrahedral complex consisting of four tetrahedra such
  that each tetrahedron has two facets in common with
  this triangulation of $S^{2}$.}
\label{fig:vlnk444444}
\end{figure}

\begin{figure}
\centering
\begin{minipage}[c]{80pt}
\includegraphics[width=80pt, trim=0pt 0pt 0pt 0pt, clip]
  {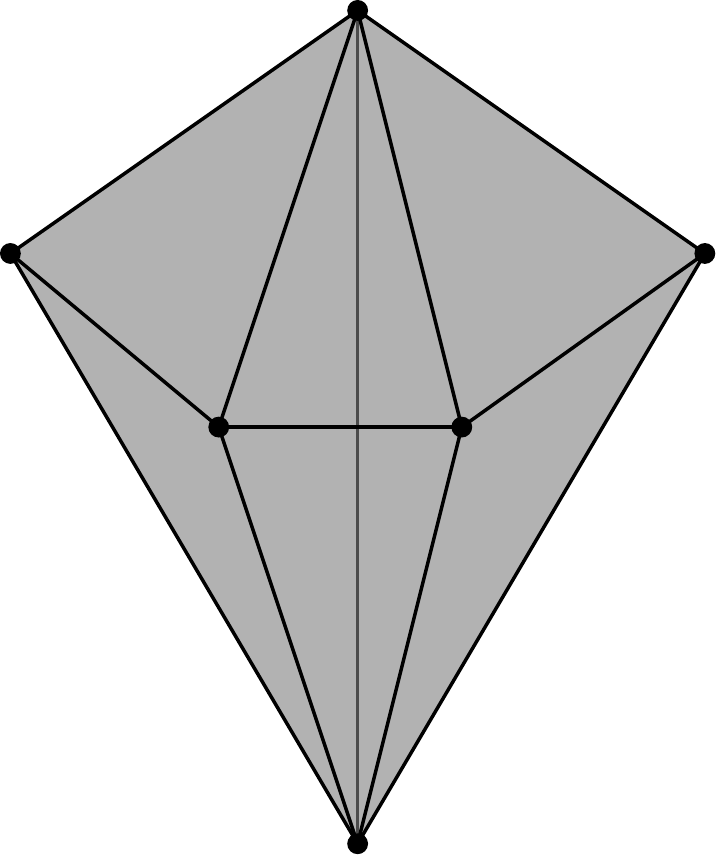}
\end{minipage}%
\hspace{30pt}%
\begin{minipage}[c]{100pt}
\includegraphics[width=100pt, trim=0pt 0pt 0pt 0pt, clip]
  {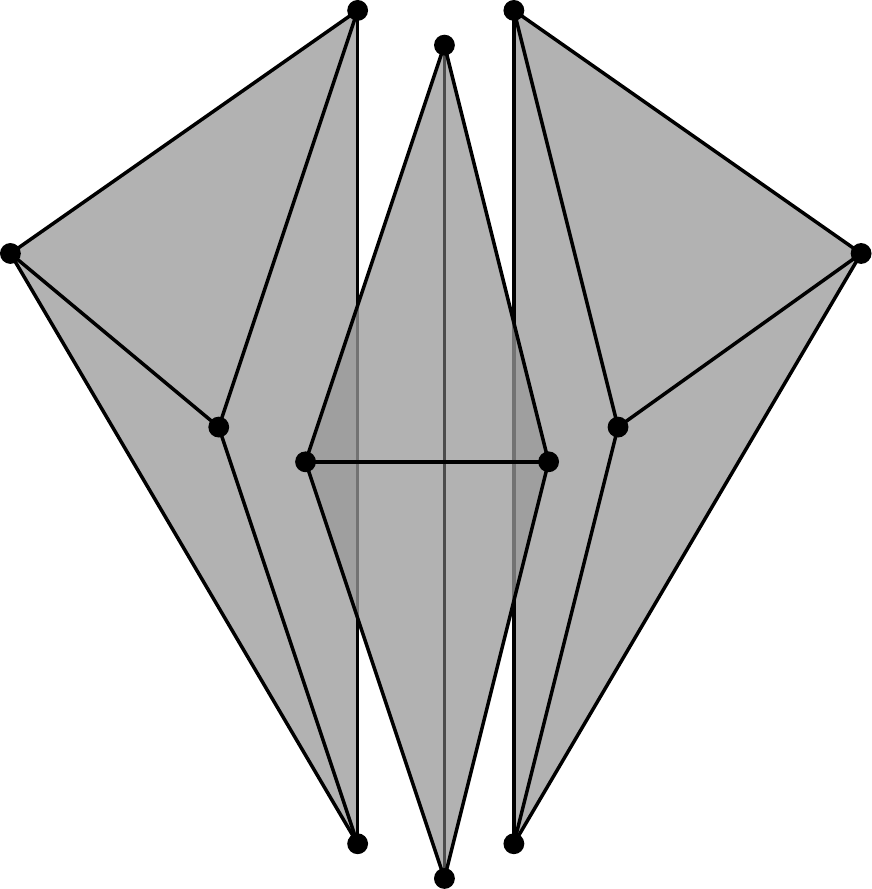}
\end{minipage}%
\caption{In the other triangulation of $S^2$ with $6$
  vertices, the degree list is $(5, 5, 4, 4, 3, 3)$.
  This triangulation of $S^{2}$ also has a corresponding
  tetrahedral complex such that each tetrahedron has at
  least two facets in common with the triangulation.}
\label{fig:vlnk554433}
\end{figure}

For six vertices there are two nonisomorphic triangulations
of $S^{2}$.  The first is shown in Fig.~\ref{fig:vlnk444444}
along with its corresponding tetrahedral complex.  The
second is drawn in Fig.~\ref{fig:vlnk554433} along with its
corresponding tetrahedral complex.
\end{proof}

\bigskip

\begin{figure}
\centering
\begin{minipage}[c]{150pt}
\centering
\includegraphics[width = 70pt, trim = 782pt 138pt 655pt 241pt, clip]%
  {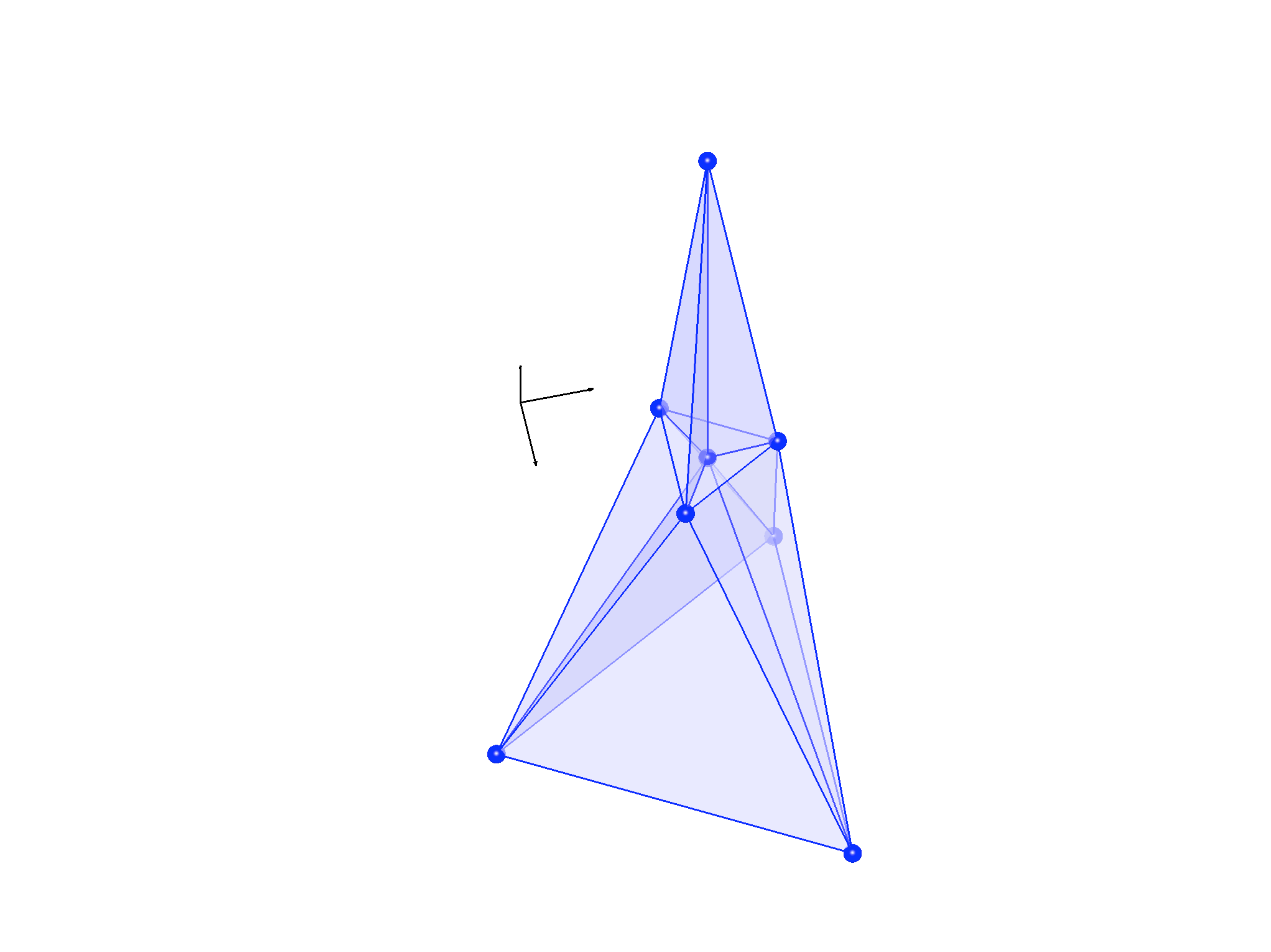}%
\begin{picture}(0, 0)
\put(-68, 75){\scriptsize{$x$}}
\put(-58, 81){\scriptsize{$y$}}
\put(-61.7, 91){\scriptsize{$z$}}
\end{picture}
\end{minipage}%
\hspace{50pt}%
\begin{minipage}[c]{153pt}
\settowidth{\dotlen}{.}
\begin{tabular}{| r @{.} l | r @{.} l | r @{.} l |}
\hline
\multicolumn{2}{|c|}{$x$}
    & \multicolumn{2}{c|}{$y$} & \multicolumn{2}{c|}{$z$}\\
\hline
\multicolumn{1}{|r @{\hspace{\dotlen}}}{$0$} & &
  \multicolumn{1}{r @{\hspace{\dotlen}}}{$0$} & &
  \multicolumn{1}{r @{\hspace{\dotlen}}}{$0$} & \\
\hline
\multicolumn{1}{|r @{\hspace{\dotlen}}}{$0$} & &
  \multicolumn{1}{r @{\hspace{\dotlen}}}{$0$} & &
  \multicolumn{1}{r @{\hspace{\dotlen}}}{$1$} & \\
\hline
$-0$ & $1041$   &   $-0$ & $0601$   &   $0$ & $0117$\\
\hline
$0$ & $1041$   &   $-0$ & $0601$   &   $0$ & $0117$\\
\hline
\multicolumn{1}{|r @{\hspace{\dotlen}}}{$0$} &
    &   $0$ & $1202$   &   $0$ & $0117$\\
\hline
\multicolumn{1}{|r @{\hspace{\dotlen}}}{$0$} &
    &   $-0$ & $3622$   &   $-0$ & $8656$\\
\hline
$0$ & $3137$   &   $0$ & $1811$   &   $-0$ & $8656$\\
\hline
$-0$ & $3137$   &   $0$ & $1811$   &   $-0$ & $8656$\\
\hline
\end{tabular}
\end{minipage}
\caption{A $3$-well-centered mesh with an interior vertex $u$
  such that $\link~u$ has seven vertices and
  degree list $(5,5,5,4,4,4,3)$.  The vertex coordinates
  are listed in the table at right; vertex $u$ is at the origin.
}
\label{fig:onefreev_10/bestrndsym}
\end{figure}

\begin{figure}
\centering
\begin{minipage}[c]{150pt}
\centering
\includegraphics[width = 150pt, trim = 625pt 567pt 137pt 209pt, clip]%
  {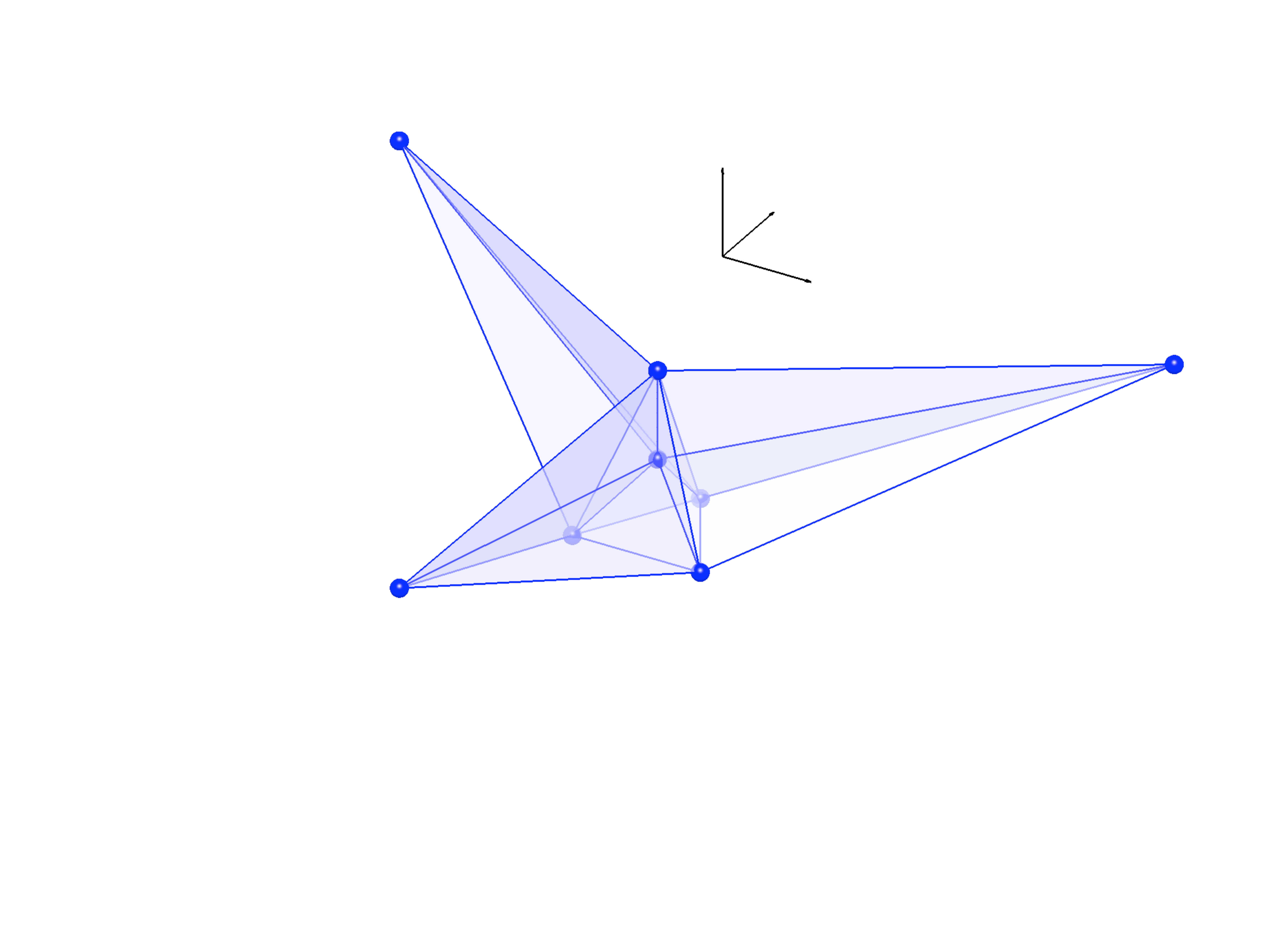}%
\begin{picture}(0, 0)
\put(-83, 56){\scriptsize{$x$}}
\put(-77, 67){\scriptsize{$y$}}
\put(-93, 72){\scriptsize{$z$}}
\end{picture}
\end{minipage}%
\hspace{50pt}%
\begin{minipage}[c]{153pt}
\settowidth{\dotlen}{.}
\begin{tabular}{| r @{.} l | r @{.} l | r @{.} l |}
\hline
\multicolumn{2}{|c|}{$x$}
    & \multicolumn{2}{c|}{$y$} & \multicolumn{2}{c|}{$z$}\\
\hline
\multicolumn{1}{|r @{\hspace{\dotlen}}}{$0$} & &
  \multicolumn{1}{r @{\hspace{\dotlen}}}{$0$} & &
  \multicolumn{1}{r @{\hspace{\dotlen}}}{$0$} & \\
\hline
\multicolumn{1}{|r @{\hspace{\dotlen}}}{$0$} & &
  \multicolumn{1}{r @{\hspace{\dotlen}}}{$0$} & &
  \multicolumn{1}{r @{\hspace{\dotlen}}}{$1$} & \\
\hline
\multicolumn{1}{|r @{\hspace{\dotlen}}}{$0$} &
    &   $0$ & $8334$   &   $-0$ & $8588$\\
\hline
$-0$ & $7217$   &   $-0$ & $4167$   &   $-0$ & $8588$\\
\hline
$0$ & $7217$   &   $-0$ & $4167$   &   $-0$ & $8588$\\
\hline
\multicolumn{1}{|r @{\hspace{\dotlen}}}{$0$} &
    &   $-5$ & $0494$   &   $1$ & $0696$\\
\hline
$4$ & $3729$   &   $2$ & $5247$   &   $1$ & $0696$\\
\hline
$-4$ & $3729$   &   $2$ & $5247$   &   $1$ & $0696$\\
\hline
\end{tabular}
\end{minipage}%
\caption{A $3$-well-centered mesh with an interior vertex $u$
  such that $\link~u$ has seven vertices and
  degree list $(6,5,5,5,3,3,3)$.  The vertex coordinates
  are listed in the table at right; vertex $u$ is at the origin.
}
\label{fig:onefreevver2_10/goodrndsym}
\end{figure}

When there are $m \ge 7$ vertices, there exist triangulations
$L$ of $S^{2}$ with $m$ vertices such that there is no
tetrahedral complex $K$ satisfying both $\boundary{K} = L$
and the condition that every tetrahedron of $K$ have at
least two facets in $L$.  In particular, the triangulations
of $S^{2}$ with $7$ vertices and degree lists
$(5,5,5,4,4,4,3)$ and $(6,5,5,5,3,3,3)$, i.e.,
polyhedra $7$--$1$ and $7$--$4$ in the catalog of Britton
and Dunitz, both can appear as the link of a vertex
in a $3$-well-centered mesh.
Figure~\ref{fig:onefreev_10/bestrndsym} shows an
example of a $3$-well-centered mesh in $\RR^{3}$
consisting of a single vertex $u$
and its neighborhood $\closure{\star~u}$ such that
$\link~u$ is a triangulation with degree list
$(5, 5, 5, 4, 4, 4, 3)$.
Figure~\ref{fig:onefreevver2_10/goodrndsym} shows
a similar example for the degree list
$(6, 5, 5, 5, 3, 3, 3)$.
There are three other triangulations of $S^{2}$ with $7$
vertices.  Each has a corresponding tetrahedral complex
$K$ satisfying the requirements of
the One-Ring Necessary Condition
(Theorem~\ref{thm:oneringcond}), so
none of these triangulations can appear as the link
of a vertex in a $3$-well-centered mesh.

There are $14$ nonisomorphic triangulations of $S^{2}$
with $8$ vertices.  Of these, $5$ have tetrahedral
complexes $K$ that certify they cannot be the link of
a vertex in a $3$-well-centered tetrahedral mesh in
$\RR^{3}$.  Each of the other $9$ triangulations can
appear as the link of a vertex in a $3$-well-centered
tetrahedral mesh in $\RR^{3}$. (We mention these
results without proof here.)  For $m \le 8$ vertices,
then, the necessary condition of
Theorem~\ref{thm:oneringcond} completely
characterizes which triangulations can and cannot
be made $3$-well-centered.  We leave open the question
of whether the One-Ring Necessary Condition stated in
Theorem~\ref{thm:oneringcond} is a complete
characterization for $m > 8$ vertices in $3$ dimensions
or for $n$-well-centeredness in dimensions $n \ge 4$.

The triangulations on $8$ vertices that cannot be made
$3$-well-centered are polyhedra $8$--$4$, $8$--$5$, $8$--$6$,
$8$--$7$, and $8$--$13$ in the catalog~\cite{BrDu1973} of Britton
and Dunitz.  It is interesting to note that the
degree list of $8$--$7$, which cannot be made
$3$-well-centered, is the same as
the degree list of $8$--$8$, which can
be made $3$-well-centered~\cite{VanderZee2010}.
{\em Thus the degree list of a triangulation
does not provide enough information to determine whether
the triangulation can be the link of a vertex in
a $3$-well-centered tetrahedral mesh in $\RR^{3}$.}
There are $50$ nonisomorphic triangulations with $9$
vertices and an exponentially growing number of
triangulations with more vertices \cite{BrMc2007},
so although making a catalog for $9$ or $10$ vertices
might be somewhat interesting, something more abstract will
be necessary to definitively characterize
which triangulations can be made $3$-well-centered.

In the rest of this section we discuss some
more general results in the direction of characterizing
which triangulations of $S^{2}$ can appear as
the link of a vertex in a $3$-well-centered mesh
in $\RR^{3}$.  The results fall short of a complete
characterization, but do show that the set of
triangulations of $S^{2}$ that cannot appear as
the link of a vertex in a $3$-well-centered mesh
and the set of triangulations that can appear
as the link of a vertex are both infinite.

\bigskip

\begin{corollary}
\label{cor:infno3wc}
For any integer $m \ge 4$ there is a triangulation of $S^{2}$
with $m$ vertices that cannot appear as the link
of a vertex in a $3$-well-centered mesh.
\end{corollary}
\begin{proof}
We have already proved that this holds for $4 \le m \le 6$.

For $m \ge 7$ we note that the tetrahedral
complexes shown on the right hand sides of
Figs.~\ref{fig:vlnk44433} and~\ref{fig:vlnk444444}
can be generalized.  Consider a tetrahedral complex $K$
consisting of a set of $m - 2$ tetrahedra that close
around a common edge.  The complex $K$ satisfies the
conditions of Theorem~\ref{thm:oneringcond}, so
$\boundary{K}$ cannot appear as the link of a vertex in a
$3$-well-centered mesh.  $\boundary{K}$ is a
triangulation of $S^{2}$ on $m$ vertices with
degree list $(m - 2, m - 2, 4, \ldots, 4)$.
\end{proof}

\bigskip

We note that by removing a single
tetrahedron from the example complex $K$ of
Corollary~\ref{cor:infno3wc}, we obtain another
infinite family of triangulations of $S^{2}$ that
cannot appear as the link of a vertex in a
$3$-well-centered mesh.  Each member of this family
is a triangulation on $m$ vertices with
degree list $(m - 1, m - 1, 4, \ldots, 4, 3, 3)$.
This family generalizes the tetrahedral complexes
shown on the left hand side of Fig.~\ref{fig:vlnk44433}
and the right hand side of Fig.~\ref{fig:vlnk554433}.

Corollary~\ref{cor:infno3wc} is one instance that shows
how substantial the difference is between tetrahedral
and triangle meshes.
In the case of triangle meshes in $\RR^{2}$, where we
consider triangulations of $S^{1}$ as the link of a
vertex, the only two triangulations that cannot appear
as the link of a vertex are the $3$-cycle
and the $4$-cycle.  In contrast, there are infinitely
many triangulations of $S^{2}$ that cannot be the link
of a vertex in a $3$-well-centered mesh in $\RR^{3}$.
One may wonder whether there are still infinitely many triangulations
of $S^{2}$ that \emph{can} appear as the link of a
vertex in a $3$-well-centered mesh in $\RR^{3}$.  The answer is
yes.  One way to prove this is to explicitly construct an infinite
family of $3$-well-centered meshes with different vertex
links.  We will do exactly that in a moment, with the help
of the following lemma, which we prove using the
Prism Condition (Proposition~\ref{prop:char_suf}).

\bigskip

\begin{lemma}
\label{lemma:smplxonsphere}
Let $S_{0}^{n-1}$ be a unit $(n-1)$-sphere centered at a
point $u$.  If $\tau^{n-1}$ is an $(n-1)$-well-centered
$(n-1)$-simplex whose vertices lie on $S^{n-1}$, and
the distance from $u$ to $\affine(\tau^{n-1})$ is
greater than $1/\sqrt{2}$, then
$\sigma^{n} := \cone{u}{\tau^{n-1}}$ is
an $n$-well-centered $n$-simplex.
\end{lemma}
\begin{proof}
Suppose that $\tau^{n-1}$ is an $(n -1)$-well-centered
simplex meeting the conditions specified in the hypothesis.
Let $S_{0}^{n-2}$ be the circumsphere of $\tau^{n-1}$.
$S_{0}^{n-2}$ is the intersection of $\affine(\tau^{n-1})$
with $S_{0}^{n-1}$, i.e., an $(n - 2)$-sphere lying in
$S_{0}^{n-1}$.  The orthogonal
projection of $u$ into $\affine(\tau^{n-1})$, which
we denote by $P(u)$, is the
center of $S_{0}^{n-2}$, i.e., the circumcenter
$c(\tau^{n-1})$ of $\tau^{n-1}$.

Since $\tau^{n-1}$ is $(n - 1)$-well-centered, it contains
the point $c(\tau^{n-1})$.  Thus $\tau^{n-1}$ contains
the reflection of $P(u)$ through $c(\tau^{n-1})$.  The
circumradius of $\tau^{n-1}$ satisfies
$R(\tau^{n-1})^2 + z^{2} = 1$, where
$z$ is the distance from $u$ to
$\affine(\tau^{n-1})$, so because $z^{2} > 1/2$,
we have $R(\tau^{n-1}) < 1/\sqrt{2}$, and $u$
lies outside the equatorial ball of
$\tau^{n-1}$.  By the Prism Condition, $\sigma^{n}$
is $n$-well-centered.
\end{proof}

\bigskip

It is relatively straightforward to prove the converse
as well.  For $\sigma^{n} = \cone{u}{\tau^{n-1}}$ with
the vertices of $\tau^{n-1}$ lying on a sphere $S_{0}^{n-1}$
centered at $u$, if $\tau^{n-1}$ is not $(n - 1)$-well-centered
or the distance $z$ from $u$ to $\affine(\tau^{n-1})$
satisfies $z \le 1/\sqrt{2}$, then $\sigma^{n}$ is not
$n$-well-centered.  This proof is left to the reader; the
result is not needed in this paper.

The simplex $\sigma^{n} = \cone{u}{\tau^{n-1}}$ in
Lemma~\ref{lemma:smplxonsphere} is an {\emph{isosceles simplex}}
with all vertices of $\tau^{n-1}$ equidistant from
the apex vertex $u$.  When $n = 2$,
Lemma~\ref{lemma:smplxonsphere} reduces to the
statement that an isosceles triangle is acute if
the apex angle is acute.  In higher dimensions
Lemma~\ref{lemma:smplxonsphere} tells us when an
isosceles simplex is $n$-well-centered. Note that
in an isosceles simplex all of the faces incident
to the apex vertex~$u$ are isosceles; the plane
of each such face intersects the sphere $S^{n-1}$
in some lower-dimensional sphere
centered at $u$, and Lemma~\ref{lemma:smplxonsphere}
can be applied to these isosceles faces.
It follows that $\sigma^{n}$ will be completely well-centered
if $\tau^{n-1}$ is completely well-centered and
$z > 1/\sqrt{2}$.  In particular, for the case
$n = 3$, an isosceles tetrahedron with an acute
triangle facet opposite the apex vertex is a completely
well-centered tetrahedron.

Thus from any triangulation of a unit sphere $S^{2}$ with
sufficiently small
acute triangles we can create a completely well-centered
tetrahedral mesh in $\RR^{3}$ by taking the cone
$\cone{u}{\tau^{2}}$ of each acute triangle $\tau^{2}$
with the center of the sphere $u$.  Figure~\ref{fig:poly7}
shows a completely well-centered tetrahedral mesh
constructed in this fashion.  The boundary of
the mesh in Fig.~\ref{fig:poly7} is an acute
triangulation of $S^{2}$ selected from
an infinite family of acute triangulations
of $S^{2}$.  The next two paragraphs describe
this family.

Consider the set of vertices consisting of the north pole
$(0, 0, 1)$, the south pole $(0, 0, -1)$, and the vertices
of two regular $k$-gons, one in the plane $z = 0.352$ and the
other in the plane $z = -0.352$.  We set the polygons
exactly off phase from each other.  For instance, let the
coordinates of the polygon vertices be
\[
  (0.936\cos\left(\frac{2i\pi}{k}\right),
   0.936\sin\left(\frac{2i\pi}{k}\right), 0.352)\text{,}
  \quad i = 0,1,\ldots,k-1\,\text{,}
\]
\[
  (0.936\cos\left(\frac{(2i + 1)\pi}{k}\right),
   0.936\sin\left(\frac{(2i + 1)\pi}{k}\right), -0.352)\text{,}
  \quad i = 0,1,\ldots,k-1\,\text{.}
\]
Let each pole vertex be adjacent to all of the vertices of
the closer regular polygon.  This constructs $k$ isosceles
triangles incident to each pole.  We take each vertex of a
regular polygon to be adjacent to the closer pole, the two
neighbors on its own regular polygon, and two vertices from
the other regular polygon.  Triangles formed entirely from
vertices of the two regular polygons are also isosceles.
The example in Fig.~\ref{fig:poly7} uses the result of this
construction for the case $k = 7$.

\begin{figure}
\centering
\includegraphics[width=150pt, trim=343pt 337pt 548pt 275pt, clip]
  {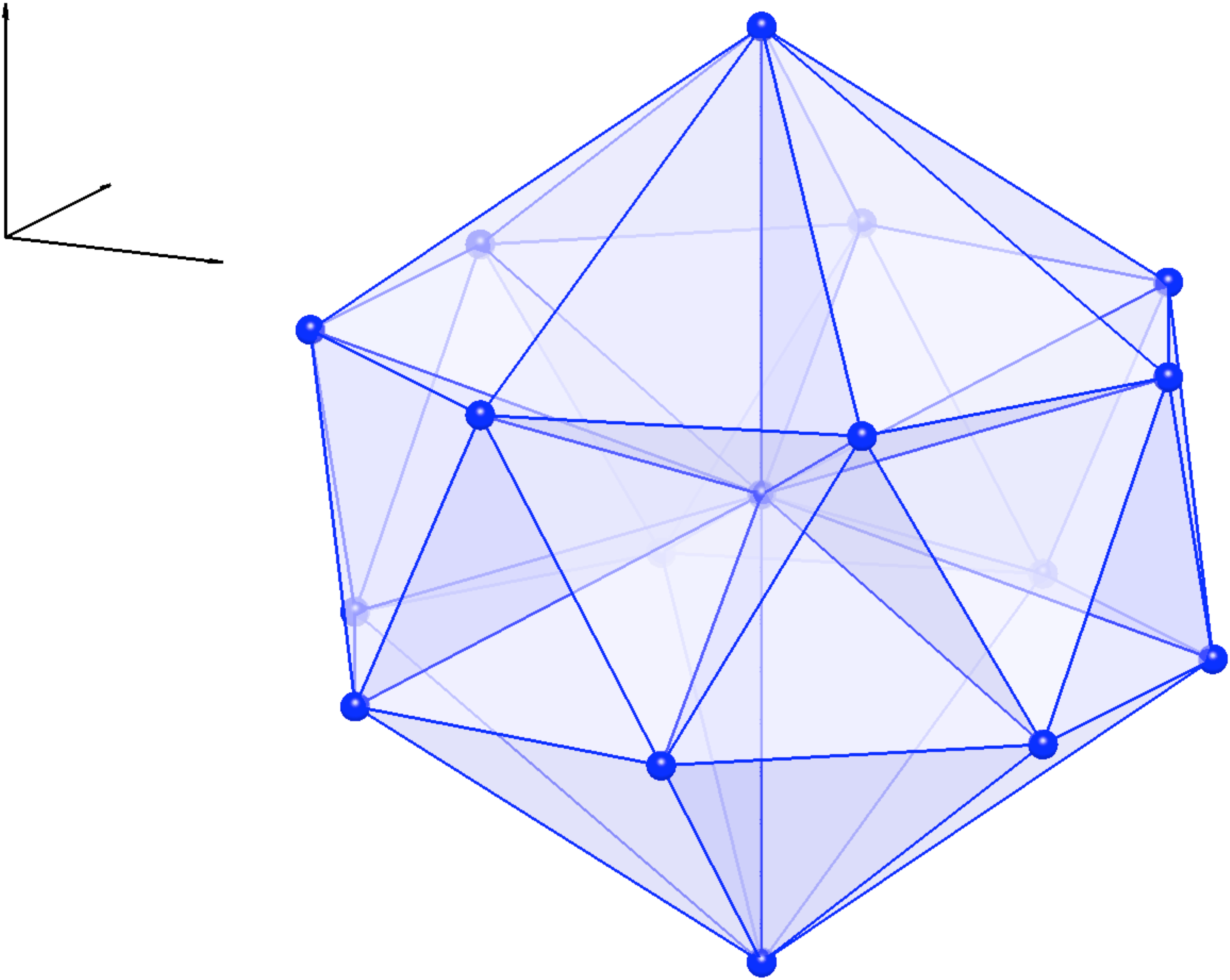}%
\begin{picture}(0, 0)
\put(-138, 83){\scriptsize{$x$}}
\put(-144, 98){\scriptsize{$y$}}
\put(-147, 108){\scriptsize{$z$}}
\end{picture}
\caption{For $k \ge 4$ we can create an acute triangulation
  of the unit sphere from a set of vertices consisting of
  the north and south poles and two out-of-phase regular
  $k$-gons.  Coning such a triangulation to the origin produces a
  completely well-centered tetrahedral mesh.  The figure
  shows the tetrahedral mesh obtained for $k = 7$.}
\label{fig:poly7}
\end{figure}
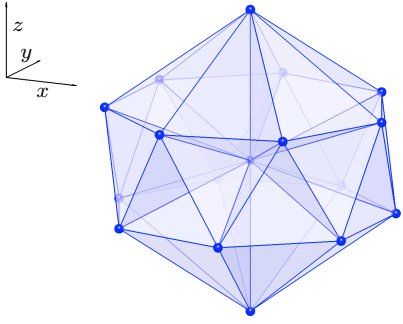

We claim that if $k \ge 4$, then each triangle $\tau^{2}$
of this construction is acute and
satisfies the condition that the distance from the
origin to $\tau^{2}$ is greater than $1/\sqrt{2}$.
Since $k \ge 4$ it is clear that the apex angles of
the isosceles triangles incident to the poles are
acute angles.  Verifying that the other triangles
are acute and that the triangles are far enough
from the origin is straightforward and we omit the details.
Lemma~\ref{lemma:smplxonsphere} applies, and as an immediate
consequence we have the following.

\bigskip

\begin{proposition}
\label{prop:inftycwc}
There are infinitely many triangulations of
$S^{2}$ that can appear as the link of a vertex
in a completely well-centered mesh.
\end{proposition}

\bigskip

For large enough $k$, this construction of completely well-centered
neighborhoods of a vertex using acute triangulations
of a unit sphere $S^{2}$ can be generalized.  One
can use more than two regular $k$-gons, alternating the phase
between each successive $k$-gon.

We have seen that there are infinitely many triangulations of
$S^{2}$ that cannot appear and infinitely many that
can appear as the link of a vertex in
a $3$-well-centered mesh.  The authors suspect
that for $m \ge 8$ vertices the majority of triangulations
of $S^{2}$ on $m$ vertices are triangulations
that can appear as a link of a vertex
in a $3$-well-centered mesh.  We do not formally prove that
conjecture in this paper, but in light of the the next
proposition, it is highly likely;  Proposition~\ref{prop:deg3vtx3wc}
provides a method for constructing new triangulations that can
appear as the link of a vertex in a $3$-well-centered tetrahedral
mesh in $\RR^{3}$.

In Proposition~\ref{prop:deg3vtx3wc} we
consider a triangulation $G$ of $S^{2}$ with a vertex
of degree $3$.  In this context, the notation $G - v_{1}$
refers to the triangulation of $S^{2}$ obtained by
deleting vertex $v_{1}$ and all faces incident to $v_{1}$,
replacing them with the face $[v_{2}v_{3}v_{4}]$, where
$v_{2}, v_{3}, v_{4}$ are the neighbors of $v_{1}$
in $G$, ordered to keep the orientation consistent.

\bigskip

\begin{proposition}
\label{prop:deg3vtx3wc}
Let $G$ be a triangulation of $S^{2}$ with
a vertex $v_{1}$ of degree $3$, and let $v_{2}$,
$v_{3}$, $v_{4}$ be the neighbors of $v_{1}$ in
$G$.  Let $M$ be a tetrahedral mesh in $\RR^{3}$ consisting
of a vertex $u$ and its closed neighborhood $\closure{\star~u}$,
with $\link~{u}$ isomorphic to $G - v_{1}$.  If
\begin{enumerate}[(i)]
\item $M$ is $3$-well-centered
\item face angle $\measuredangle uv_{i}v_{j}$ is acute for
  each $i,j \in \{2,3,4\}, i \ne j$\,\text{,}
\end{enumerate}
then there exists a tetrahedral mesh $\widetilde{M}$
in $\RR^{3}$ and a vertex $u$ of $\widetilde{M}$
such that
\begin{enumerate}[(i)]
\item $\link~{u}$ is isomorphic to $G$
\item $\widetilde{M}$ is $3$-well-centered
\item face angle $\measuredangle uv_{i}v_{j}$ is acute for
  each $i,j \in \{1,2,3,4\}, i \ne j$\,\text{.}
\end{enumerate}
\end{proposition}
\begin{proof}
Figure~\ref{fig:deg3vtx3wc} accompanies this proof and may
help the reader understand the geometric constructions
discussed in the proof.
Consider a particular tetrahedral mesh that satisfies
the conditions of the hypothesis.  In this mesh the tetrahedron
$\sigma = \sigma^{3} = [uv_{2}v_{3}v_{4}]$
is $3$-well-centered, so $c(\sigma)$ is
interior to $\sigma$.

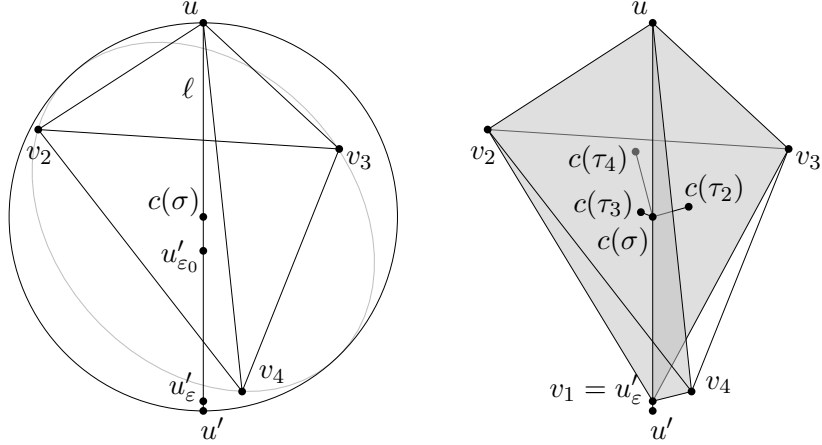
\begin{figure}
\centering
\begin{minipage}[c]{150pt}
\centering
\includegraphics[width=148pt, trim=0pt -27pt 0pt -27pt, clip]%
  {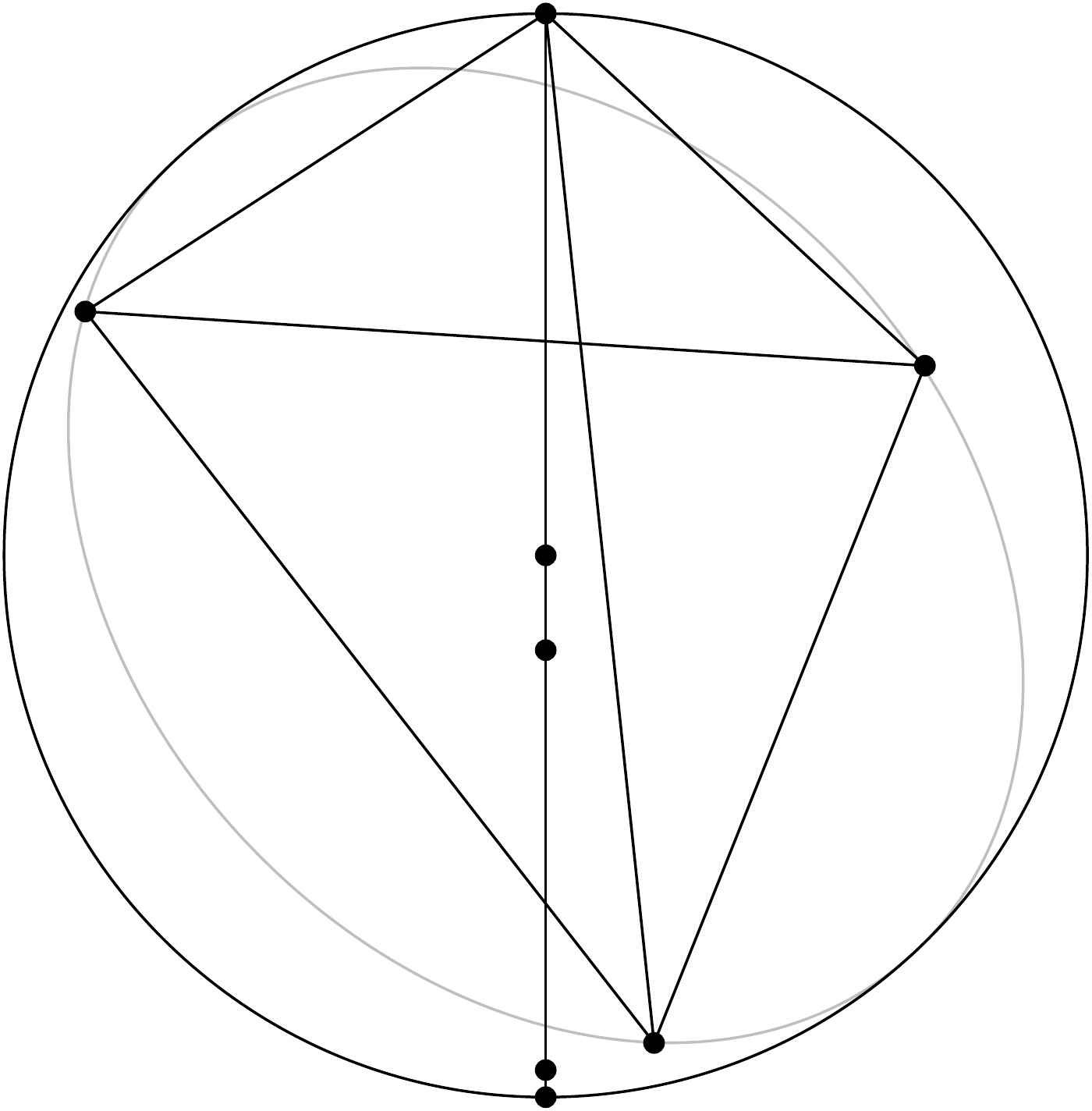}%
\begin{picture}(0,0)
\put(-82,162){$u$}
\put(-82,130){$\ell$}
\put(-141,107){$v_{2}$}
\put(-20,105){$v_{3}$}
\put(-53,24){$v_{4}$}
\put(-95,88){$c(\sigma)$}
\put(-75,1){$u^{\prime}$}
\put(-87,18){$u^{\prime}_{\varepsilon}$}
\put(-90,68){$u^{\prime}_{\varepsilon_{0}}$}
\end{picture}
\end{minipage}%
\hspace{20pt}%
\begin{minipage}[c]{150pt}
\centering
\includegraphics[width=148pt, trim=0pt -27pt 0pt -27pt, clip]%
  {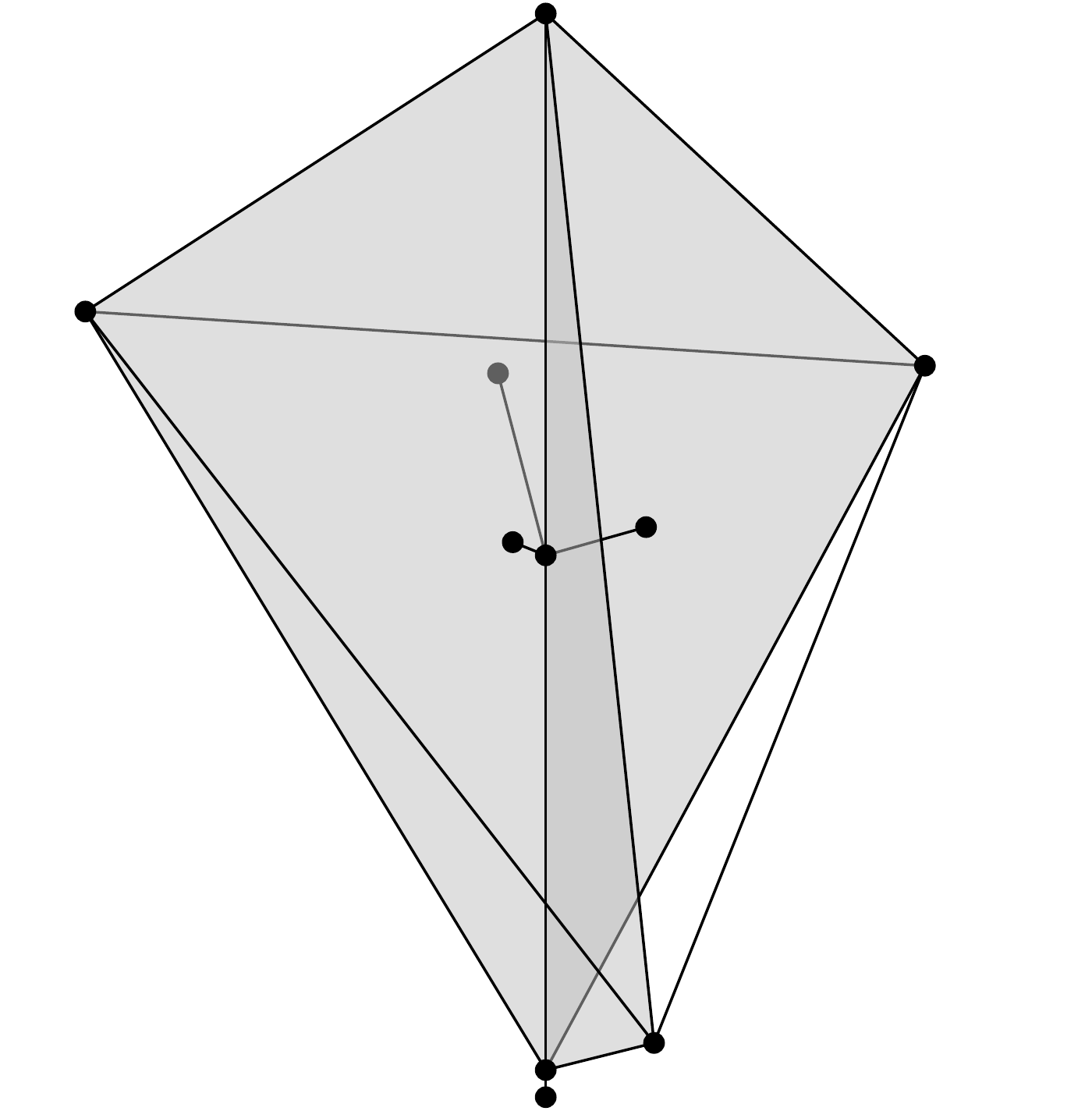}%
\begin{picture}(0,0)
\put(-82,162){$u$}
\put(-143,107){$v_{2}$}
\put(-20,105){$v_{3}$}
\put(-54,20){$v_{4}$}
\put(-95,73){$c(\sigma)$}
\put(-75,1){$u^{\prime}$}
\put(-113,17){$v_{1} = u^{\prime}_{\varepsilon}$}
\put(-62,93){$c(\tau_{2})$}
\put(-104,87){$c(\tau_{3})$}
\put(-106,104){$c(\tau_{4})$}
\end{picture}
\end{minipage}
\caption{Given a $3$-well-centered tetrahedron
  $\sigma = [uv_{2}v_{3}v_{4}]$ with acute angles
  $\measuredangle uv_{i}v_{j}$, one can
  construct three tetrahedra $[uv_{2}v_{3}v_{1}]$,
  $[uv_{3}v_{4}v_{1}]$, and $[uv_{4}v_{2}v_{1}]$ by
  adding a new vertex $v_{1} = u^{\prime}_{\varepsilon}$
  along the line $\ell$ through $u$ and $c(\sigma)$.
  The circumcenters of the constructed tetrahedra lie
  along lines connecting $c(\sigma)$ to the circumcenters
  $c(\tau_{i})$ of the $[uv_{i}v_{j}]$ facets of $\sigma$.
  As discussed in Proposition~\ref{prop:deg3vtx3wc},
  when $v_{1}$ is close enough to
  $u^{\prime}$---the reflection of $u$
  through $c(\sigma)$---the constructed tetrahedra
  will be $3$-well-centered and the angles
  $\measuredangle uv_{1}v_{i}$, $\measuredangle uv_{i}v_{1}$
  will be acute.  The angles $\measuredangle v_{i}uv_{j}$
  do not need to be acute for this construction.
  For example, $\measuredangle v_{2}uv_{3}$ is not an
  acute angle in this figure.}
\label{fig:deg3vtx3wc}
\end{figure}

Let $\ell$ be the line through $u$ and $c(\sigma)$.
Line $\ell$ intersects the circumsphere of $\sigma$ at two
points.  One of these is $u$, and the other we name $u^{\prime}$.
We define
\[
u^{\prime}_{\varepsilon}
    = (1 - \varepsilon)u^{\prime} + \varepsilon u\,\text{,}
\]
a point lying on $\ell$.  Because $\sigma$ is $3$-well-centered,
we know that segment $u u^{\prime}$ intersects triangle
$[v_{2}v_{3}v_{4}]$ at some point
$u^{\prime}_{\varepsilon_{0}}$, with $1/2 > \varepsilon_{0} > 0$.  We
can cut $\sigma$ into the three tetrahedra
$[uv_{2}v_{3}u^{\prime}_{\varepsilon_{0}}]$,
$[uv_{3}v_{4}u^{\prime}_{\varepsilon_{0}}]$, and
$[uv_{4}v_{2}u^{\prime}_{\varepsilon_{0}}]$.

For $\varepsilon_{0} > \varepsilon > 0$ we consider the three tetrahedra
$[uv_{2}v_{3}u^{\prime}_{\varepsilon}]$,
$[uv_{3}v_{4}u^{\prime}_{\varepsilon}]$, and
$[uv_{4}v_{2}u^{\prime}_{\varepsilon}]$.
We claim that for sufficiently small $\varepsilon > 0$ these three
tetrahedra are $3$-well-centered and the face angles
$\measuredangle uu^{\prime}_{\varepsilon}v_{i}$,
$\measuredangle uv_{i}u^{\prime}_{\varepsilon}$
are acute for $i = 2,3,4$.

Examining the face angles first, we note that at $\varepsilon = 0$ the
circumcenters of the facets $[uv_{i}u^{\prime}_{\varepsilon}]$
coincide with $c(\sigma)$ and with each other.  Indeed, each of these
facets is a right triangle with its circumcenter lying on the
hypotenuse $uu^{\prime}_{\varepsilon}$.  As $\varepsilon$ increases,
$\measuredangle v_{i}uu^{\prime}_{\varepsilon}$ does not change,
$\measuredangle uv_{i}u^{\prime}_{\varepsilon}$ decreases, becoming
smaller than $\pi/2$, and $\measuredangle
uu^{\prime}_{\varepsilon}v_{i}$ increases but remains less than
$\pi/2$ for sufficiently small $\varepsilon$.

Turning to the tetrahedra, then, we will argue that the specific
tetrahedron $[uv_{2}v_{3}u^{\prime}_{\varepsilon}]$ is
$3$-well-centered for sufficiently small $\varepsilon$.  An argument
identical except for changed labels applies to the other two
tetrahedra, so this will complete the proof.  We know that, regardless
of the value of $\varepsilon$, the circumcenter of
$[uv_{2}v_{3}u^{\prime}_{\varepsilon}]$ lies on the line orthogonal to
$\affine([uv_{2}v_{3}]) = \affine(\tau_{4})$ passing through
$c(\tau_{4})$; this line is the locus of points equidistant from $u$,
$v_{2}$, and $v_{3}$.  The location of
$c([uv_{2}v_{3}u^{\prime}_{\varepsilon}])$ varies continuously with
$\varepsilon$.  At $\varepsilon = 0$, the circumcenter of tetrahedron
$[uv_{2}v_{3}u^{\prime}_{\varepsilon}]$ coincides with $c(\sigma)$,
and as $\varepsilon$ increases from $0$ towards $\varepsilon_{0}$,
$c([uv_{2}v_{3}u^{\prime}_{\varepsilon}])$ moves in the direction of
vector $c(\tau_{4}) - c(\sigma)$.  Because $\measuredangle
uv_{2}v_{3}$ and $\measuredangle uv_{3}v_{2}$ are acute, we know that
$c(\tau_{4})$ lies in the sector of $\affine(\tau_{4})$ interior to
angle $\measuredangle v_{2}uv_{3}$.  Thus segment
$c(\sigma)c(\tau_{4}) \cap \sigma$ is contained in
$[uv_{2}v_{3}u^{\prime}_{\varepsilon_{0}}] \subset
[uv_{2}v_{3}u^{\prime}_{\varepsilon}]$, and for sufficiently small
$\varepsilon > 0$, tetrahedron $[uv_{2}v_{3}u^{\prime}_{\varepsilon}]$
is $3$-well-centered.
\end{proof}

\bigskip

Because the face angles $\measuredangle uv_{1}v_{i}$,
$\measuredangle uv_{i}v_{1}$ are acute in the construction
of Proposition~\ref{prop:deg3vtx3wc}, the construction
can be iterated.  If a triangulation $G$ of $S^{2}$
satisfies the conditions of Proposition~\ref{prop:deg3vtx3wc},
then a degree $3$ vertex $v_{1}$ can be inserted into face
$[v_{2}v_{3}v_{4}]$.  In the new triangulation of $S^{2}$,
the three new faces incident to $v_{1}$ satisfy the conditions of
Proposition~\ref{prop:deg3vtx3wc}, so a degree $3$ vertex
can be inserted into any one of those three faces,
and so on.  In particular, starting from
any completely well-centered mesh constructed from an acute
triangulation of a unit sphere $S^{2}$, one can successively
insert vertices of degree $3$ to create an infinite family of
triangulations that can appear as the link of a vertex
in a $3$-well-centered mesh.

It is also worth mentioning that each triangulation of a topological
$S^{2}$ with $8$ vertices $v_{1},\ldots,v_{8}$ that can appear as
$\link~u$ for a vertex $u$ in a $3$-well-centered mesh in $\RR^{3}$
has an embedding into $\RR^{3}$ for which all of the face angles
$\measuredangle uv_{i}v_{j}$ are acute for $i,j\in \{1,\ldots,8\}$, $i
\ne j$, where $v_iv_j$ is an edge in the triangulation.  Recall that
there are $50$ nonisomorphic triangulations of $S^{2}$ with $9$
vertices~\cite{BrMc2007}.  Using Proposition~\ref{prop:deg3vtx3wc} to
add vertices of degree $3$ to the various faces of triangulations of
$S^{2}$ with $8$ vertices, one can show that at least $34$ of these
$50$ triangulations of $S^{2}$ with $9$ vertices can appear as the
link of a vertex in a $3$-well-centered tetrahedral mesh embedded in
$\RR^{3}$.

\section{Local Combinatorial Properties of
  $2$-Well-Centered Tetrahedral Meshes}
\label{sec:combinatorial2wctetcond}

Corollary~\ref{cor:min3wc} shows that in a $3$-well-centered mesh
there are at least $7$ edges incident to each vertex.
In the following discussion we will see that the combinatorial
constraints for a mesh to be $2$-well-centered are
quite different from the constraints for
a mesh to be $3$-well-centered, and in terms of the minimum
number of edges incident to a vertex, are more stringent.
As in Sec.~\ref{sec:combinatorial3wccond}, the discussion
focuses on $\link~u$ where $u$ is a vertex interior
to a tetrahedral mesh in $\RR^{3}$.

\begin{definition} We say that a particular triangulation
$G$ of $S^{2}$ {\emph{permits a $2$-well-centered
neighborhood}} of a vertex~$u$ if there exists a
tetrahedral mesh $M$ in $\RR^{3}$ such that $u$ is an
interior vertex of $M$, $\link~u$ is isomorphic to $G$
(as a simplicial complex), and all facets of $M$ incident
to $u$ are $2$-well-centered.  (A facet means a
$2$-simplex in this context---a face of dimension $n - 1$.)
\end{definition}

It should be noted that this definition does not directly
address the question of whether the tetrahedra incident to $u$ are
$2$-well-centered, since each tetrahedron incident to $u$
has one facet lying on $\link~u$, and that facet
is not incident to $u$.  We shall see, however, that
for tetrahedral meshes in $\RR^{3}$, the smallest
triangulation that permits a $2$-well-centered
neighborhood in the sense of this definition can, in
fact, appear as the link of a vertex in a completely
well-centered mesh.  Finally, note that phrasing the
problem in terms of the facets of $M$ incident
to $u$ actually reduces the problem to determining
whether the face angles at $u$ are acute, because if
there is an arrangement of rays at $u$ such that all
of the face angles formed at $u$ by these rays are acute,
then we can place the neighbors of $u$ at the points where
these rays intersect a unit sphere centered at $u$.
This will create a neighborhood of $u$ in which every
$2$-dimensional face incident to $u$ is an isosceles
triangle with an acute apex angle at $u$.

The first result of this section is a simple observation
that forms the foundation for the theory developed in the
rest of the section.

\bigskip

\begin{lemma}
\label{lemma:orthogonalplane}
Let $u$ and $v_{1}$ be adjacent vertices in a
tetrahedral mesh $M$ embedded in $\RR^{3}$ and let
$v_{i}$ be a vertex of $\link~u$ that is adjacent
to $v_{1}$.  The angle $\measuredangle v_{1}uv_{i}$
is acute if and only if
$v_{i} \in H_{1}$, where
$H_{1}$ is the open halfspace that contains
$v_{1}$ and is bounded by the plane through
$u$ orthogonal to the vector $v_{1} - u$.
\end{lemma}
\begin{proof}
The angle $\measuredangle v_{1}uv_{i}$ is acute if and only
if $\langle v_{1} - u, v_{i} - u\rangle > 0$,
where $\langle \cdot, \cdot\rangle$ is the standard
inner product on $\RR^{3}$, and this holds if
and only if $v_{i}$ lies in $H_{1}$.
\end{proof}

\bigskip

The next two technical lemmas are
based on Lemma~\ref{lemma:orthogonalplane}.
They lead to the proof of the main result
of this section.  In both lemmas
and in the subsequent theorem we use the following notation.
We denote by $u$ a vertex in a tetrahedral
mesh in $\RR^{3}$, and the $m$ vertices of
$\link~u$ are labeled $v_{1}, \ldots, v_{m}$.
For each vertex $v_{i}$, the plane through
$u$ orthogonal to $v_{i} - u$ is denoted $P_{i}$,
and the open halfspace bounded by $P_{i}$ that
contains $v_{i}$ is denoted $H_{i}$.  The
other halfspace bounded by $P_{i}$ will be
called $H^{\prime}_{i}$, and we take this to
be a closed halfspace, which contains its
boundary $P_{i}$.  The orthogonal
projection of a vertex $v_{j}$ into $P_{i}$
will be denoted $Pr_{i}(v_{j})$.

\bigskip

\begin{lemma}
\label{lemma:2wcvtxproj}
Let $v_{1}$ and $v_{2}$ be nonadjacent vertices
of $\link~u$, with $v_{2} \in H^{\prime}_{1}$.
If $v_{i}$ is a vertex adjacent to
both $v_{1}$ and $v_{2}$ such that
$\measuredangle v_{1}uv_{i}$ and
$\measuredangle v_{2}uv_{i}$ are both
acute angles, then the orthogonal
projection of $v_{i}$ into $P_{1}$ lies
in $P_{1} \cap H_{2}$.
\end{lemma}
\begin{proof}

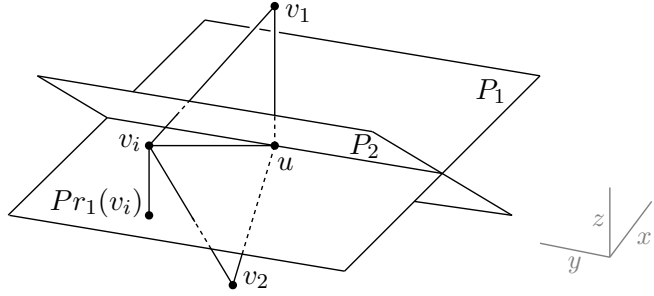
\begin{figure}
\centering
\includegraphics[width=244pt, trim=0pt 0pt 0pt 0pt, clip]
  {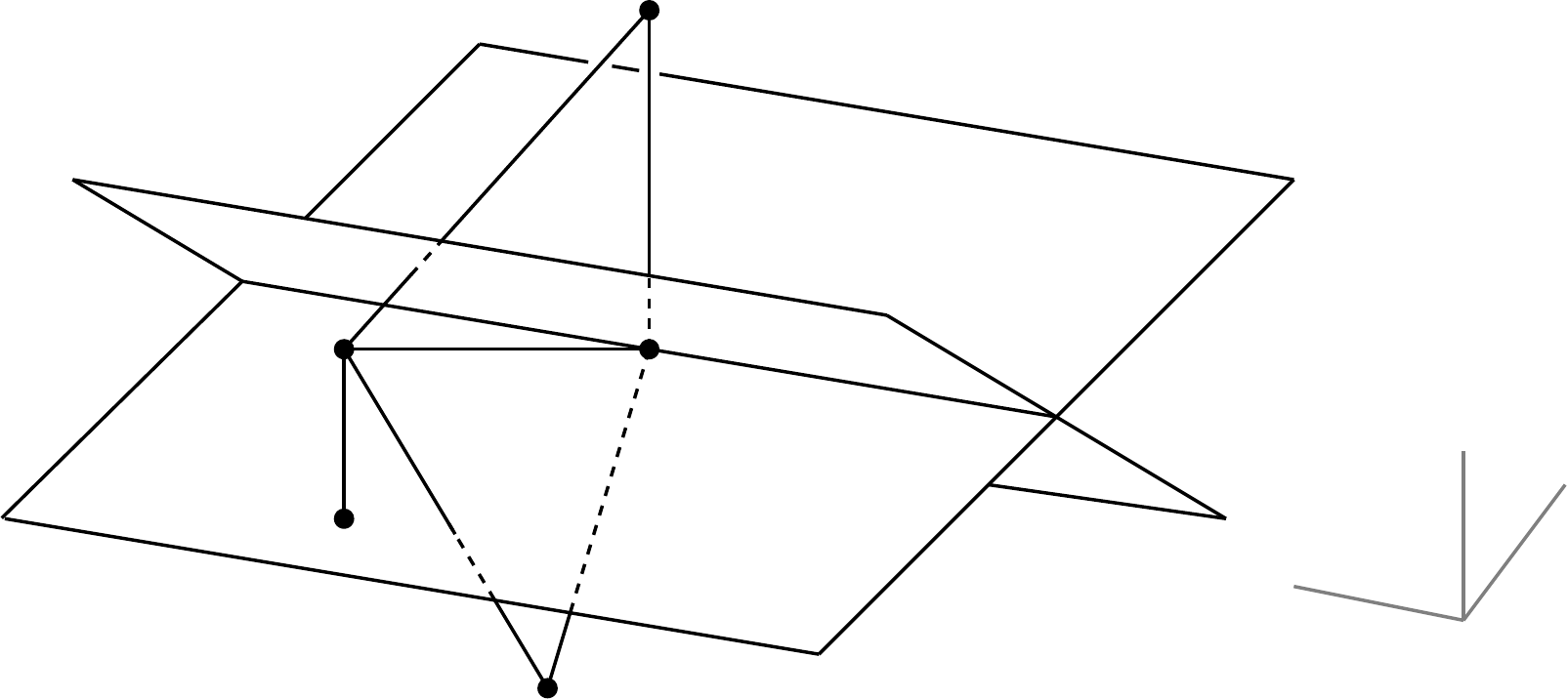}%
\begin{picture}(0,0)
\put(-142,45){$u$}
\put(-139,103){$v_{1}$}
\put(-155,2){$v_{2}$}
\put(-202,54){$v_{i}$}
\put(-67,72){$P_{1}$}
\put(-114.8,51.3){$P_{2}$}
\put(-228,30){$Pr_{1}(v_{i})$}
\put(-6,17){\gray{\small{$x$}}}
\put(-32,8){\gray{\small{$y$}}}
\put(-23,24){\gray{\small{$z$}}}
\end{picture}
\caption{If a $2$-well-centered mesh contains two vertices
  $v_{1}$ and $v_{2}$ that both lie in $\link~u$,
  are not adjacent to each other, and have a common
  neighbor $v_{i}$ and if $v_{2}$ lies in $H^{\prime}_{1}$,
  then the orthogonal projection of $v_{i}$ into $P_{1}$,
  i.e., the point $Pr_{1}(v_{i})$, must lie in $P_{1} \cap H_{2}$.}
\label{fig:2wcvtxproj}
\end{figure}

The sketch in Fig.~\ref{fig:2wcvtxproj} illustrates
this result.  For an algebraic proof we assign a
coordinate system with $u$ as the origin and $v_{1}$
lying on the positive $z$ axis.  Using coordinates
$(x_{i}, y_{i}, z_{i})$ for vertex $v_{i}$, the condition
$v_{2} \in H^{\prime}_{1}$ means that $z_{2} \le 0$.  Since
the angle $\measuredangle v_{1}uv_{i}$ is acute,
Lemma~\ref{lemma:orthogonalplane} implies that $v_{i}$
must lie in $H_{1}$, and since the angle
$\measuredangle v_{2}uv_{i}$ is acute,
Lemma~\ref{lemma:orthogonalplane} implies that $v_{i}$
must lie in $H_{2}$.  Thus $v_{i}$ lies
in $H_{1} \cap H_{2}$.  Since $H_{1} \cap H_{2}$ would
be empty if $v_{2}$ had coordinates $(0, 0, z_{2})$,
we can conclude that $v_{2}$ does not lie on the $z$-axis.
With the remaining freedom in defining a coordinate
system we specify that $v_{2}$ has coordinates
$(x_{2}, 0, z_{2})$ with $x_{2} < 0$.

Now since $v_{i} \in H_{1}$, we know that $z_{i} \ge 0$.
We also know that
$\langle v_{i}, v_{2}\rangle = x_{i}x_{2} + z_{i}z_{2} > 0$,
because $v_{i} \in H_{2}$.
We have established that $z_{i}z_{2} \le 0$ and
that $x_{2} < 0$.  It follows that $x_{i} < 0$.  The
projection $Pr_{1}(v_{i})$ has coordinates
$(x_{i}, y_{i}, 0)$ and is interior to
$P_{1} \cap H_{2} = \{(x,y,0):x < 0\}$.
\end{proof}

\bigskip


\begin{lemma}
\label{lemma:facetintrsct}
Let $v_{1}$ and $v_{2}$ be nonadjacent vertices of $\link~u$,
with $v_{2} \in H^{\prime}_{1}$.
If $[v_{i}v_{2}v_{j}]$ is a $2$-simplex of $\link~u$,
such that $v_{i}, v_{j}$ are both adjacent to $v_{1}$
and the face angles
$\measuredangle v_{1}uv_{i}$, $\measuredangle v_{1}uv_{j}$,
$\measuredangle v_{2}uv_{i}$, $\measuredangle v_{2}uv_{j}$,
are all acute angles, then
$Pr_{1}([v_{i}v_{2}v_{j}]) \subset P_{1} \cap H_{2}$,
i.e., the orthogonal projection of the entire
facet $[v_{i}v_{2}v_{j}]$ into $P_{1}$ lies in $H_{2}$.
\end{lemma}
\begin{proof}
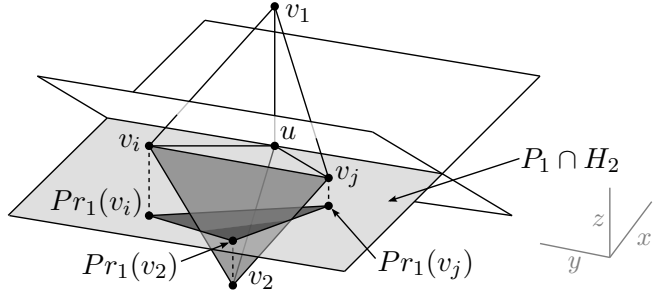
\begin{figure}
\centering
\includegraphics[width=244pt, trim=0pt 0pt 0pt 0pt, clip]
  {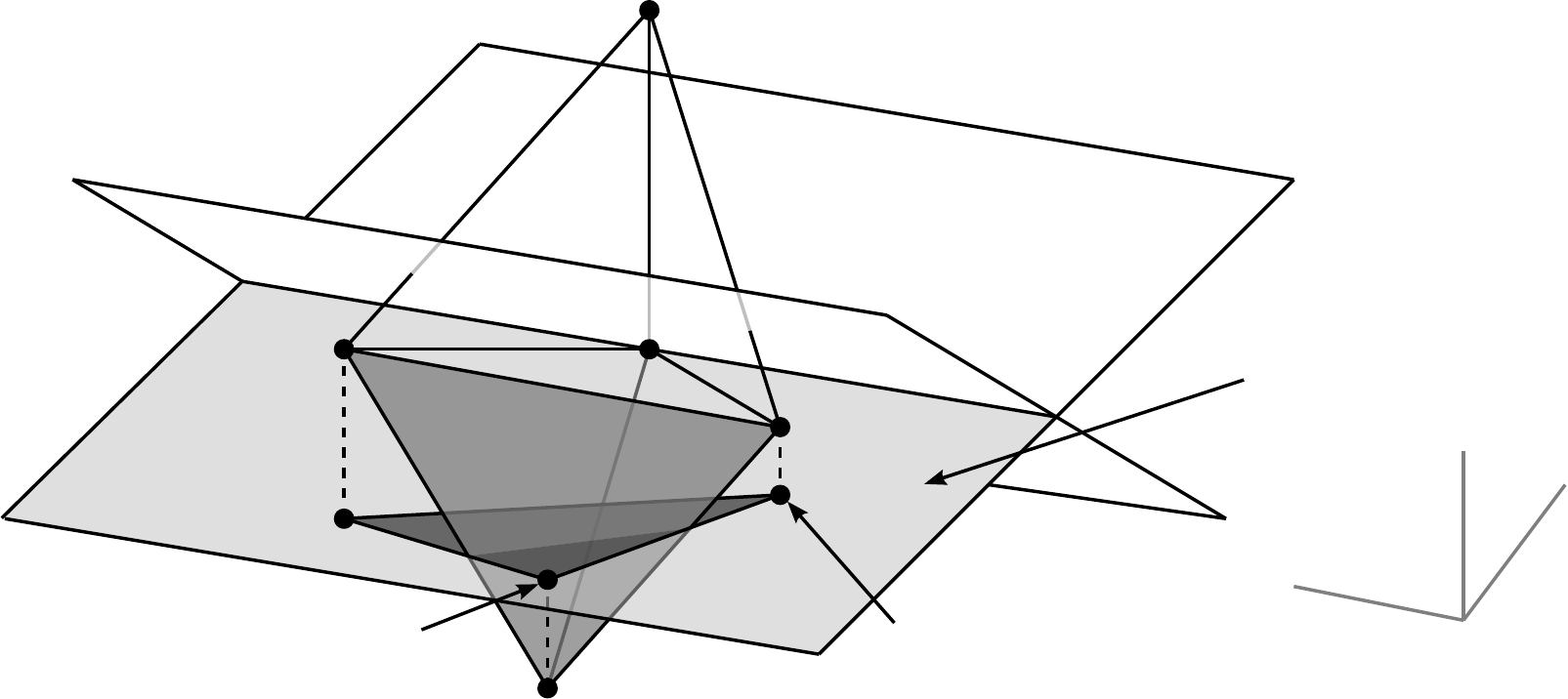}%
\begin{picture}(0,0)
\put(-141,57){$u$}
\put(-139,103){$v_{1}$}
\put(-153,2){$v_{2}$}
\put(-216,6){$Pr_{1}(v_{2})$}
\put(-202,54){$v_{i}$}
\put(-228,30){$Pr_{1}(v_{i})$}
\put(-120,42){$v_{j}$}
\put(-104,7){$Pr_{1}(v_{j})$}
\put(-50,46){$P_{1}\cap H_{2}$}
\put(-6,17){\gray{\small{$x$}}}
\put(-32,8){\gray{\small{$y$}}}
\put(-23,24){\gray{\small{$z$}}}
\end{picture}
\caption{Let $u$ be a vertex of a tetrahedral mesh
embedded in $\RR^{3}$, and let $v_{1}$, $v_{2}$,
$v_{i}$, $v_{j}$ be vertices of $\link~u$ with
adjacencies as shown.  If the face angles at $u$
between adjacent vertices of $\link~u$ are all
acute angles, but $\measuredangle v_{1}uv_{2}$
is nonacute, then the projection of facet
$[v_{i}v_{2}v_{j}]$ into $P_{1}$ lies in
$P_{1} \cap H_{2}$.}
\label{fig:2wcfacetproj}
\end{figure}

See the sketch in Fig.~\ref{fig:2wcfacetproj}.
From the given hypotheses we can conclude by
Lemma~\ref{lemma:2wcvtxproj} that
$Pr_{1}(v_{i})$ and $Pr_{1}(v_{j})$ both lie in
$P_{1} \cap H_{2}$.  Using the same coordinate system
defined in the proof of Lemma~\ref{lemma:2wcvtxproj},
the point $Pr_{1}(v_{2})$ has coordinates $(x_{2}, 0, 0)$
with $x_{2} < 0$, thus it lies in $P_{1} \cap H_{2}$ as
well.  It follows that the orthogonal projection of the
facet $[v_{i}v_{2}v_{j}]$ into $P_{1}$, which is
the convex hull of $Pr_{1}(v_{i})$, $Pr_{1}(v_{2})$,
and $Pr_{1}(v_{j})$, lies entirely in
the convex set $P_{1} \cap H_{2}$.
\end{proof}

\bigskip

Applying the above two lemmas, we obtain a combinatorial
necessary condition on the neighborhood of an interior vertex
in a $2$-well-centered mesh.

\bigskip

\begin{theorem}
\label{thm:nminus3}
Let $G$ be a triangulation of $S^{2}$ with
$m$ vertices.  If $G$ contains a vertex $v_{1}$ of degree
$d(v_{1}) \ge m - 3$, then $G$ does not permit a
$2$-well-centered neighborhood.
\end{theorem}
\begin{proof}
We consider a vertex $u$ such that $\link~u$ is
isomorphic to $G$ where $G$ has a vertex of degree
at least $m - 3$ and consider a geometric realization of
$\closure{\star~u}$ in $\RR^{3}$.  Label the vertices
of $\link~u$ with the labels $v_{1},v_{2},\ldots,v_{m}$
such that $v_{1}$ is a vertex of maximum degree and the
(at most two) vertices not adjacent to $v_{1}$ are
listed immediately after $v_{1}$ (e.g., labeled $v_{2}$,
$v_{3}$ if there are two of them).
We choose a coordinate system on $\RR^{3}$
such that $u$ is at the origin and $v_{1}$ lies on the
positive $z$ axis.

\begin{figure}
\centering
\begin{minipage}[c]{180pt}%
\includegraphics[width=180pt, trim=0pt 0pt 0pt 0pt, clip]
  {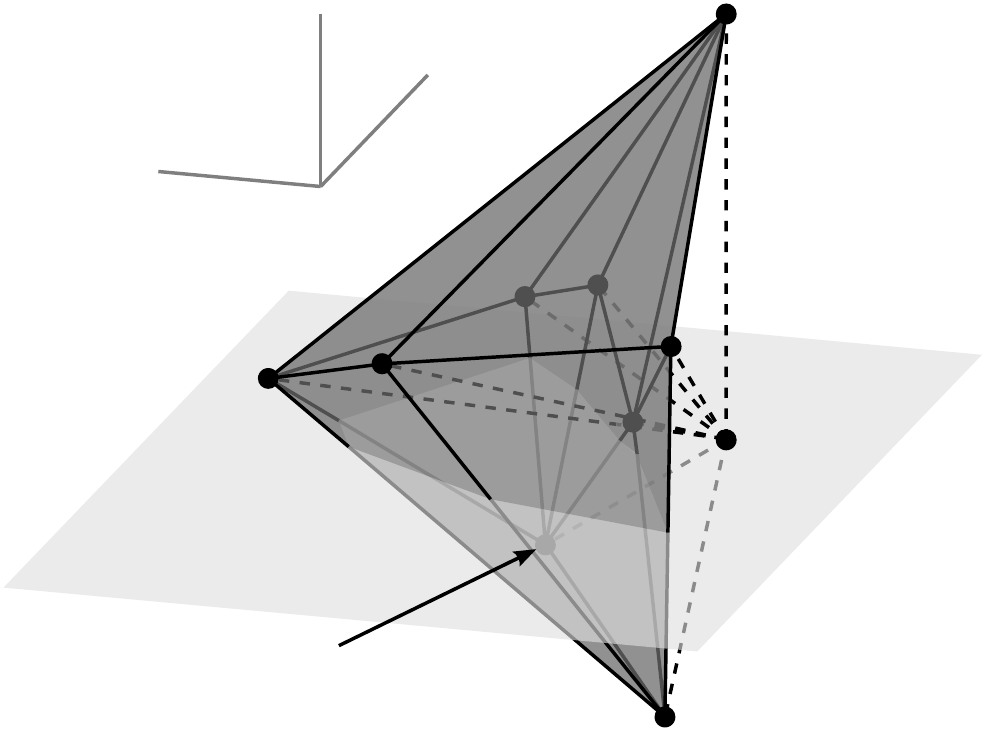}%
\begin{picture}(0,0)
\put(-43,52){$u$}
\put(-43,127){$v_{1}$}
\put(-54,3){$v_{2}$}
\put(-124,9){$v_{3}$}
\put(-160,30){\small{$P_{1}$}}
\put(-109,106){\gray{\small{$x$}}}
\put(-141,94){\gray{\small{$y$}}}
\put(-129,114){\gray{\small{$z$}}}
\end{picture}%
\end{minipage}%
\hspace{30pt}%
\begin{minipage}[c]{120pt}%
\includegraphics[width=120pt, trim=0pt 0pt 0pt 0pt, clip]
  {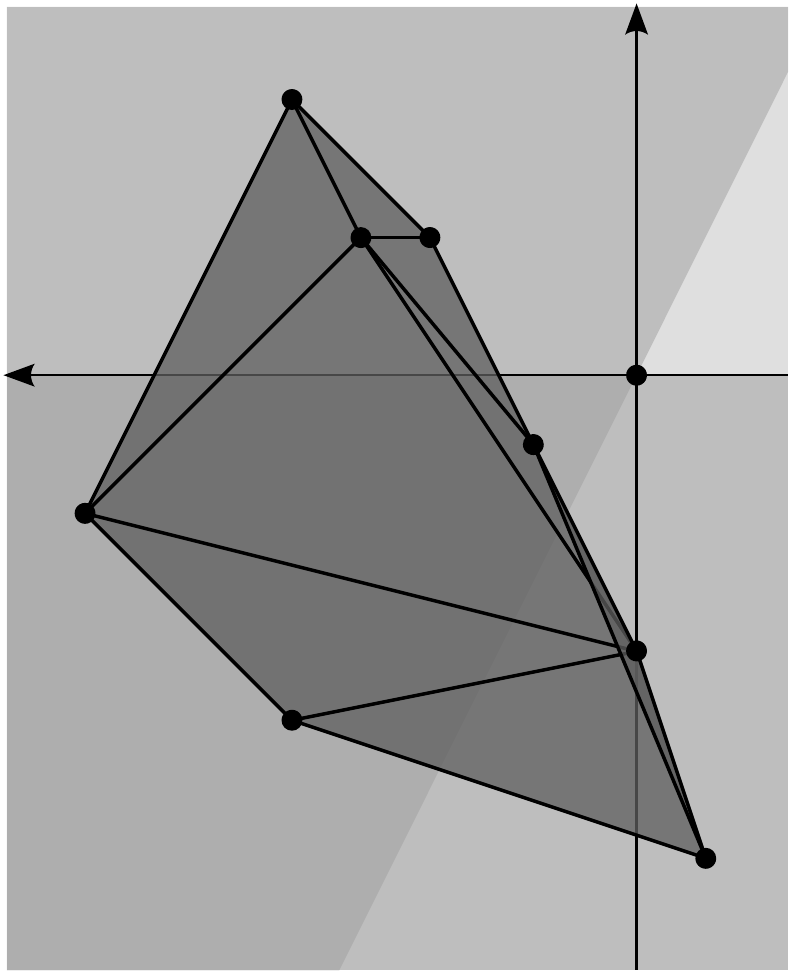}%
\begin{picture}(0,0)
\put(-21,82){$u$}
\put(-21,50){$v_{2}$}
\put(-78,112){$v_{3}$}
\put(-21,140){\scriptsize{$x$}}
\put(-117,83){\scriptsize{$y$}}
\put(-60,3){\small{$P_{1}\cap H_{2}$}}
\put(-117,136){\small{$P_{1}\cap H_{3}$}}
\put(-12,98){\small{$P_{1}$}}
\end{picture}%
\end{minipage}
\caption{When $\link~u$ has $m$ vertices
  and one of the vertices $v_{1}$ has degree $d(v_{1}) \ge m - 3$,
  any geometric realization of $\closure{\star~u}$ in $\RR^{3}$ with
  all face angles at $u$ acute is not an embedding.
  Theorem~\ref{thm:nminus3} shows that if we consider such a
  geometric realization and project every facet that intersects
  $H_{1}^{\prime}$ into $P_{1}$, then the union of the projected
  facets does not contain $u$.  The sketch at left shows an
  example of a geometric realization of a tetrahedral mesh
  $\closure{\star~u}$ in $\RR^{3}$ such that
  every face angle at vertex $u$ is acute.
  In the sketch, $v_{1}$ has degree $d(v_{1}) = 6 = m - 3$
  in~$\link~u$.  The sketch at right shows the result of taking
  the geometric realization on the left and projecting each facet
  that intersects $H_{1}^{\prime}$ into $P_{1}$.}
\label{fig:2wcthmconcept}
\end{figure}
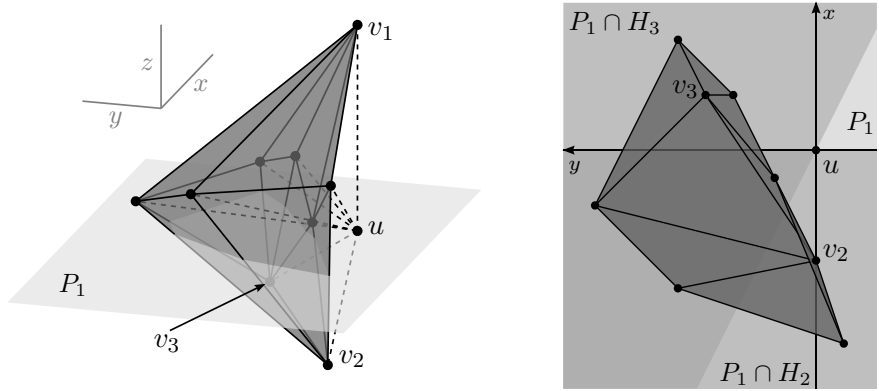

Assume that all of the face angles $\measuredangle v_{i}uv_{j}$
are acute.  We claim this implies that for any facet
$[v_{i}v_{j}v_{k}]$ with at least one vertex
in $H^{\prime}_{1}$, the orthogonal projection of
the facet into $P_{1}$, i.e., $Pr_{1}([v_{i}v_{j}v_{k}])$,
does not contain vertex $u$.
Assuming this claim for the moment, we see that $u$ lies
outside the (solid) polyhedron bounded by $\link~u$.
(See Fig.~\ref{fig:2wcthmconcept}.)  Since $u$
is outside this polyhedron, some
$3$-simplex incident to $u$ must be inverted.
Thus the geometric realization of
$\closure{\star~u}$ is not an embedding, and the
claim completes the proof.

We proceed to prove the claim.
Noting that $v_{1} \in H_{1}$ by
our definition of $H_{1}$, we observe that
for $i \ge 4$, vertex $v_{i}$ must lie
in $H_{1}$ because $v_{i}$ is adjacent to $v_{1}$.
(This follows from Lemma~\ref{lemma:orthogonalplane}.)
Thus there are only two
types of facets that may have nonempty intersection
with $H^{\prime}_{1}$.  The first
type is $[v_{i}v_{2}v_{j}]$ or $[v_{i}v_{3}v_{j}]$
where $v_{i}$ and $v_{j}$ both are adjacent to $v_{1}$,
and the second type is $[v_{2}v_{3}v_{j}]$ for $j \ge 4$.
Consider, then, the first type of facet, taking
the specific notation $[v_{i}v_{2}v_{j}]$.
(The same argument applies to $[v_{i}v_{3}v_{j}]$.)
If $v_{2}$ lies in $H_{1}$, we are done; the facet
does not intersect $H^{\prime}_{1}$.  Otherwise
$v_{2}$ lies in $H^{\prime}_{1}$.  Hence
$\measuredangle v_{1}uv_{2}$ is nonacute, and
$v_{2}$ is not adjacent to $v_{1}$.
Lemma~\ref{lemma:facetintrsct} applies.

The proof for facets of the second type is
more complicated.  If both $v_{2}$ and $v_{3}$
lie in $H_{1}$, we are done.  If one vertex lies in
$H_{1} \cup P_{1}$ and the other lies in $H^{\prime}_{1}$,
we assume without loss of generality that
$z_{2} \le 0$ and $z_{3} \ge 0$.

Then $v_{2}$ is not
adjacent to $v_{1}$.   If $v_{3}$ is adjacent to $v_{1}$,
then Lemma~\ref{lemma:facetintrsct} applies directly 
with $v_{3}$ functioning as $v_{i}$. On the other hand,
even if $v_{3}$ is not adjacent to $v_{1}$, the arguments of
Lemmas~\ref{lemma:2wcvtxproj} and~\ref{lemma:facetintrsct}
can be applied with $v_{3}$ functioning as $v_{i}$.
(In the proofs of Lemmas\ref{lemma:2wcvtxproj}
and~\ref{lemma:facetintrsct} we used
$v_{i}$ adjacent to $v_{1}$ to establish only that
$z_{i} \ge 0$ and that $v_{2}$ does not lie on the
$z$-axis.  The latter holds in this
case because $v_{2}$ and $v_{1}$ have common neighbor
$v_{j} \ne v_{3}$.)

This leaves the case $z_{2} < 0$ and
$z_{3} < 0$.  As noted above, $v_{2}$ does not lie on
the $z$-axis.  We choose the coordinate system
with $v_{2} = (x_{2}, 0, z_{2})$, $x_{2} < 0$.
We also assume without loss of generality
that $y_{3} \ge 0$.  (We can reflect through the plane
$y = 0$ if $y_{3} < 0$.)  See Fig.~\ref{fig:2wcthmprf}
for sketches related to this case.

By applying Lemma~\ref{lemma:orthogonalplane} three
times, we obtain $v_{j} \in H_{1} \cap H_{2} \cap H_{3}$,
and by applying Lemma~\ref{lemma:2wcvtxproj} twice we
obtain $Pr_{1}(v_{j}) \in P_{1} \cap H_{2} \cap H_{3}$.
If $x_{3} < 0$, then the whole segment
$Pr_{1}([v_{2}v_{3}])$ lies in $P_{1} \cap H_{2}$,
and since $Pr_{1}(v_{j}) \in P_{1} \cap H_{2}$, it follows
that $Pr_{1}([v_{2}v_{3}v_{j}]) \subset P_{1} \cap H_{2}$.
Thus $Pr_{1}([v_{2}v_{3}v_{j}])$ does not contain $u$.

\begin{figure}
\centering
\begin{minipage}[c]{150pt}%
\includegraphics[width=150pt, trim=0pt 0pt 0pt 0pt, clip]
  {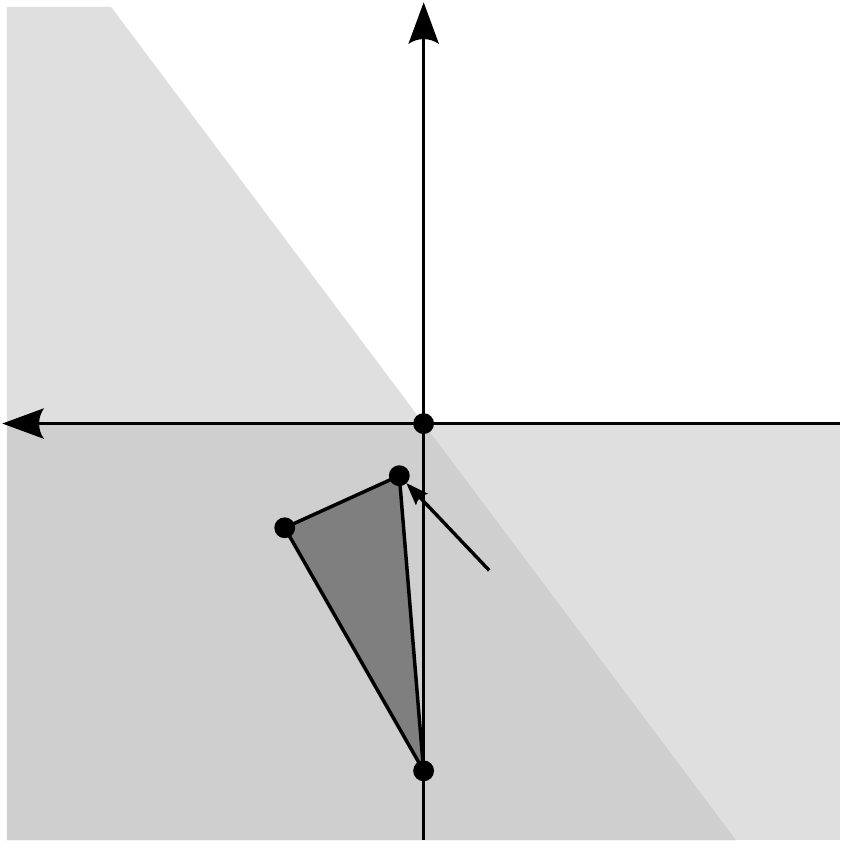}%
\begin{picture}(0,0)
\put(-73,78){$u$}
\put(-73,6){$Pr_{1}(v_{2})$}
\put(-139,52){$Pr_{1}(v_{3})$}
\put(-63,44){$Pr_{1}(v_{j})$}
\put(-70,142){\scriptsize{$x$}}
\put(-146,67){\scriptsize{$y$}}
\put(-38,65){\small{$P_{1}\cap H_{2}$}}
\put(-147,118){\small{$P_{1}\cap H_{3}$}}\end{picture}%
\end{minipage}%
\hspace{40pt}%
\begin{minipage}[c]{150pt}%
\includegraphics[width=150pt, trim=0pt 0pt 0pt 0pt, clip]
  {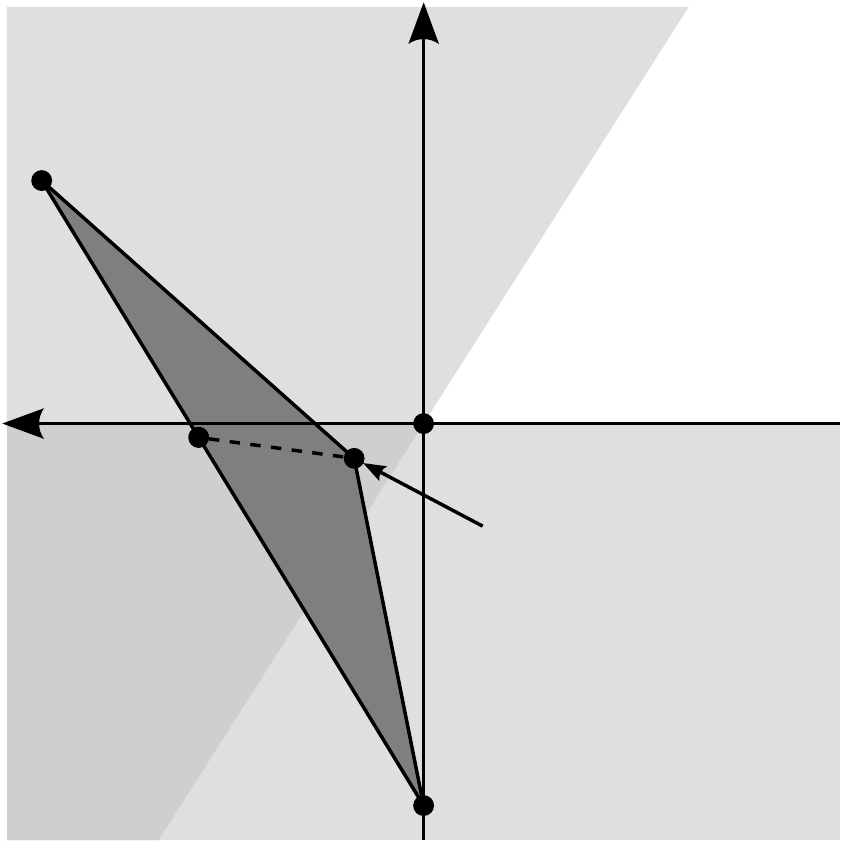}%
\begin{picture}(0,0)
\put(-73,78){$u$}
\put(-73,6){$Pr_{1}(v_{2})$}
\put(-138,118){$Pr_{1}(v_{3})$}
\put(-65,51){$Pr_{1}(v_{j})$}
\put(-124,65){$p_{0}$}
\put(-70,142){\scriptsize{$x$}}
\put(-146,67){\scriptsize{$y$}}
\put(-38,65){\small{$P_{1}\cap H_{2}$}}
\put(-147,140){\small{$P_{1}\cap H_{3}$}}
\end{picture}%
\end{minipage}
\caption{In the proof of Theorem~\ref{thm:nminus3}, the most
  difficult case to analyze is a facet of the second type
  $[v_{2}v_{3}v_{j}]$ with $z_{2} < 0$ and $z_{3} < 0$.  These
  sketches illustrate the projection of the facet $[v_{2}v_{3}v_{j}]$
  into $P_{1}$ for the two subcases $x_{3} < 0$ (left)
  and $x_{3} \ge 0$ (right).  In both subcases, $Pr_{1}(v_{j})$
  lies in $P_{1} \cap H_{2} \cap H_{3}$.  When $x_{3} < 0$,
  $Pr_{1}(v_{3}) \in P_{1} \cap H_{2}$, so the projection of
  facet $[v_{2}v_{3}v_{j}]$ into $P_{1}$ is a subset of
  $P_{1} \cap H_{2}$.  When $x_{3} \ge 0$, the projected
  facet can be decomposed into two pieces meeting along
  $[p_{0}Pr_{1}(v_{j})]$.  One piece lies in $P_{1} \cap H_{2}$,
  and the other piece lies in $P_{1} \cap H_{3}$.  In both
  subcases we see that $u \notin Pr_{1}([v_{2}v_{3}v_{j}])$.}
\label{fig:2wcthmprf}
\end{figure}
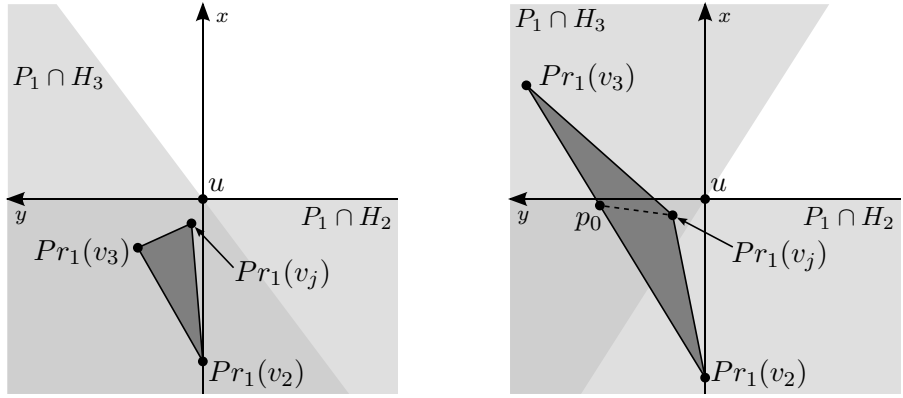

So we assume that $x_{3} \ge 0$.  Now if $x_{3} = 0$
we know that $y_{3} \ne 0$ because $v_{3}$, like~$v_{2}$,
does not lie on the $z$-axis.  Moreover, $x_{3} > 0$ also
implies $y_{3} \ne 0$, since otherwise $P_{1} \cap H_{3}$ would
be $\{(x,y,0): x > 0\}$, yielding
$P_{1} \cap H_{2} \cap H_{3} = \emptyset \owns Pr_{1}(v_{j})$.
A point on $Pr_{1}([v_{2}v_{3}])$ has the form
\[
  \lambda Pr_{1}(v_{2}) + (1 - \lambda) Pr_{1}(v_{3})
    = \lambda (x_{2}, 0, 0) + (1 - \lambda) (x_{3}, y_{3}, 0)\,\text{,}
\]
with $0 \le \lambda \le 1$.  Thus for a point
$p = (x_{p}, y_{p}, 0)$ on $Pr_{1}([v_{2}v_{3}])$, either
$x_{p} < 0$ and the point lies in $P_{1} \cap H_{2}$ or
both $x_{p} \ge 0$ and $y_{p} > 0$
so that $\langle p, v_{3}\rangle = x_{p}x_{3} + y_{p}y_{3} > 0$
and the point lies in $P_{1} \cap H_{3}$.  We conclude
that $Pr_{1}([v_{2}v_{3}]) \subset P_{1} \cap (H_{2} \cup H_{3})$.

Finally we note that there must exist a point
$p_{0} = (\varepsilon, y_{p_{0}}, 0)$ with $\varepsilon < 0$
such that $p_{0}$ lies on $Pr_{1}([v_{2}v_{3}])$ and
$p_{0} \in P_{1} \cap H_{2} \cap H_{3}$.  Thus we can
decompose $Pr_{1}([v_{2}v_{3}v_{j}])$ into the two pieces
$\cone{p_{0}}{Pr_{1}([v_{2}v_{j}])}$ and
$\cone{p_{0}}{Pr_{1}([v_{3}v_{j}])}$,
with the first piece lying in $P_{1} \cap H_{2}$
and the second piece lying in $P_{1} \cap H_{3}$.
It follows that $Pr_{1}([v_{2}v_{3}v_{j}])$
does not contain $u$.
\end{proof}

\bigskip

Recall that Euler's formula specifies a relationship between
the number of vertices, edges, and faces in a planar
graph.  If $m$, $e$, and $f$ are the number of vertices,
edges, and faces respectively, then Euler's formula states
that $m - e + f = 2$ for planar graphs.  In a
planar triangulation
each face is incident to three edges and each edge is
incident to two faces, so $2e = 3f$, and the relationship
$f = 2(m - 2)$ can be derived.  Moreover, in a
planar triangulation
each face is incident to three vertices, and vertex
$v_{i}$ is incident to $d(v_{i})$ faces, so
$\sum_{i} d(v_{i}) = 3f = 6(m - 2)$.  Combining these
consequences of Euler's formula with
Theorem~\ref{thm:nminus3}, we easily obtain a
lower bound on the number of edges incident to an
interior vertex in a $2$-well-centered tetrahedral mesh
in $\RR^{3}$.

\bigskip

\begin{corollary}
\label{cor:min2wc}
Let $M$ be a $2$-well-centered (or completely well-centered)
tetrahedral mesh in $\RR^{3}$.  For every
vertex $u$ interior to $M$, at least $9$ edges of $M$
are incident to $u$.
\end{corollary}
\begin{proof}
Let $G = \link~u$ for some interior vertex $u$ of a
$2$-well-centered mesh $M$, and let
$m$ be the number of edges incident
to $u$, i.e., the number of vertices of $G$.
Consider the possibility $m = 8$. Euler's formula shows
that for $m = 8$ we have $\sum_{i} d(v_{i}) = 36$ so the
average vertex degree is $4.5$, and there must be at least
one vertex of degree at least $5 = m - 3$.  By
Theorem~\ref{thm:nminus3}, this cannot occur, for such a graph
$G$ would not permit a $2$-well-centered neighborhood of $u$.
Similarly, if $m = 7$ the average degree is $30/7 > 4$
and there must be a vertex of degree at least $m - 2$.
For $m = 6$ the average degree is $4$ and there
must be a vertex of degree at least $m - 2$.  In each of the
cases $m = 5$ and $m = 4$, there is only one
triangulation, and this triangulation has a vertex
of degree $m - 1$.
\end{proof}

\bigskip

\begin{figure}
\centering
\begin{minipage}[c]{150pt}
\centering
\includegraphics[width = 150pt, trim = 524pt 237pt 522pt 222pt, clip]%
  {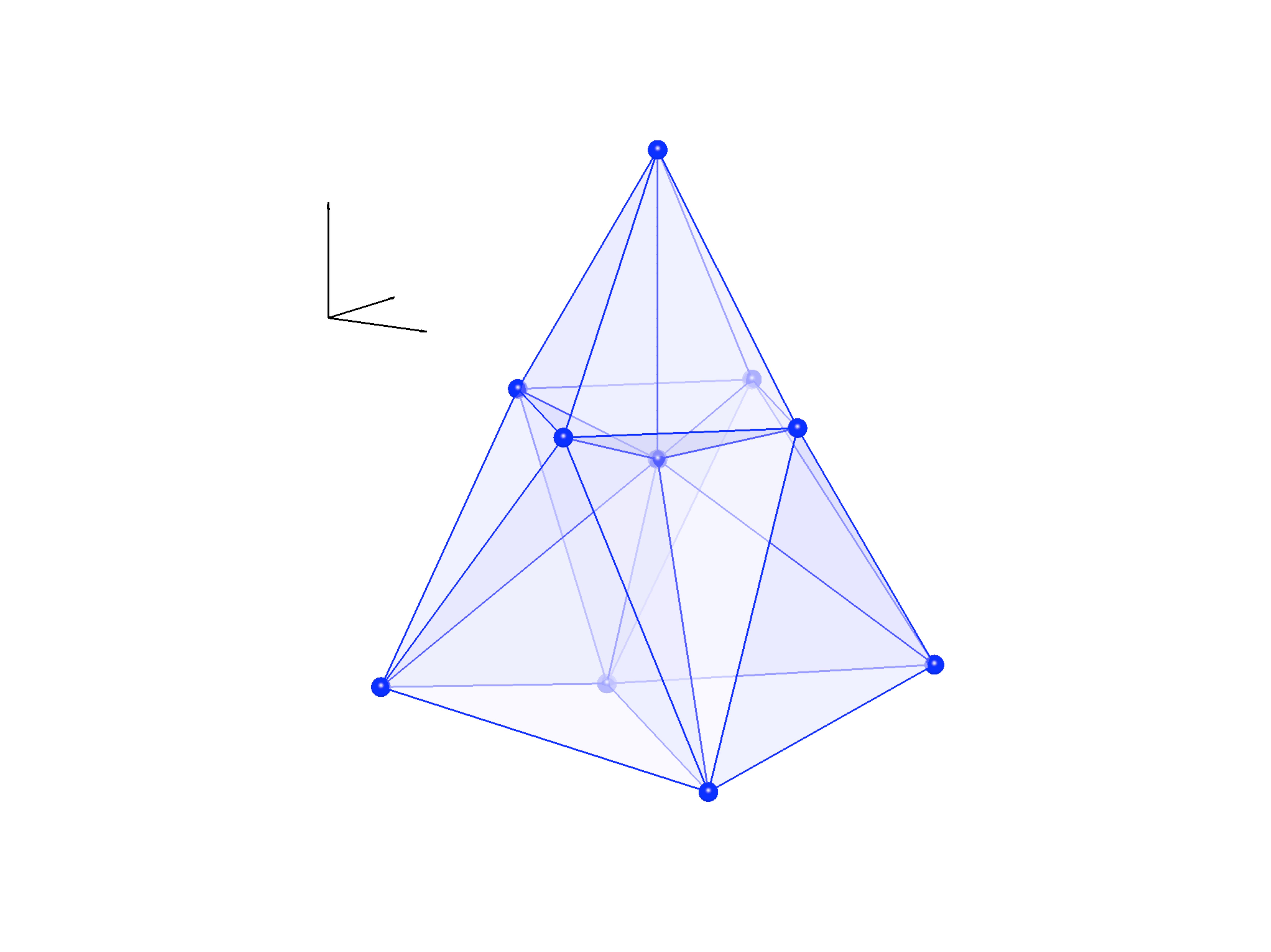}%
\begin{picture}(0, 0)
\put(-140, 110){\scriptsize{$x$}}
\put(-142, 124){\scriptsize{$y$}}
\put(-147, 132){\scriptsize{$z$}}
\end{picture}
\end{minipage}%
\hspace{50pt}%
\begin{minipage}[c]{153pt}
\settowidth{\dotlen}{.}
\begin{tabular}{| r @{.} l | r @{.} l | r @{.} l |}
\hline
\multicolumn{2}{|c|}{$x$}
    & \multicolumn{2}{c|}{$y$} & \multicolumn{2}{c|}{$z$}\\
\hline
\multicolumn{1}{|r @{\hspace{\dotlen}}}{$0$} & &
  \multicolumn{1}{r @{\hspace{\dotlen}}}{$0$} & &
  \multicolumn{1}{r @{\hspace{\dotlen}}}{$0$} & \\
\hline
\multicolumn{1}{|r @{\hspace{\dotlen}}}{$0$} & &
  \multicolumn{1}{r @{\hspace{\dotlen}}}{$0$} & &
  \multicolumn{1}{r @{\hspace{\dotlen}}}{$1$} & \\
\hline
\multicolumn{1}{|r @{\hspace{\dotlen}}}{$0$} &
    &   $0$ & $533$   &   $0$ & $164$\\
\hline
$0$ & $533$   &   \multicolumn{1}{r @{\hspace{\dotlen}}}{$0$} &
    &   $0$ & $164$\\
\hline
\multicolumn{1}{|r @{\hspace{\dotlen}}}{$0$} &
    &   $-0$ & $533$   &   $0$ & $164$\\
\hline
$-0$ & $533$   &   \multicolumn{1}{r @{\hspace{\dotlen}}}{$0$} &
    &   $0$ & $164$\\
\hline
$0$ & $63$   &   $0$ & $63$   &   $-0$ & $7$\\
\hline
$-0$ & $63$   &   $-0$ & $63$   &   $-0$ & $7$\\
\hline
$0$ & $594$   &   $-0$ & $594$   &   $-0$ & $9$\\
\hline
$-0$ & $594$   &   $0$ & $594$   &   $-0$ & $9$\\
\hline
\end{tabular}
\end{minipage}
\caption{A completely well-centered mesh with an interior vertex $u$
  such that $\link~u$ has nine vertices and
  degree list $(5,5,5,5,5,5,4,4,4)$.  The vertex coordinates
  are listed in the table at right; vertex $u$ is at the origin.}
\label{fig:onefreev_14/goodrndsym}
\end{figure}

When $m = 9$, the average degree is $4 \frac{2}{3}$, and
there is a triangulation of $S^{2}$ with degree list
$(5, 5, 5, 5, 5, 5, 4, 4, 4)$ that permits a completely
well-centered neighborhood.
Figure~\ref{fig:onefreev_14/goodrndsym}
shows a figure of a completely-well-centered mesh
that has a single interior vertex $u$
such that $\link~u$ is a $9$-vertex triangulation of
$S^{2}$ with the specified degree list.

We have already seen that there are infinitely
many triangulations of $S^{2}$ that can appear
as the link of an interior vertex in a
$2$-well-centered mesh (Proposition~\ref{prop:inftycwc}).
In the spirit of Proposition~\ref{prop:deg3vtx3wc},
we now discuss some ways to use an existing triangulation
that permits a $2$-well-centered neighborhood
to construct new triangulations that permit a
$2$-well-centered neighborhood.
The next two propositions show that one can add vertices
of degree $3$, subtract vertices of degree $3$,
or add vertices of degree $4$ to obtain new
triangulations that permit a $2$-well-centered
neighborhood.
In Proposition~\ref{prop:deg3} we again use the
notation $G - v_{1}$ used in Proposition~\ref{prop:deg3vtx3wc}.

\bigskip

\begin{proposition}
\label{prop:deg3}
A triangulation $G$ of $S^{2}$ that contains a vertex $v_{1}$
of degree three permits a $2$-well-centered neighborhood
if and only if the triangulation $G - v_{1}$ permits
a $2$-well-centered neighborhood.
\end{proposition}
\begin{proof}
First we suppose that $G$ permits a $2$-well-centered
neighborhood.  Then consider some tetrahedral mesh embedded
in $\RR^{3}$ that contains a vertex $u$ with
$\link~u$ isomorphic to $G$ and all
face angles $\measuredangle v_{i}uv_{j}$ acute.  We choose
a coordinate system on $\RR^{3}$ such that $u$
lies at the origin and identify each vertex $v_{i}$ of
$\link~u$ with the vector originating at the
origin and terminating at $v_{i}$.  Now
vector $v_{1}$ makes an acute face
angle for each of the three facets that are incident
to the edge $[uv_{1}]$.  Deleting $v_{1}$ from
$\link~u$ removes the three facets that are incident to
edge $[uv_{1}]$, but has no effect on the other facets
incident to $u$ or face angles at $u$.  Thus
all facets incident to $u$ remain acute after removing $v_{1}$,
and the modified neighborhood of $u$ is a mesh embedded in
$\RR^{3}$ that certifies that
$G - v_{1}$ permits a $2$-well-centered neighborhood.

On the other hand, if we suppose that $G - v_{1}$ permits a
$2$-well-centered neighborhood, we will be able to add vertex
$v_{1}$ and still have all face angles at $u$ acute.  Consider
some specific tetrahedral mesh embedded in $\RR^{3}$ containing a
vertex $u$ such that $\link~u$ is isomorphic to $G - v_{1}$.
Let $v_{2}, v_{3},$ and $v_{4}$ be the
three vertices of $\link~u$ that are adjacent to
$v_{1}$ in $G$.  Then the mesh contains facets
$[v_{2}uv_{3}]$, $[v_{3}uv_{4}]$, and $[v_{4}uv_{2}]$.
Moreover, since the face angles at $u$
are acute, we have
$\langle v_{i}, v_{j}\rangle > 0$ for each
$(i, j) \in \{2,3,4\}\times\{2,3,4\}.$  It follows that if
we insert vertex $v_{1}$ satisfying
$v_{1} = \lambda_{2}v_{2} + \lambda_{3}v_{3}
  + \lambda_{4}v_{4}$
with each $\lambda_{i} > 0,$ then for $i = 2, 3, 4$ we have
\[
\langle v_{1}, v_{i}\rangle =
    \lambda_{2}\langle v_{2}, v_{i}\rangle
  + \lambda_{3}\langle v_{3}, v_{i}\rangle
  + \lambda_{4}\langle v_{4}, v_{i}\rangle  > 0.
\]
In other words, as long as $v_{1}$ lies interior to the cone at $u$
bounded by vectors $v_{2}$, $v_{3}$, and $v_{4}$,
it will make acute face angles with each of $v_{2}$,
$v_{3},$ and $v_{4}$.
\end{proof}

\bigskip

Notice that Proposition~\ref{prop:deg3} also implies
that adding or deleting a degree three vertex from a
triangulation of $S^{2}$ that does not permit a $2$-well-centered
neighborhood creates another triangulation of $S^{2}$ that
does not permit a $2$-well-centered neighborhood.  In
particular, this means that Theorem~\ref{thm:nminus3} does not
characterize the triangulations of $S^{2}$ that cannot
appear as the link of a vertex in a $2$-well-centered
tetrahedral mesh in $\RR^{3}$.

In the next proposition, we consider the case of
a triangulation $G$ of $S^{2}$ with a vertex $v_{1}$
such that $d(v_{1}) = 4$.  To talk about removing
vertex $v_{1}$ from $G$ in this
case, we need to specify an edge to add after removing
the vertex.  Let $v_{2}$, $v_{3}$, $v_{4}$, and $v_{5}$
be the neighbors of $v_{1}$, listing in cyclic order.
Then $(G - v_{1}) \cup [v_{2}v_{4}]$ is the triangulation
of $S^{2}$ obtained from $G$ by removing vertex $v_{1}$
along with the four edges and triangles incident to $v_{1}$
and adding the edge $[v_{2}v_{4}]$ along with the two
triangles $[v_{2}v_{3}v_{4}]$ and $[v_{2}v_{4}v_{5}]$.

\bigskip

\begin{proposition}
\label{prop:deg4}
Consider a triangulation $G$ of $S^{2}$ that contains a vertex $v_{1}$
of degree four with neighbors $v_{2}$, $v_{3}$, $v_{4}$, $v_{5}$
(listed in clockwise order).
If $(G - v_{1}) \cup [v_{2}v_{4}]$ or $(G - v_{1}) \cup [v_{3}v_{5}]$
permits a  $2$-well-centered neighborhood, then $G$
permits a $2$-well-centered neighborhood.
\end{proposition}
\begin{proof}
Suppose without loss of generality that
$(G - v_{1}) \cup [v_{2}v_{4}]$ permits a
$2$-well-centered neighborhood.  Let $u$ be a vertex
for which $\link~u$ is isomorphic to $(G - v_{1}) \cup [v_{2}v_{4}]$
and consider some embedding of $\cone{u}{\link~u}$ into
$\RR^{3}$ such that all face angles at $u$ are acute.
We choose a coordinate system on $\RR^{3}$ such that $u$
lies at the origin and identify each vertex of
$\link~u$ with the vector originating at the
origin $u$ and terminating at the vertex.  We know that
$\langle v_{2}, v_{3}\rangle > 0,$
$\langle v_{3}, v_{4}\rangle > 0,$ $\langle v_{4}, v_{5}\rangle > 0,$
$\langle v_{5}, v_{2}\rangle > 0,$ and $\langle v_{2}, v_{4}\rangle > 0,$
because each pair of vectors bounds a face with an acute
face angle at $u$.

Now let $v_{1} = (v_{2} + v_{4})/2$ and, deleting the facet
$[v_{2}uv_{4}]$, add the four facets $[v_{1}uv_{i}]$
for $i = 2,3,4,5$.
The new facets $[v_{1}uv_{2}]$ and $[v_{1}uv_{4}]$ have face angles
at $u$ that are smaller than the face angle
$\measuredangle v_{2}uv_{4}$ was, so they are acute. The facets
$[v_{1}uv_{3}]$ and $[v_{1}uv_{5}]$ also have acute
face angles at $u$ because
\[
\langle v_{1}, v_{3}\rangle = \frac{1}{2}\langle v_{2}, v_{3}\rangle
   + \frac{1}{2}\langle v_{4}, v_{3}\rangle > 0
\]
and similarly
\[
\langle v_{1}, v_{5}\rangle = \frac{1}{2}\langle v_{2}, v_{5}\rangle
   + \frac{1}{2}\langle v_{4}, v_{3}\rangle > 0.
\]
We see that adding $v_{1} = (v_{2} + v_{4})/2$ has created
a new mesh for which all face angles at $u$ are acute.  Thus
$G$ permits a $2$-well-centered neighborhood.
\end{proof}

\bigskip

\section{Applications to the Cube}
\label{sec:cubeapps}

The theoretical results presented in this paper are useful
for creating well-centered meshes of specific
regions in $\RR^{3}$.  In particular, one might
design a tetrahedral mesh of a volume so that it meets all
of the combinatorial conditions discussed in
Secs.~\ref{sec:combinatorial3wccond}
and~\ref{sec:combinatorial2wctetcond}.  Then applying
the optimization procedure discussed in~\cite{VaHiGuRa2009},
one may hope to obtain a well-centered mesh of the domain.
This technique was successfully used in \cite{VaHiGu2008} to create
well-centered meshes of several domains in $\RR^{3}$,
including the cube.

The theory developed in this paper has several
obvious implications for the combinatorial properties of
a well-centered triangulation of the cube.  For example,
no cube corner tetrahedron, e.g., the tetrahedron
shown in Fig.~\ref{fig:badcubecrnr}, can be
$3$-well-centered;  considering the bottom
facet to be a given facet, we see that the fourth
vertex of the tetrahedron projects onto (not inside) the
circumcircle of the given facet, violating the necessary
Cylinder Condition of Proposition~\ref{prop:char_nec}.

It follows that in a $3$-well-centered mesh of the
cube there must be at least two tetrahedra incident
to each corner of the cube.  Indeed, there must be
at least three tetrahedra incident to each corner
of the cube in a $3$-well-centered mesh.
In the case of two tetrahedra incident
to a corner vertex there must be exactly four edges
incident to the corner vertex, of which three are in
the directions of the coordinate axes.  The fourth
edge must lie in a face, and both tetrahedra are
incident to the axis orthogonal to the face containing
the fourth edge.  The Cylinder Condition applies again,
and we see that the mesh cannot be $3$-well-centered.

\begin{figure}
\centering
\begin{minipage}[b]{180pt}
\centering
\includegraphics[width=120pt, trim=247pt 338pt 242pt 333pt, clip]%
  {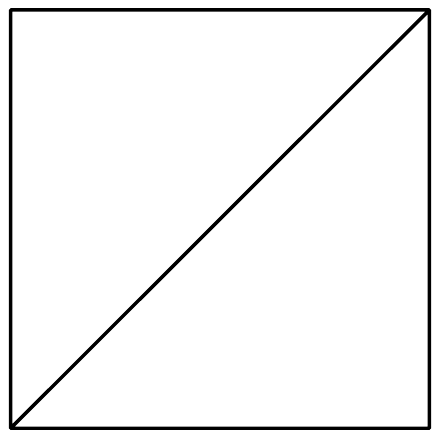}
\end{minipage}%
\hspace{20pt}%
\begin{minipage}[b]{180pt}
\centering
\includegraphics[width=80pt, trim=178pt 268pt 279pt 253pt, clip]%
  {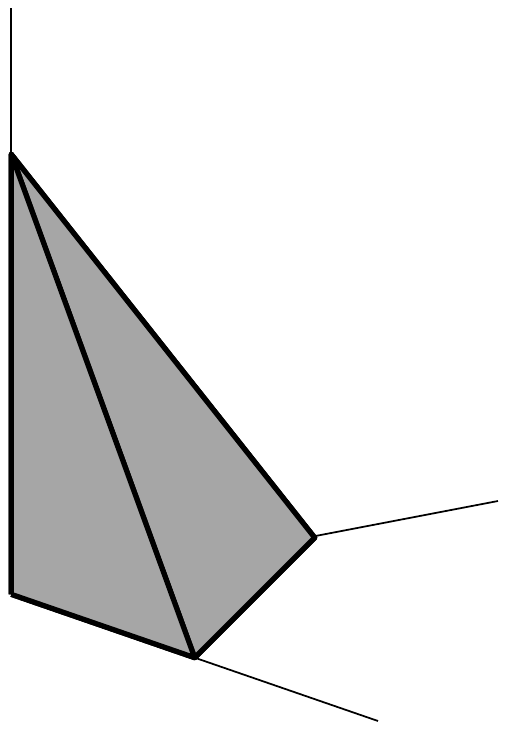}
\end{minipage}\\[-10pt]
\begin{minipage}[t]{180pt}
\caption{No $3$-well-centered mesh of the cube has a face
  with this triangulation.}
\label{fig:badcubeface}
\end{minipage}%
\hspace{20pt}%
\begin{minipage}[t]{180pt}
\caption{No $2$-well-centered or $3$-well-centered mesh of the solid cube
  has a tetrahedron with three of its facets lying in faces of the cube.}
\label{fig:badcubecrnr}
\end{minipage}
\end{figure}

Ad hoc arguments from basic Euclidean geometry
provide more restrictions on
well-centered triangulations of the solid cube.  For instance,
in any $3$-well-centered mesh of the cube, no face of
the cube is triangulated as shown in Fig.~\ref{fig:badcubeface},
with two right triangles meeting along the hypotenuse.
The two right triangles have the same circumcenter, which lies at
the midpoint of the common hypotenuse of the triangles ---
the center of the face of the cube.  For either triangle,
a tetrahedron having that triangle as a facet must have its
circumcenter on a line $\ell$ perpendicular to the face of the cube
that meets the cube face at its center.  Considering two
tetrahedra $\sigma_{1}$ and $\sigma_{2}$, each having one
of the right triangles as a face, it can be shown that
at most one of $\sigma_{1}, \sigma_{2}$ can be
$3$-well-centered.  There is a plane that contains the
hypotenuse of the right triangles and divides $\RR^3$
into two open halfspaces $H_{1}$ and $H_{2}$ such that
$\sigma_{1} \in \closure{H_{1}}$ and
$\sigma_{2} \in \closure{H_{2}}$.  If $\sigma_{1}$
and $\sigma_{2}$ share a common face, the plane must be
$\affine(\sigma_{1} \cap \sigma_{2})$, but otherwise
there is some flexibility in the choice of the plane.
The portion of $\ell$ interior
to the cube is either in the boundary between $H_{1}$ and $H_{2}$
or without loss of generality can be assumed
to lie entirely in $H_{1}$.
In either case, the circumcenter of $\sigma_{2}$ is not
strictly interior to $\sigma_{2}$,
so $\sigma_{2}$ is not $3$-well-centered.

The paper \cite{VaHiGu2008}, along with discussing
some well-centered meshes of the cube, raises the
question of the minimum number of tetrahedra
needed to create a well-centered triangulation of the
cube.  We can use the statements above to derive some
simple lower bounds on the number of tetrahedra in
a well-centered mesh of the $3$-cube.  In a triangulation
of the cube, the number of tetrahedra incident to the
surface of the cube is a lower bound on the total number
of tetrahedra, so one can obtain a lower bound on the
number of tetrahedra by counting the number of
triangular facets in a surface triangulation.  The number
of facets is not a direct lower bound, since there may
be a single tetrahedron with multiple facets in the
surface of the cube.  Because there are at least three
distinct tetrahedra incident to a cube corner in a
$3$-well-centered triangulation of the cube, a tetrahedron
cannot be counted more than twice in counting the
number of surface facets of a $3$-well-centered triangulation
of the cube.  The same holds true for $2$-well-centered
triangulations of the cube, since three of the facets of
a cube corner tetrahedron are right triangles.

Noting, then, that each face of the cube must contain at
least $3$ triangles in a $3$-well-centered mesh of the cube
and at least $8$ triangles in a $2$-well-centered mesh of the
cube, we easily obtain a lower bound of $9$ tetrahedra
for a $3$-well-centered triangulation of the cube, and
$24$ tetrahedra for a $2$-well-centered triangulation of the cube.
(These lower bounds are mentioned in~\cite{VaHiGu2008}
without the details of the geometric or combinatorial
arguments.)  It should be possible to improve both of these
bounds, but these relatively simple bounds help demonstrate
a possible application of this paper's theory and
are a starting place for a more careful analysis.

\section{Conclusions}
\label{sec:conclusion}
In this paper we introduced several geometric propositions
related to $n$-well-centered simplices and gave an algebraic
characterization of an $n$-well-centered simplex in terms
of cubic polynomial inequalities.  We applied the geometric
propositions to the study of the combinatorial properties
of well-centered meshes, especially well-centered tetrahedral
meshes.

We considered triangulations of topological $S^{2}$ and showed
that the set of such triangulations that cannot appear as the link
of a vertex in a $3$-well-centered (or $2$-well-centered
or completely well-centered) tetrahedral mesh
embedded in $\RR^{3}$ is an infinite set, contrasting
this to the analogous question for triangle meshes in $\RR^{2}$.
We showed also that the set of triangulations of $S^{2}$ that
do appear as the link of a vertex in some completely
well-centered (or $3$-well-centered or $2$-well-centered)
tetrahedral mesh embedded in $\RR^{3}$ is an infinite set.
We proved several results in the direction of classifying
which triangulations of $S^{2}$ can appear as the link of
a vertex in a $2$-well-centered or $3$-well-centered tetrahedral
mesh embedded in $\RR^{3}$.

The work on combinatorial properties of well-centered meshes
leads to some interesting open questions.  Is there a compact
way to express a complete characterization of which triangulations
of $S^{2}$ can appear in a $3$-well-centered (or $2$-well-centered
or completely well-centered) mesh embedded in $\RR^{3}$?
Is the necessary condition described in
Theorem~\ref{thm:oneringcond} a complete characterization
for vertex links in $3$-well-centered tetrahedral meshes
in $\RR^{3}$?  If a triangulation of $S^{2}$ permits a
$2$-well-centered neighborhood in the sense defined in
this paper, does this imply that it can appear as the
link of a vertex in a $2$-well-centered tetrahedral mesh
embedded in $\RR^{3}$?  If a triangulation of $S^{2}$
can appear as the link of a vertex in both a $2$-well-centered
tetrahedral mesh in $\RR^{3}$ and a $3$-well-centered tetrahedral
mesh in $\RR^{3}$, does this guarantee that it can appear
as the link of a vertex in a completely well-centered mesh?

Beyond the questions about tetrahedral meshes there are
questions about higher dimensions.  Is it possible to
extend the results of Sec.~\ref{sec:combinatorial2wctetcond}
to say something about $2$-well-centered meshes in
higher dimensions?  Certainly Lemmas~\ref{lemma:orthogonalplane},
\ref{lemma:2wcvtxproj}, and~\ref{lemma:facetintrsct} can be
generalized to higher dimensions.  Is it the case that
for each $n$ there exists a completely well-centered
$n$-simplicial neighborhood of a vertex embedded in $\RR^{n}$?
If so, Lemma~\ref{lemma:smplxonsphere} may provide a way
of constructing such neighborhoods.  We note, however,
that for $n \ge 5$ there is no dihedral acute $n$-simplicial
neighborhood of a vertex embedded in
$\RR^{n}$~\cite{KoPaPr2009}~\cite{Krizek2006}.

This work also leads to some interesting practical
questions about creating well-centered tetrahedral meshes
in $\RR^{3}$.  In particular, now that there is some
understanding of which triangulations can appear as the link
of a vertex in a well-centered mesh, we hope that practical
methods can be developed for improving the local mesh
connectivity of tetrahedral meshes.  An algorithm for
improving mesh connectivity of triangle meshes in $\RR^{2}$
appears in~\cite{VaHiGuRa2007}, but it is not obvious how to
formulate that type of algorithm for tetrahedral meshes
in $\RR^{3}$.  It is also worth noting that a triangulation
of $S^{2}$ that can theoretically appear as the link of a
vertex in a well-centered mesh might be a poor neighborhood
for a vertex in a practical setting.  For instance, in a
triangle mesh in $\RR^{2}$, a vertex link with $100$ vertices
can appear in a $2$-well-centered triangle mesh embedded
in $\RR^{2}$, but a mesh containing such a vertex link would
be considered poor quality in almost any application.  Is
there a good way to rate vertex links in tetrahedral
meshes according to their applicability in a practical setting?

We hope that our results will motivate others to
investigate these interesting and important questions.

\section*{Acknowledgement}
The authors thank Doug West for a helpful discussion.  The
work of Anil N. Hirani and Evan VanderZee was supported by
an NSF CAREER Award (Grant No. DMS-0645604).  Evan VanderZee
was also partially supported by a fellowship jointly funded
by the Computational Science and Engineering Program and the
Applied Mathematics Program of the University of Illinois
at Urbana-Champaign.  Vadim Zharnitsky was partially supported
by NSF grant DMS 08-07897.

\bibliographystyle{acmurldoi}
\bibliography{wct}

\end{document}

%% file: figs/wct/examples/wctconditions/smplxrflct3dprta.pdf_t
\begin{picture}(0,0)%
\includegraphics{figs/wct/examples/wctconditions/smplxrflct3dprta.pdf}%
\end{picture}%
\setlength{\unitlength}{3947sp}%
\begingroup\makeatletter\ifx\SetFigFont\undefined%
\gdef\SetFigFont#1#2#3#4#5{%
  \reset@font\fontsize{#1}{#2pt}%
  \fontfamily{#3}\fontseries{#4}\fontshape{#5}%
  \selectfont}%
\fi\endgroup%
\begin{picture}(8234,6349)(2529,-6155)
\put(6481,-6001){\makebox(0,0)[lb]{\smash{{\SetFigFont{29}{34.8}{\familydefault}{\mddefault}{\updefault}{\color[rgb]{0,0,0}$C$}%
}}}}
\put(2881,-5257){\makebox(0,0)[lb]{\smash{{\SetFigFont{29}{34.8}{\familydefault}{\mddefault}{\updefault}{\color[rgb]{0,0,0}$\affine(\tau)$}%
}}}}
\put(5476,-1561){\makebox(0,0)[lb]{\smash{{\SetFigFont{29}{34.8}{\familydefault}{\mddefault}{\updefault}{\color[rgb]{0,0,0}$\tau^{\prime}$}%
}}}}
\put(2977,-2113){\makebox(0,0)[lb]{\smash{{\SetFigFont{29}{34.8}{\familydefault}{\mddefault}{\updefault}{\color[rgb]{0,0,0}$\affine(\tau^{\prime})$}%
}}}}
\put(7561,-4633){\makebox(0,0)[lb]{\smash{{\SetFigFont{29}{34.8}{\familydefault}{\mddefault}{\updefault}{\color[rgb]{0,0,0}$\tau$}%
}}}}
\put(6781,-3061){\makebox(0,0)[lb]{\smash{{\SetFigFont{29}{34.8}{\familydefault}{\mddefault}{\updefault}{\color[rgb]{0,0,0}$c(\sigma)$}%
}}}}
\put(6481,-241){\makebox(0,0)[lb]{\smash{{\SetFigFont{29}{34.8}{\familydefault}{\mddefault}{\updefault}{\color[rgb]{0,0,0}$C^{\prime}$}%
}}}}
\end{picture}%

%% file: figs/wct/examples/wctconditions/smplxrflct3dprtb.pdf_t
\begin{picture}(0,0)%
\includegraphics{figs/wct/examples/wctconditions/smplxrflct3dprtb.pdf}%
\end{picture}%
\setlength{\unitlength}{3947sp}%
\begingroup\makeatletter\ifx\SetFigFont\undefined%
\gdef\SetFigFont#1#2#3#4#5{%
  \reset@font\fontsize{#1}{#2pt}%
  \fontfamily{#3}\fontseries{#4}\fontshape{#5}%
  \selectfont}%
\fi\endgroup%
\begin{picture}(5132,6349)(751,-6155)
\put(2857,-4501){\makebox(0,0)[lb]{\smash{{\SetFigFont{29}{34.8}{\familydefault}{\mddefault}{\updefault}{\color[rgb]{0,0,0}$\tau$}%
}}}}
\put(3241,-241){\makebox(0,0)[lb]{\smash{{\SetFigFont{29}{34.8}{\familydefault}{\mddefault}{\updefault}{\color[rgb]{0,0,0}$C^{\prime}$}%
}}}}
\put(3241,-6001){\makebox(0,0)[lb]{\smash{{\SetFigFont{29}{34.8}{\familydefault}{\mddefault}{\updefault}{\color[rgb]{0,0,0}$C$}%
}}}}
\put(3421,-1861){\makebox(0,0)[lb]{\smash{{\SetFigFont{29}{34.8}{\familydefault}{\mddefault}{\updefault}{\color[rgb]{0,0,0}$\tau^{\prime}$}%
}}}}
\put(3565,-3109){\makebox(0,0)[lb]{\smash{{\SetFigFont{29}{34.8}{\familydefault}{\mddefault}{\updefault}{\color[rgb]{0,0,0}$c(\sigma)$}%
}}}}
\end{picture}%